\documentclass[10pt, reqno]{amsart}
\usepackage{amsaddr}
\usepackage[a4paper, total={6.5in, 8.4in}]{geometry}

\usepackage[symbol]{footmisc}

\usepackage{amsmath}
\usepackage{amssymb}
\usepackage{amsthm}
\usepackage[latin1]{inputenc}
\usepackage{eurosym}
\usepackage{graphicx}
\usepackage{epsfig}
\usepackage{hyperref}
\usepackage{dsfont}
\usepackage[space]{cite}
\usepackage{subcaption}
\usepackage{caption}
\usepackage{placeins}
\usepackage{mathtools}
\usepackage{xcolor}
\usepackage[ruled,vlined]{algorithm2e}
\usepackage[title]{appendix}

\usepackage[displaymath,mathlines]{lineno}
\modulolinenumbers[1]

\allowdisplaybreaks

\usepackage{hyperref}

\usepackage{ifthen}
\usepackage{comment}

\makeindex
\usepackage{color}
\usepackage{float}
\usepackage{commath}

\renewcommand{\div}{\operatorname{div}}

\def\leq{\leqslant}

\numberwithin{equation}{section}

\newtheoremstyle{thmlemcorr}{10pt}{10pt}{\itshape}{}{\bfseries}{.}{10pt}{{\thmname{#1}\thmnumber{
			#2}\thmnote{ (#3)}}}
\newtheoremstyle{thmlemcorr*}{10pt}{10pt}{\itshape}{}{\bfseries}{.}\newline{{\thmname{#1}\thmnumber{
			#2}\thmnote{ (#3)}}}
\newtheoremstyle{defi}{10pt}{10pt}{\itshape}{}{\bfseries}{.}{10pt}{{\thmname{#1}\thmnumber{
			#2}\thmnote{ (#3)}}}
\newtheoremstyle{remexample}{10pt}{10pt}{}{}{\bfseries}{.}{10pt}{{\thmname{#1}\thmnumber{
			#2}\thmnote{ (#3)}}}
\newtheoremstyle{ass}{10pt}{10pt}{}{}{\bfseries}{.}{10pt}{{\thmname{#1}\thmnumber{
			A#2}\thmnote{ (#3)}}}

\theoremstyle{thmlemcorr}
\newtheorem{theorem}{Theorem}
\numberwithin{theorem}{section}

\theoremstyle{thmlemcorr*}
\newtheorem{theorem*}{Theorem}
\newtheorem{lemma*}[theorem]{Lemma}
\newtheorem{corollary*}[theorem]{Corollary}
\newtheorem{proposition*}[theorem]{Proposition}
\newtheorem{problem*}[theorem]{Problem}
\newtheorem{conjecture*}[theorem]{Conjecture}

\theoremstyle{defi}

\newtheorem{hyp}{Assumption}

\newtheorem{problem}{Problem}

\theoremstyle{remexample}
\newtheorem{remark}[theorem]{Remark}

\theoremstyle{ass}

\newtheorem*{notations*}{Notations}
\newtheorem*{related_works*}{Notations}
\newtheorem*{acknowledgement*}{Acknowledgement}
\allowdisplaybreaks

\title[Decoding MFGs from Population and Environment Observations By GPs]{Decoding Mean Field Games from Population and Environment Observations By  Gaussian Processes}

\author{Jinyan Guo$^1$, Chenchen Mou$^2$, Xianjin Yang$^{3,*}$, Chao Zhou$^1$}

\address[Y. Guo]{
	$^1$Department of Mathematics and Risk Management Institute, National University of Singapore, Singapore.}
\email{e0708165@u.nus.edu}

\address[C. Mou]{
	$^2$Department of Mathematics, City University of Hong Kong, Hong Kong SAR, China.}
\email{chencmou@cityu.edu.hk}

\thanks{$^*$Corresponding author.}
\address{$^3$Department of Computing and Mathematical Sciences, California Institute of Technology, Pasadena, CA 91125, USA.}
\email{yxjmath@caltech.edu}

\email{matzc@nus.edu.sg}

\begin{document}
\maketitle

\begin{abstract}
This paper presents a Gaussian Process (GP) framework, a non-parametric technique widely acknowledged for regression and classification tasks, to address inverse problems in mean field games (MFGs). By leveraging GPs, we aim to  recover agents' strategic actions and the environment's configurations from partial and noisy observations of the population of agents and the setup of the environment. Our method is a probabilistic tool to infer the behaviors of agents in MFGs from data in scenarios where the comprehensive dataset is either inaccessible or contaminated by noises.

\end{abstract}

\section{Introduction}
In this paper, we present a Gaussian Process (GP) framework to infer strategies and environment configurations in mean field games (MFGs) from partially noisy observations of agent populations and environmental setups. 
MFGs, as proposed in \cite{lasry2006jeux1, lasry2006jeux, lasry2007mean, huang2006large, huang2007large, huang2007nash, huang2007invariance}, study the actions of a large number of indistinguishable rational agents. As the number of agents approaches infinity, the Nash equilibrium of a standard MFG is described by a coupled system of two distinct partial differential equations (PDEs): the Hamilton--Jacobi--Bellman (HJB) and the Fokker--Plank (FP) equations. The HJB determines the agent's value function, while the FP governs the evolution of the distribution of agents. We refer the readers to \cite{gueant2011mean, gomes2015economic, lee2020mean, gao2021belief, gomes2021mean, evangelista2018first, lee2021mean} for applications of MFGs in various fields.

To date, the consistency and completeness of MFGs have been explored thoroughly across different contexts. The earliest findings can be traced back to the foundational studies by Lasry and Lions, and further expanded upon in Lions' lectures at Coll\`ege de France (referenced by \cite{lionslec}). Notably, only a few MFG models provide explicit solutions, highlighting the significance of numerical calculations in MFGs. 

A typical time-dependent MFG admits the following form
\begin{align}
\label{intro_MFG}
\begin{cases}
-\partial_t u(t, x) - \nu \Delta u+ H(x, Du) = F[m](x), &\forall (t, x)\in (0, T)\times \mathbb{R}^d,\\
\partial_t m(t, x) - \nu \Delta m - \div(m D_pH(x, Du)) = 0, &\forall (t, x)\in (0, T)\times \mathbb{R}^d,\\
m(0, x) = \mu(x), u(T, x) = g(x), &\forall x\in \mathbb{R}^d,
\end{cases}
\end{align}
where $T$ is the terminal time, $u$ stands for a representative agent's value function, \( m \) is the probability density function representing the distribution of agents, \( \nu \) denotes the volatility of agents' movements, \( H \) is referred as the Hamiltonian, \( F \) is the mean field term  describing the effect of agents' interactions, \( \mu \) is the initial distribution, and \( g \) represents the terminal cost. If agents play a MFG in \eqref{intro_MFG}, the optimal strategy of agents is given by $-D_pH(x, Du)$ \cite{lasry2006jeux1, lasry2006jeux, lasry2007mean}.

In applications, MFGs have emerged as a prominent framework for modeling and analyzing the strategic interactions of a large number of agents. Traditionally, a thorough understanding of these interactions requires comprehensive data on individual agent strategies, population dynamics, and the intricacies of the environment in which these agents operate. However, in many real-world scenarios, our observations are often partial, noisy, or even sporadic. The challenge then is to distill meaningful information and insights from such limited data.

In this research, we focus on deducing agents' strategies and environmental information by observing the agents' distribution and leveraging limited data about the surrounding environment. Specifically, we are concerned with the following problem.
\begin{problem}
\label{main_prob}
Suppose that agents are engaged in a MFG within a competitive environment (For instance, suppose that agents play a MFG specified in \eqref{intro_MFG}). Based on limited and noisy observations of the agents' population ($m$ in \eqref{intro_MFG}) and partial observations on the environment ($\nu, H, F$, and $g$ in \eqref{intro_MFG}), we infer the complete population (values of $m$), the optimal strategy employed by the agents ($-D_pH(x, Du)$ in \eqref{intro_MFG}), and the setup of the environment ($\nu$, $H$, $F$, and $g$ in \eqref{intro_MFG}). 
\end{problem}
We further explore Problem \ref{main_prob} by delving into the realms of both general stationary and time-dependent MFGs, as detailed in Sections \ref{secStationary} and \ref{secTimDep} respectively. Stationary MFGs' inverse problems are of interest due to their representation as asymptotic limits of time-dependent MFGs \cite{
cannarsa2020long, cirant2021long, cardaliaguet2012long}. Agents typically reach a steady state rapidly under appropriate conditions \cite{
cannarsa2020long, cirant2021long, cardaliaguet2012long}. Therefore, in practical scenarios, observed MFG data is likely in this steady state, emphasizing the need to recover strategies and environmental information at equilibrium.

Inverse problems in MFGs often suffer from non-uniqueness, implying multiple PDE profiles can fit the same data. Particularly, observing only the distribution of agents rarely yields a unique profile for quantities \((m, u, \nu, H, F, \text{ and } g)\) as defined in Problem \ref{main_prob}. To motivate the solvability of Problem \ref{main_prob} and highlight the necessity of observations on the environment,  we consider a scenario where agents engage in a MFG described in \eqref{intro_MFG}, with a Hamiltonian expressed as $H(x, Du)=\frac{|Du(x)|^2}{2}+V(x)$. In this setting, the optimal strategy for a typical agent is represented by $-Du$ \cite{lasry2006jeux1, lasry2006jeux, lasry2007mean}.  First, we consider the case when we know the viscosity $\nu=1$. If we observe the aggregate population, i.e., we have access to the population density $m$, we can uniquely determine the velocity field $-Du$ across the regions where $m$ is supported. This is evident from the FP equation in \eqref{intro_MFG}. To illustrate, let $u_1$ be a solution to \eqref{intro_MFG}. Assume there exists another value function, $u_2$, fulfilling the condition $\partial_t m(t, x) - \Delta m -\div(mDu_2)=0$. By subtracting this equation from the FP equation in \eqref{intro_MFG}, we obtain $\div(m(Du_1 - Du_2))=0$. Further, upon multiplying this resultant equation by $u_1 - u_2$ and applying integration by parts, it follows that $\int m|Du_1 - Du_2|^2\dif x=0$. Consequently, we deduce that $Du_1 = Du_2$ almost everywhere with respect to the measure $m$. Therefore, even without explicit knowledge of the potential function $V$ in the Hamiltonian or the specific interactions among agents, as described by $F$, it is possible to ascertain the strategies $-Du$ of agents solely by observing the population's distribution. However, when the viscosity $\nu$ is not known, it is impossible to recover a unique pair $(\nu, -Du)$ just only from the Fokker--Plank equation. Indeed, given $m$, if $(\nu, u)$ solve $\partial_t m(t, x) - \nu\Delta m -\div(mDu)=0$, $(\nu + C, u - C\log m)$ is also a solution for any constant $C$. Therefore, to uniquely determine the optimal strategy of agents, in addition to the knowledge of the density  $m$,  we need to know either the viscosity $\nu$ or some information about the Hamiltonian and the coupling function $F$ in the HJB equation in \eqref{intro_MFG}. 
Thus, without any information about the environment characterized by $(\nu, H, F, g)$,  it is unlikely to determine the optimal strategies of agents uniquely. 

Problem \ref{main_prob} expands upon the scenarios outlined in the preceding paragraph. It delves into the recovery of the strategy, population, and environment through partial and noisy observations of both the population and the environment.  We tackle Problem \ref{main_prob} through Gaussian Processes (GPs). a powerful non-parametric statistical method renowned for its ability to model and predict complex phenomena with a degree of uncertainty. The beauty of GPs lies in their inherent flexibility, allowing them to capture a wide range of behaviors and patterns without being tied down by a strict parametric form. This makes them an appealing choice for tackling the multifaceted challenges posed by MFGs. We provide a general GP framework which determines the unknown variables using Maximum A Posteriori (MAP) estimators. These estimators are conditioned on the fulfillment of PDEs at designated sample points. Additionally, the process incorporates regularization through the inclusion of observations that are subject to noises.

In Section  \ref{secNumercialExpe}, we demonstrate the effectiveness of our methodology through a variety of examples. Our tests focus on reconstructing comprehensive profiles of populations, strategies of agents, and complete environmental setups from partial and noisy data about populations and environments. In particular, Subsection \ref{sub:num:smfg0}  details the calculation of the viscosity constant, an essential parameter for gauging uncertainty in dynamic systems.  Subsections \ref{sub:num:smfguk} and \ref{sub:num:smfgnlc}  are devoted to identifying the coupling functions in MFG, which reflect the interactions between agents. In Subsection \ref{sub:num:mspfcd}, we address a misspecification problem. This section demonstrates the robustness and effectiveness of our methodology by recovering agents' information using a misspecified model. The data originates from a MFG characterized by local interactions, yet our model employs a non-local coupling approach for this reconstruction. This contrast underscores the adaptability and precision of our technique in handling complex data scenarios. Finally, in Subsection \ref{sub:num:tdmfg}, we apply our methods to a time-dependent MFG, highlighting its flexibility in time-dependent environments.

\subsection{Related Works in MFGs}
Although numerous numerical algorithms exist for solving MFGs (see \cite{achdou2010mean, achdou2012iterative, nurbekyan2019fourier, liu2021computational, liu2020splitting, briceno2018proximal, briceno2019implementation, carmona2021convergence, carmona2019convergence, ruthotto2020machine, lin2020apac, lauriere2021numerical, achdou2020mean, gomes2020hessian,  mou2022numerical, meng2023sparse}), there is a notably limited amount of research addressing the inverse problems in MFGs \cite{liu2023inverse, ding2022mean, chow2022numerical, klibanov2023convexification, klibanov2023h, imanuvilov2023lipschitz, liu2023simultaneously, ding2023determining}. To the best of our knowledge, only three numerical papers have been published that specifically address the inverse problems in MFGs. The work \cite{ding2022mean} introduces models to reconstruct ground metrics and interaction kernels in MFG running costs. Efficiency and robustness of the models are demonstrated through numerical simulations. The paper \cite{chow2022numerical} presents a numerical algorithm to solve an inverse problem in MFGs using partial boundary measurements. Authors in \cite{klibanov2023convexification} propose a globally convergent convexification numerical method for recovering the global interaction term from the single measurement data.  In a recent study, \cite{yang2023context} tackles mean field control inverse problems through operator learning, introducing a novel algorithm designed for broader applications in operator learning. Their approach, as detailed in \cite{yang2023context}, involves training models on pairs of input and output data, all calculated under the same parameters. Contrasting with this method, our paper presents a distinct approach where we estimate the unknowns by analyzing the data in a single observation. 

In contrast to previous research on inverse problems in MFGs, the GP framework introduced in this paper offers a comprehensive approach to simultaneously infer all unknown variables using only noisy, partial observations of the population and environment. Notably, to the best of our knowledge, this is the first study to estimate the viscosity parameter $\nu$ in MFGs. This parameter is vital as it reflects the uncertainty in agents' dynamics.

\subsection{Related Works in GPs}
GP Regression, a Bayesian non-parametric technique for supervised learning, stands out for its ability to quantify uncertainty. GPs have been successfully applied in solving and learning ODEs \cite{hamzi2023learning, yang2023learning} and PDEs \cite{chen2021solving, chen2023sparse, yang2023mini, raissi2017machine, raissi2018numerical}. In the context of MFGs, GPs have also been deployed to solve MFG systems, see \cite{mou2022numerical, meng2023sparse}.

This paper extends the earlier work of learning unknowns in PDEs using GPs in \cite{chen2021solving} for inverse problems in MFGs. The key difference between \eqref{intro_MFG} and a general PDE system lies in the fact that \eqref{intro_MFG} consists of a forward and backward PDE system, where $m$ is given the initial condition and $u$ admits the terminal condition. Hence, the time discretization scheme presented in \cite{chen2021solving} can not be simply applied in the MFG setting. Instead, in Section \ref{secTimDep}, we adopt a forward-backward discretization which is inspired by the finite difference method proposed in  \cite{achdou2010mean} for solving time-dependent inverse problems of MFGs.

\subsection{Outlines}
The structure of this paper is as follows. In Section \ref{secGPR}, we provide a review of the fundamental concepts of GP regression for the sake of readers. In Section \ref{secStationary}, we delve into a comprehensive analysis of Problem \ref{main_prob}, focusing on its application in stationary MFGs. This is followed by formulating the inverse problem of recovering unknown elements in stationary MFGs as an optimal recovery problem. There, we establish the existence of minimizers. Progressing to Section \ref{secTimDep}, we introduce a method of time-discretization to address inverse problems in time-dependent MFGs. The efficacy of our proposed techniques is validated through various numerical experiments, as presented in Section \ref{secNumercialExpe}. The paper concludes in Section \ref{secDisc}, where we engage in further discussion and discuss potential future work.

\begin{notations*}
In this paper, a real-valued vector  $\boldsymbol{v}$ is consistently displayed in boldface, except when it denotes a point in the physical domain. The Euclidean norm of \(\boldsymbol{v}\) is denoted by \(|\boldsymbol{v}|\), while its transpose is represented as \(\boldsymbol{v}^T\). For a vector \(\boldsymbol{v}\), \(v_i\) represents its \(i^{\textit{th}}\) element. For a matrix $M$, we denote by $|M|$ the Frobenius norm.  Given a function \(u\) and a vector \(\boldsymbol{v}\), \(u(\boldsymbol{v})\) denotes the vector \((u(v_1), \dots, u(v_N))\), with \(N\) being the length of \(\boldsymbol{v}\). We denote by $\delta_x$ the Dirac delta function concentrated at $x$. 
  Let \(\Omega\) be a subset of \(\mathbb{R}^d\). The terms \(\operatorname{int}\Omega\) and \(\partial\Omega\) refer to the interior and boundary of \(\Omega\), respectively. We denote by $\mathcal{N}(0, \gamma^2I)$ the standard multivariate normal distribution with the covariance matrix given by $\gamma^2 I$, where $\gamma>0$. For a normed vector space labeled \(V\), its norm is represented by \(\|\cdot\|_V\).

Let's consider a Banach Space \(\mathcal{U}\) associated with a quadratic norm \(\|\cdot\|_{\mathcal{U}}\). The dual space of \(\mathcal{U}\) is represented as \(\mathcal{U}^*\), and we use the notation \([\cdot, \cdot]\) for the duality pairing. We propose the existence of a covariance operator, \(\mathcal{K}_{\mathcal{U}}: \mathcal{U}^*\rightarrow \mathcal{U}\), which is linear, bijective, symmetric (meaning \([\mathcal{K}_{\mathcal{U}}\phi, \psi] = [\mathcal{K}_{\mathcal{U}}\psi, \phi]\)), and maintains positivity (\([\mathcal{K}_{\mathcal{U}}\phi, \phi] > 0\) whenever \(\phi \neq 0\)). The norm in \(\mathcal{U}\) is defined as
	$\|u\|_{\mathcal{U}}^2 = [\mathcal{K}^{-1}_{\mathcal{U}}u, u],  \text{for every } u \in \mathcal{U}$. 
Let \(\boldsymbol{\phi} = (\phi_1, \dots, \phi_P), P \in \mathbb{N}\),  which is situated within the  product space \({(\mathcal{U}^*)}^{\bigotimes P}\). For any element \(u\) in \(\mathcal{U}\), we  express the pairing \([\boldsymbol{\phi}, u]\) as 
$[\boldsymbol{\phi}, u] := ([{\phi}_1, u], \dots, [{\phi}_P, u])$.

Let $\{\mathcal{U}_i\}_{i=1}^N$ be a collection of RKHSs and define the product space $\boldsymbol{\mathcal{U}}=\mathcal{U}_1\times \dots \times \mathcal{U}_N$. We define the dual $\boldsymbol{\mathcal{U}}^*$ of $\boldsymbol{\mathcal{U}}$ as  the product space $\mathcal{U}^*_1\times\dots\times\mathcal{U}_N^*$. Let $\boldsymbol{\phi}=(\phi_1,\dots, \phi_N)\in \boldsymbol{\mathcal{U}}^*$ and  $\boldsymbol{u}=(u_1,\dots, u_N)\in \boldsymbol{\mathcal{U}}$, we define $\boldsymbol{\phi}\odot \boldsymbol{u} = ([\phi_1, u_1], \dots, [\phi_N, u_N])$. 
Let $\Phi$ be a matrix of linear operators with the $i^{\textit{th}}$ column, denoted by $\boldsymbol{\phi}_i$, in $\boldsymbol{\mathcal{U}}^*$. For $\boldsymbol{u}=(u_1,\dots, u_N)\in \mathbb{\mathcal{U}}$, we represent the matrix $[\Phi, \boldsymbol{u}]$ such that the  $i^{\textit{th}}$ column is defined as $\boldsymbol{\phi}_i \odot \boldsymbol{u}$. For a matrix $\boldsymbol{Y}$, let $\overset{\rightarrow}{\boldsymbol{Y}}$ be the vector formed by concatenating ${\boldsymbol{Y}}$'s columns.  The symbol \(C\) stands for a variable positive real number which can vary from one instance to another.
\end{notations*}

\section{Prerequisites for GP Regression}
\label{secGPR}
In this section, we briefly revisit the fundamentals of GP Regression and its connection to Reproducing Kernel Hilbert Spaces (RKHSs). We refer readers to \cite{rasmussen2006gaussian, owhadi2019operator} for an in-depth exploration of the subject.

\subsection{Learning Scalar-valued Functions with GPs}
\label{sub:GPR:DM}
Let \(\Omega \subseteq \mathbb{R}^d\) be an open subset. A real-valued GP, \(f: \Omega \to \mathbb{R}\), is defined such that for any finite set $\boldsymbol{x}$ of points in $\Omega$, \(f(\boldsymbol{x})\) follows a joint Gaussian distribution. A GP $f$ is usually characterized by a mean function \(\mu: \Omega \to \mathbb{R}\) and a covariance function \(K: \Omega \times \Omega \to \mathbb{R}\) such that \(E(f(x)) = \mu(x)\) and \(\operatorname{Cov}(f(x), f(x')) = K(x, x'), \forall x, x'\in \Omega\). In this case, we denote $f\sim \mathcal{GP}(\mu, K)$. In supervised learning, we aim to build a GP estimator \(f^\dagger\) from a training set \((x_i, y_i)_{i=1}^N\). To do that, we  assume \(f \sim \mathcal{GP}(0, K)\) and define the estimator $f^\dagger$ as the mean of the posterior of $f$ conditioned on the training set, i.e., 
\begin{align}
\label{fptst}
f^\dagger = E[f| f(x_i) = y_i, i = 1, \dots, N].
\end{align}
The estimator $f^\dagger$ admits an explicit formula given by 
\begin{align}
\label{fpexpt}
f^\dagger(x) = K(x, \boldsymbol{x})K(\boldsymbol{x}, \boldsymbol{x})^{-1}\boldsymbol{y},
\end{align}
where $\boldsymbol{x}=(x_i)_{i=1}^N, \boldsymbol{y}=(y_i)_{i=1}^N$, $K(x, \boldsymbol{x})=(K(x, x_i))_{i=1}^N$, and \(K(\boldsymbol{x}, \boldsymbol{x})\) is a matrix with entries $K(x_i, x_j)$. The estimator \( f^\dagger \) can also be understood from an optimization perspective. Let $K$ be a positive definite covariance function. Then, according to the Moore--Aronszajn theorem, there is a RKHS $\mathcal{U}$ associated with $K$, such that $f(x) = \langle f, K(x, \cdot)\rangle$. Then, the estimator $f^\dagger$ in \eqref{fptst} and \eqref{fpexpt} is a minimizer of the following optimal recover problem
\begin{align}
\label{gprorp}
\begin{cases}
\min\limits_{f\in \mathcal{U}} \|f\|_{\mathcal{U}}^2\\
\operatorname{s.t.} f(x_i) = y_i, \forall i\in \{1,\dots, N\}. 
\end{cases}
\end{align}
We call \eqref{fpexpt} the representer formula for the optimization problem \eqref{gprorp} and it is the cornerstone of GP Regression. Let $\mathcal{U}^*$ be the dual of $\mathcal{U}$ and denote by $[,]$ the duality pairing. Let $\delta_x$ bet the Dirac function centered at $x$. Then, $\delta_x\in \mathcal{U}^*$ for $x\in \Omega$. Define $\boldsymbol{\delta}=(\delta_{x_1}, \dots, \delta_{x_N})$. Then, for simplicity, we write \eqref{gprorp} as 
\begin{align*}
\begin{cases}
\min\limits_{f\in \mathcal{U}} \|f\|_{\mathcal{U}}^2\\
\operatorname{s.t.} [\boldsymbol{\delta}, f] = \boldsymbol{y}.
\end{cases}
\end{align*}
For a kernel $K$ and a collection $\boldsymbol{\phi}$ of linear operators in the dual space of the RKHS associated with $K$, let $K(x, \boldsymbol{\phi})$ be the vector with entries $\int_{\Omega}K(x,x' )\phi_i(x')\dif x'$  and let $K(\boldsymbol{\phi}, \boldsymbol{\phi})$, called the Gram matrix, be with entries $\int_{\Omega}\int_{\Omega}K(x, x)\phi_i(x')\phi_j(x')\dif x\dif x'$.  Then, we reformulate  $f^\dagger$ in \eqref{fpexpt} as 
\begin{align*}
f^\dagger(x) = K(x, \boldsymbol{\delta})K(\boldsymbol{\delta}, \boldsymbol{\delta})^{-1}\boldsymbol{y}.
\end{align*}
Next, we rewrite the training set $(x_i, y_i)_{i=1}^N$ as $(\delta_{x_i}, y_i)_{i=1}^N$. Then, if we replace the training set $(\delta_{x_i}, y_i)_{i=1}^N$ with $(\phi_i, y_i)_{i=1}^N$, where $\{\phi_i\}_{i=1}^N$ contains other independent linear operators in $\mathcal{U}^*$,  the mean of the posterior of the prior $f$, $f\sim \mathcal{GP}(0, K)$, conditioned on the training set $(\phi_i, y_i)_{i=1}^N$ is given by \cite[Sec 17.8]{owhadi2019operator}
\begin{align}
\label{goprs}
f^\dagger(x) = K(x, \boldsymbol{\phi})K(\boldsymbol{\phi}, \boldsymbol{\phi})^{-1}\boldsymbol{y},
\end{align}
which is the minimizer of the following optimal recovery problem
\begin{align*}
\begin{cases}
\min\limits_{f\in \mathcal{U}} \|f\|_{\mathcal{U}}^2\\
\operatorname{s.t.} [\boldsymbol{\phi}, f] = \boldsymbol{y}.
\end{cases}
\end{align*}

\subsection{Learning Vector-valued Functions with GPs}
\label{sub:GPR:LM}
The above discussion can be extended to consider regressing vector-valued functions with GPs. We refer readers to \cite{alvarez2012kernels, micchelli2005learning} for further details.

Let \(\Omega \subseteq \mathbb{R}^d\) be an open subset. A vector-valued GP, \( \boldsymbol{f}: \Omega \to \mathbb{R}^m \), is defined such that for any  \(\boldsymbol{X} \in \Omega^{\bigotimes N}\), \( \boldsymbol{f}(\boldsymbol{X})\in \mathbb{R}^{N\times m} \) follows a joint Gaussian distribution. A vector-valued GP \(\boldsymbol{f}\) is characterized by a mean function \(\boldsymbol{\mu}: \Omega \to \mathbb{R}^m\) and a covariance function \(K: \Omega \times \Omega \to \mathbb{R}^{m \times m}\) such that \(E(\boldsymbol{f}(x)) = \boldsymbol{\mu}(x)\) and \(\operatorname{Cov}(\boldsymbol{f}(x), \mathbf{f}(x')) = K(x, x')\), for all $x, x'\in \Omega$. We denote this as \(\boldsymbol{f} \sim \mathcal{GP}(\boldsymbol{\mu}, K)\). In the setting of learning vector-valued functions, our goal is to construct a GP estimator \(\boldsymbol{f}^\dagger\) from a training set \((x_i, \boldsymbol{Y}_i)_{i=1}^N\), where $\boldsymbol{Y}_i\in \mathbb{R}^m$. We assume \(\boldsymbol{f} \sim \mathcal{GP}(\boldsymbol{0}, K)\) and define the estimator \(\boldsymbol{f}^\dagger\) as the mean of the posterior of \(\boldsymbol{f}\) conditioned on the training set, i.e.,
\begin{align*}
\boldsymbol{f}^\dagger = E[\boldsymbol{f}| \boldsymbol{f}(x_i) = \boldsymbol{Y}_i, i = 1, \dots, N].
\end{align*}
Let $\boldsymbol{Y}$ be the matrix with the $i^{\textit{th}}$ column being $\boldsymbol{Y}_i$, and let \(\overset{\rightarrow}{\boldsymbol{Y}}\) bet the vector obtained by concatenating the columns of $\boldsymbol{Y}$. The estimator \(\boldsymbol{f}^\dagger\) is explicitly given by 
\begin{align*}
\boldsymbol{f}^\dagger(x) = K(x, \boldsymbol{x})K(\boldsymbol{x}, \boldsymbol{x})^{-1}\overset{\rightarrow}{\boldsymbol{Y}},
\end{align*}
where $K(x, \boldsymbol{x})$ is a matrix with $m$ rows and $N\times m$ columns by concatenating matrices $K(x, x_i)$ for $i=1,\dots, N$, and  \(K(\boldsymbol{x}, \boldsymbol{x})\) is defined as
\begin{align*}
K(\boldsymbol{x}, \boldsymbol{x}')=\begin{bmatrix}
K(x_1, x_1) & \dots & K(x_1, x_N)\\
\dots & \dots & \dots \\
K(x_N, x_1) & \dots & K(x_N, x_N)
\end{bmatrix}.
\end{align*}
The estimator \(\boldsymbol{f}^\dagger\) also has an optimization interpretation. Let $\boldsymbol{\mathcal{U}}$ be the vector-valued RKHS associated with $K$, where \(\boldsymbol{f}(x)^T\boldsymbol{c} = \langle \boldsymbol{f}, K(\cdot, x)\boldsymbol{c}\rangle\), for any $\boldsymbol{c}\in \mathbb{R}^m$ and $\boldsymbol{f}\in \boldsymbol{\mathcal{U}}$. The estimator \(\boldsymbol{f}^\dagger\) minimizes the following optimal recovery problem
\begin{align*}
\begin{cases}
\min\limits_{\boldsymbol{f}\in \mathcal{U}} \|\boldsymbol{f}\|_{\mathcal{U}}^2\\
\operatorname{s.t.} \boldsymbol{f}(x_i) = \boldsymbol{Y}_i, \forall i\in \{1,\dots, N\}. 
\end{cases}
\end{align*}
In this paper, we are interested in the special case where each output of $\boldsymbol{f}$ can be modeled as an independent GP, i.e., $K(x, x')$ is a diagonal matrix for all $x, x'\in \Omega$. Denote $K(x, x')=\operatorname{diag}(K_1(x, x'), \dots,  K_m(x, x'))$. For simplicity, we represent $\boldsymbol{f}(x)=(f_1(x), \dots, f_m(x))$ and $f_i \in \mathcal{U}_i$ with $\mathcal{U}_i$ being a RKHS associated with the kernel $K_i$ for $i=1,\dots, m$. Hence, the RKHS $\boldsymbol{\mathcal{U}}$ is equivalent to $\mathcal{U}_1\times\dots\times \mathcal{U}_m$. Accordingly, $\|\boldsymbol{f}\|_{\boldsymbol{\mathcal{U}}}^2 = \sum_{i=1}^N\|f_i\|_{\mathcal{U}_i}^2, \forall \boldsymbol{f}\in \boldsymbol{\mathcal{U}}$. 

Let $\boldsymbol{\mathcal{U}}^*=\mathcal{U}_1^*\times \dots \times \mathcal{U}_m^*$.  Consider the training set with $(\boldsymbol{\phi}_i, \boldsymbol{Y}_i)_{i=1}^N$, where $\boldsymbol{\phi}_i$ is a vector of linear operators with $\phi_{i,j}\in \mathcal{U}_j^*$. Let $\Phi$ be the matrix with the $i^{\textit{th}}$ column being $\boldsymbol{\phi}_i$. Let $\boldsymbol{\phi} := \overset{\rightarrow}{\Phi}$,  the vector concatenating $(\boldsymbol{\phi}_i)_{i=1}^N$. Then, the posterior mean of the prior $\boldsymbol{f}$, assumed as $\boldsymbol{f} \sim \mathcal{GP}(\boldsymbol{0}, K)$, is given by $\boldsymbol{f}^\dagger = E[\boldsymbol{f}| [\Phi, \boldsymbol{f}] = \boldsymbol{Y} ]$ admits the following representer formula
\begin{align}
\label{reprefmlve}
\boldsymbol{f}^\dagger(x) = K(x, \boldsymbol{\phi})K(\boldsymbol{\phi}, \boldsymbol{\phi})^{-1}\overset{\rightarrow}{\boldsymbol{Y}},
\end{align}
where $\overset{\rightarrow}{\boldsymbol{Y}}$ is the vector concatenating columns of $\boldsymbol{Y}$ and $\boldsymbol{f}^\dagger$ minimizes the following  optimal recovery problem:
\begin{align*}
\begin{cases}
\min\limits_{\boldsymbol{f}\in \boldsymbol{\mathcal{U}} } \|\boldsymbol{f}\|_{\boldsymbol{\mathcal{U}}}^2\\
\operatorname{s.t.} [\Phi, \boldsymbol{f}] = \boldsymbol{Y}.
\end{cases}
\end{align*}

\section{An Inverse Problem for Stationary MFGs}
\label{secStationary}
In this section, we elaborate Problem \ref{main_prob} in the context of stationary MFGs. Inverse problems in stationary MFGs are intriguing as they represent the long-term limits of time-dependent MFGs \cite{
cannarsa2020long, cirant2021long, cardaliaguet2012long}. Under certain conditions, agents' states converge exponentially to a steady state \cite{
cannarsa2020long, cirant2021long, cardaliaguet2012long}. Hence, MFG data is likely to reflect this steady state, making it practical to focus on recovering strategies and environmental information in this equilibrium. Then, we introduce a GP framework to solve the inverse problem. This approach builds upon the GP method from \cite{chen2021solving}, adapting it to coupled PDE systems with integral constraints. Subsection \ref{sub:ProbSMFG} redefines Problem \ref{main_prob} for stationary MFGs, Subsection \ref{sub:OptProb} formulates the inverse problem as an optimal recovery problem and provides a general framework for solving inverse problems of stationary MFGs. To illustrate the framework, we describe our method through a concrete example in Subsection \ref{sub:concex}. Subsection \ref{sub:eximi} proves the existence of a minimizer, confirming the solvability of our proposed inverse problem. Subsection \ref{sub:probpersp} offers a probabilistic interpretation of our method, linking it to maximum likelihood estimation.

\subsection{An Inverse Problem  for  Stationary MFGs}
\label{sub:ProbSMFG}
Let $\Omega$ be a subset of $\mathbb{R}^d$. Suppose that the stationary MFGs of our interests admit the following form
\begin{align}
	\label{gsmfgs}
	\begin{cases}
		\boldsymbol{\mathcal{P}}(u^*, m^*, \overline{H}^*, \boldsymbol{V}^*)(x)= 0, &\forall x \in \operatorname{int}\Omega,\\
		\mathcal{B}(u^*, m^*, \boldsymbol{V}^*)(x) = 0, &\forall x \in \partial \Omega,\\
		\int_{\Omega}u^*\dif x=0, \int_{\Omega}m^*\dif x = 1, 
	\end{cases}
\end{align}  
where $\boldsymbol{V}^*$ contains a set of functions, which represent  environment configurations. Given $\boldsymbol{V}^*$, $(u^*, m^*, \overline{H}^*)$ solves the stationary MFG \eqref{gsmfgs} such that $u^*$ stands for the value function, $m^*$ represents the density of populations of agents, and  $\overline{H}^*$ is real number ensuring the unit mass of $m^*$. In this context, $\boldsymbol{\mathcal{P}}$ is a collection of nonlinear differential operators, while $\mathcal{B}$ stands for a boundary operator. We proceed under the assumption that equation \eqref{gsmfgs} has a unique classical solution $(u^*, m^*, \overline{H}^*)$ when the underlying environment configuration $\boldsymbol{V}^*$ is given. The equation \eqref{gsmfgs} represents a general form of the long-time limit of the time-dependent MFG described in \eqref{intro_MFG}. For discussions on the long-time behavior of MFGs, see \cite{cannarsa2020long, cirant2021long, cardaliaguet2012long}. We detail Problem \ref{main_prob} for stationary MFGs of the form \eqref{gsmfgs} in the following problem.
\begin{problem}
\label{st_prob}
Let $\boldsymbol{V}^*=\{V^*_i\}_{i=1}^{N_v}$ be a collection of functions on $\Omega$, where $N_v\in \mathbb{N}$. Assume that given $\boldsymbol{V}^*$, \eqref{gsmfgs} admits a unique classical solution $(u^*, m^*, \overline{H}^*)$. Suppose that we only have some partial noisy observations on $m^*$ and on a subset $\boldsymbol{V}_o^*$ of $\boldsymbol{V}^*$. The objective is to infer the values of $u^*$, $m^*$, $\overline{H}^*$, and $\boldsymbol{V}^*$. To elaborate, we have 
\begin{enumerate}
    \item \textbf{Partial noisy observations on $m^*$}.  We have a set of linear operators $\{\phi_l^o\}_{l=1}^{N_m}$, $N_m\in \mathbb{N}$ where some noisy observations on $m^*$ are available. These observations are represented as $\boldsymbol{m}^o$, i.e. $\boldsymbol{m}^o = ([\phi_1^o, m^*], \dots, [\phi_{N_m}^o, m^*]) + \boldsymbol{\epsilon}$, $\boldsymbol{\epsilon}\sim \mathcal{N}(0, \gamma^2 I)$. We call $\boldsymbol{\phi}^o=(\phi_1^o, \dots, \phi_{N_m}^o)$ the vector of observation operators. For instance, if we only have observations of $m^*$ at a finite set of collocation points,  $\boldsymbol{\phi}^o$ contains Diracs centered at these points. 
    \item \textbf{Partial noisy observations of $\boldsymbol{V}^*$}. Let $\boldsymbol{V}_o^*=\{V_j^*\}_{j\in \mathcal{O}_v}$ with  $\mathcal{O}_v\subset \{1,\dots, N_v\}$ be a subset of $\boldsymbol{V}^*$. For this subset, we possess noisy observations at collections of  linear operators $\{\{\psi^{o,j}_e\}_{e=1}^{N^{j}_{v}}\}_{j\in \mathcal{O}_v}$. These observations can be collected into a matrix $\boldsymbol{V}^o$ such that the $j^{\textit{th}}$ row of $\boldsymbol{V}^o$, denoted by $\boldsymbol{V}_j^o$  is given by 
\begin{align*}
\boldsymbol{V}^o_j = ([\psi^{o,j}_1, V_j], \dots, [\psi^{o,j}_{N_v^j}, V_j]) + \boldsymbol{\epsilon}_j, \boldsymbol{\epsilon}_j \sim \mathcal{N}(0, \gamma_j^2I). 
\end{align*}
\end{enumerate}
\textbf{Inverse Problem Statement}\\
Suppose that agents are involved in the stationary MFG in \eqref{gsmfgs}, we infer $u^*$, $m^*$, $\overline{H}^*$, and $\boldsymbol{V}^*$ based on $\boldsymbol{m}^o$ and $\boldsymbol{V}^o$.
\end{problem}

\subsection{The Optimal Recovery Problem}
\label{sub:OptProb}
To solve Problem \ref{st_prob}, we approximate $(u^*, m^*, \overline{H}^*, \boldsymbol{V}^*)$ by GPs conditioned on PDEs at sampled collocation points in $\Omega$. Then, we compute the solution by calculating the MAP points of such conditioned GPs. More precisely, we take a set of samples $\{x_i\}_{i=1}^M$ in such a way that $x_1, \dots, x_{M_\Omega}\in \operatorname{int}\Omega$ and $x_{M_\Omega+1}, \dots, x_M\in \partial \Omega$ for $1\leq M_\Omega\leq M$. We put GP priors on $u^*, m^*, \overline{H}^*$ and $\boldsymbol{V}^*$ by assuming that $u^*\sim \mathcal{GP}(0, K_u)$, $m^*\sim \mathcal{GP}(0, K_m)$, $\overline{H}^*\sim \mathcal{N}(0, 1)$, and $\boldsymbol{V}^*\sim \mathcal{GP}(\boldsymbol{0}, K_{\boldsymbol{V}})$, where $\boldsymbol{V}^*$ is a vector-valued GP with independent outputs.  Let $\mathcal{U}$,  $\mathcal{M}$, and $\boldsymbol{\mathcal{V}}$ be RKHSs associated with $K_u, K_m$, and $K_{\boldsymbol{V}}$, separately. Denote $\boldsymbol{\phi}^o = (\phi^{o}_1, \dots, \phi_{{N_m}}^o)$. Let $\boldsymbol{m}^o$ be a noisy observations of $[\boldsymbol{\phi}^o, m^*]$ such that $[\boldsymbol{\phi}^o, m^*] - \boldsymbol{m}^o \in \mathcal{N}(0, \gamma^2 I)$. Let $\Psi$ be a matrix of linear functionals with each column in $\boldsymbol{\mathcal{V}}^*$, the dual of $\boldsymbol{\mathcal{V}}$,  such that  we have a noisy observation $\boldsymbol{V}^o$ of $[\Psi, \boldsymbol{V}^*]$. 
Then, to solve Problem \ref{st_prob}, we introduce a positive penalization parameter $\beta$ and approximate $(u^*, m^*, \overline{H}^*, \boldsymbol{V}^*)$ by a minimizer of the following optimal recovery problem
\begin{align}
\label{OptGPProb}
\begin{cases}
\min\limits_{(u, m ,\overline{H}, \boldsymbol{V})\in \mathcal{U}\times\mathcal{M}\times\mathbb{R}\times \boldsymbol{\mathcal{V}}} \|u\|_{\mathcal{U}}^2 + \|m\|_{\mathcal{M}}^2 + |\overline{H}|^2 + \|\boldsymbol{V}\|_{\boldsymbol{\mathcal{V}}}^2+ \frac{1}{\gamma^2} |[\boldsymbol{\phi}^o, m] - \boldsymbol{m}^o|^2 + |\Sigma^{-1}([\Psi, \boldsymbol{V}] - \boldsymbol{V}^o)|^2\\
\quad\quad\quad\quad\quad\quad\quad+\beta \bigg|\frac{1}{M_\Omega}\sum\limits_{i=1}^{M_\Omega}u(x_i)\bigg|^2+\beta\bigg|\frac{1}{M_\Omega}\sum\limits_{i=1}^{M_\Omega}m(x_i)-1\bigg|^2\\
\text{s.t.}\quad \boldsymbol{\mathcal{P}}(u, m, \overline{H}, \boldsymbol{V})(x_i) = 0, \quad\text{for}\ i = 1, \dots M_{\Omega},\\
\quad\quad\ \mathcal{B}(u, m, \boldsymbol{V})(x_j) = 0, \quad \text{for}\ j = M_{\Omega+1},\dots, M,
\end{cases}
\end{align}
where $\Sigma = \operatorname{diag}(\gamma_1, \dots, \gamma_{N_v})$. The formulation in \eqref{OptGPProb} extends the one in Section 4.4 of \cite{chen2021solving}, adapting it to a coupled PDE system with integral constraints. Our approach, which involves discretizing the integral constraints and incorporating them into the integration process, circumvents the need for kernel integration during covariance matrix assembly. Should the solution of \eqref{gsmfgs} lack sufficient smoothness, we recommend employing the vanishing viscosity method \cite{cardaliaguet2010notes}, or regularizing the MFG using smooth mollifiers \cite{cesaroni2019stationary}. This approach yields a system with solutions of higher regularity, to which our numerical methods can be applied.

The minimization in \eqref{OptGPProb} involves optimization within function spaces. For computational feasibility through numerical algorithms, we reformulate \eqref{OptGPProb} into a finite-dimensional minimization problem.  To do that, we make the following  assumptions about the forms of $\boldsymbol{\mathcal{P}}$ and $\mathcal{B}$. 
\begin{hyp}
\label{hyp_pb}
Assume that there exists collections of bounded linear operators $\boldsymbol{L}^{u, \Omega}$, $\boldsymbol{L}^{u, \partial\Omega}$, $\boldsymbol{L}^{m, \Omega}$, $\boldsymbol{L}^{m, \partial \Omega}$, $\boldsymbol{L}^{\boldsymbol{{V}}, \Omega}$, and $\boldsymbol{L}^{\boldsymbol{{V}}, \partial\Omega}$,  and continuous nonlinear maps $\boldsymbol{{P}}$ and $B$ such that
\begin{align*}
\begin{cases}
\boldsymbol{\mathcal{P}}(u, m, \overline{H}, \boldsymbol{V})(x) = \boldsymbol{{P}}(\boldsymbol{L}^{u, \Omega}(u)(x), \boldsymbol{L}^{m, \Omega}(m)(x), \overline{H}, \boldsymbol{L}^{\boldsymbol{V}, \Omega}(\boldsymbol{V})(x)), \forall x\in \operatorname{int}\Omega, \\
\mathcal{B}(u, m, \boldsymbol{V})(x) = B(\boldsymbol{L}^{u, \partial \Omega}(u)(x), \boldsymbol{L}^{m, \partial \Omega}(m)(x),  \boldsymbol{L}^{\boldsymbol{V}, \partial \Omega}(\boldsymbol{V})(x)), \forall x\in \partial \Omega. 
\end{cases}
\end{align*}
In the following texts, we denote $\boldsymbol{L}^{u} := (\boldsymbol{L}^{u, \Omega}, \boldsymbol{L}^{u, \partial\Omega})$, $\boldsymbol{L}^{m} := (\boldsymbol{L}^{m, \Omega}, \boldsymbol{L}^{m, \partial\Omega})$, and $\boldsymbol{L}^{\boldsymbol{{V}}} := (\boldsymbol{L}^{\boldsymbol{{V}}, \Omega}, \boldsymbol{L}^{\boldsymbol{{V}}, \partial\Omega})$.
\end{hyp}
Under Assumption \ref{hyp_pb}, we define a functional vector $\boldsymbol{\delta}:=(\boldsymbol{\delta}^{\Omega}, \boldsymbol{\delta}^{\partial\Omega})$, where $\boldsymbol{\delta}^\Omega=(\delta_{x_1},\dots, \delta_{x_{M_\Omega}})$ and $\boldsymbol{\delta}^{\partial\Omega}=(\delta_{x_{M_\Omega+1}},\dots, \delta_{x_{M}})$. Let $\widetilde{\boldsymbol{\phi}}^u = \boldsymbol{\delta}\circ \boldsymbol{L}^u$, $\widetilde{\boldsymbol{\phi}}^m = \boldsymbol{\delta}\circ \boldsymbol{L}^m$, and $\widetilde{\boldsymbol{\phi}}^{\boldsymbol{V}} = \boldsymbol{\delta}\circ \boldsymbol{L}^{\boldsymbol{V}}$. We define the nonlinear map $G$ such that for any $u\in \mathcal{U}$, $m\in \mathcal{M}$, $\overline{H}\in \mathbb{R}$, and $\boldsymbol{V}\in \boldsymbol{\mathcal{V}}$, we have 
\begin{align*}
(G([\widetilde{\boldsymbol{\phi}}^{u}, u], [\widetilde{\boldsymbol{\phi}}^m, m], \overline{H}, [\widetilde{\boldsymbol{\phi}}^{\boldsymbol{V}}, \boldsymbol{V}]))_i=\begin{cases}
\boldsymbol{{P}}([\delta_{x_i}\circ \boldsymbol{L}^{\Omega, u}, u], [\delta_{x_i}\circ \boldsymbol{L}^{\Omega, m}, m], \overline{H}, [\delta_{x_i}\circ \boldsymbol{L}^{\Omega, \boldsymbol{V}}, \boldsymbol{V}])\\
\quad\quad\quad\quad\quad\quad\quad \text{if}\ i \in \{1, \dots, M_{\Omega}\}, \\
B([\delta_{x_i}\circ \boldsymbol{L}^{\partial\Omega, u}, u], [\delta_{x_i}\circ \boldsymbol{L}^{\partial\Omega, m}, m], [\delta_{x_i}\circ \boldsymbol{L}^{\partial\Omega, \boldsymbol{V}}, \boldsymbol{V}])\\
\quad\quad\quad\quad\quad\quad\quad \text{if}\ i \in \{M_{\Omega}+1, \dots, M\}. 
\end{cases}
\end{align*}
Then, we rewrite \eqref{OptGPProb} as
\begin{align}
\label{OptGPProb_Cpt}
\begin{cases}
\min\limits_{(u, m ,\overline{H}, \boldsymbol{V})\in \mathcal{U}\times\mathcal{M}\times\mathbb{R}\times \boldsymbol{\mathcal{V}}} \|u\|_{\mathcal{U}}^2 + \|m\|_{\mathcal{M}}^2 + |\overline{H}|^2 + \|\boldsymbol{V}\|_{\boldsymbol{\mathcal{V}}}^2+ \frac{1}{\gamma^2} |[\boldsymbol{\phi}^o, m] - \boldsymbol{m}^o|^2 + |\Sigma^{-1}([\Psi, \boldsymbol{V}] - \boldsymbol{V}^o)|^2\\
\quad\quad\quad\quad\quad\quad\quad+\beta \bigg|\frac{1}{M_\Omega}\sum\limits_{i=1}^{M_\Omega}u(x_i)\bigg|^2+\beta\bigg|\frac{1}{M_\Omega}\sum\limits_{i=1}^{M_\Omega}m(x_i)-1\bigg|^2\\
\text{s.t.}\quad G([\widetilde{\boldsymbol{\phi}}^{u}, u], [\widetilde{\boldsymbol{\phi}}^{m}, m], \overline{H}, [\widetilde{\boldsymbol{\phi}}^{\boldsymbol{V}}, \boldsymbol{V}]) = 0. 
\end{cases}
\end{align}
To solve \eqref{OptGPProb_Cpt}, we introduce latent variables, $\boldsymbol{z}, \boldsymbol{\rho}$, and $\boldsymbol{v}$, and rewrite \eqref{OptGPProb_Cpt} as 
\begin{align}
\label{OptGPProb_Cpttl}
\begin{cases}
\min\limits_{\boldsymbol{z}, \boldsymbol{\rho}, \boldsymbol{v}, \overline{H}}
\begin{cases}
\min\limits_{(u, m, \boldsymbol{V})\in \mathcal{U}\times\mathcal{M}\times \boldsymbol{\mathcal{V}}} \|u\|_{\mathcal{U}}^2 + \|m\|_{\mathcal{M}}^2 + \|\boldsymbol{V}\|_{\boldsymbol{\mathcal{V}}}^2\\
\text{s.t.} \quad [\boldsymbol{\delta}^{\Omega}, u] = \boldsymbol{z}^{(1), \Omega}, [\boldsymbol{\delta}^{\partial\Omega}, u] = \boldsymbol{z}^{(1), \partial\Omega},  [\widetilde{\boldsymbol{\phi}}^u, u] = \boldsymbol{z}^{(2)}, \\ \quad
\quad[\boldsymbol{\delta}^{\Omega}, m] = \boldsymbol{\rho}^{(1), \Omega}, [\boldsymbol{\delta}^{\partial\Omega}, m] = \boldsymbol{\rho}^{(1), \partial\Omega}, [\boldsymbol{\phi}^o, m] =  \boldsymbol{\rho}^{(2)}, 
[\widetilde{\boldsymbol{\phi}}^m, m] = \boldsymbol{\rho}^{(3)},\\ \quad\quad[\widetilde{\boldsymbol{\phi}}^{\boldsymbol{V}}, \boldsymbol{V}] = \boldsymbol{v}^{(1)},  [\Psi, \boldsymbol{V}] = \boldsymbol{v}^{(2)}, 
\end{cases}\\
\quad\quad + |\overline{H}|^2 + \frac{1}{\gamma^2} |\boldsymbol{\rho}^{(2)} - \boldsymbol{m}^o|^2 + |\Sigma^{-1}(\boldsymbol{v}^{(2)} - \boldsymbol{V}^o)|^2+\beta \bigg|\frac{1}{M_\Omega}\sum\limits_{i=1}^{M_\Omega}z^{(1), \Omega}_i\bigg|^2+\beta\bigg|\frac{1}{M_\Omega}\sum\limits_{i=1}^{M_\Omega}\rho^{(1), \Omega}_i-1\bigg|^2\\
\text{s.t.}\quad G(\boldsymbol{z}^{(2)}, \boldsymbol{\rho}^{(3)}, \overline{H}, \boldsymbol{v}^{(1)}) = 0,
\end{cases}
\end{align}
where $\boldsymbol{z} = (\boldsymbol{z}^{(1), \Omega}, \boldsymbol{z}^{(1), \partial\Omega}, \boldsymbol{z}^{(2)})$, $\boldsymbol{\rho}=(\boldsymbol{\rho}^{(1), \Omega}, \boldsymbol{\rho}^{(1), \partial\Omega}, \boldsymbol{\rho}^{(2)}, \boldsymbol{\rho}^{(3)})$, and $\boldsymbol{v}=(\boldsymbol{v}^{(1)}, \boldsymbol{v}^{(2)})$.  Denote $\boldsymbol{\phi}^u := (\boldsymbol{\delta}, \widetilde{\boldsymbol{\phi}}^u)$,  $\boldsymbol{\phi}^m :=(\boldsymbol{\delta}, \boldsymbol{\phi}^o, \widetilde{\boldsymbol{\phi}}^m)$, and $\boldsymbol{\phi}^{\boldsymbol{V}} :=(\widetilde{\boldsymbol{\phi}}^{\boldsymbol{V}}, \Psi)$. 

For a kernel $K$ and a collection $\boldsymbol{\phi}$ of linear operators in the dual space of the RKHS associated with $K$, let $K(x, \boldsymbol{\phi})$ be the vector with entries $\int_{\Omega}K(x,x' )\phi_i(x')\dif x'$  and let $K(\boldsymbol{\phi}, \boldsymbol{\phi})$, called the Gram matrix, be with entries $\int_{\Omega}\int_{\Omega}K(x, x)\phi_i(x')\phi_j(x')\dif x\dif x'$. Let $K_u$, $K_m$, and $K_{\boldsymbol{V}}$ be kernels associated with RKHSs $\mathcal{U}$, $\mathcal{M}$, and $\mathcal{V}$. Let $(u^\dagger, m^\dagger, \boldsymbol{V}^\dagger)$ be the solution to the first level minimization problem in \eqref{OptGPProb_Cpttl} given $(\boldsymbol{z}, \boldsymbol{\rho}, \boldsymbol{v}, \overline{H})$.  By the representer formulas in \eqref{goprs} and \eqref{reprefmlve}, we get
\begin{align}
	\label{1drepre}
	\begin{cases}
		u^\dagger(x)=K_u(x,\boldsymbol{\phi}^u)K_u(\boldsymbol{\phi}^u, \boldsymbol{\phi}^u)^{-1}\boldsymbol{z},\\
		m^\dagger(x)=K_m(x,{\boldsymbol{\phi}}^m) K({\boldsymbol{\phi}}^m, {\boldsymbol{\phi}}^m)^{-1}\boldsymbol{\rho},\\
		\boldsymbol{V}^\dagger(x)= {K}_{\boldsymbol{V}}(x,{\boldsymbol{\phi}}^{\boldsymbol{V}}) {K}_{\boldsymbol{V}}({\boldsymbol{\phi}}^{\boldsymbol{V}},  {\boldsymbol{\phi}}^{\boldsymbol{V}})^{-1}\boldsymbol{v}. 
	\end{cases}
\end{align}
Thus, we have
\begin{align*}
	\begin{cases}		\|u^\dagger\|_{\mathcal{U}}^2=\boldsymbol{z}^TK_u(\boldsymbol{\phi}^u, \boldsymbol{\phi}^u)^{-1}\boldsymbol{z},\\
\|m^\dagger\|_{\mathcal{M}}^2=\boldsymbol{\rho}^TK_m({\boldsymbol{\phi}}^m, {\boldsymbol{\phi}}^m)^{-1}\boldsymbol{\rho}, \\
\|\boldsymbol{V}^\dagger\|_{\boldsymbol{\mathcal{V}}}^2=\boldsymbol{v}^T{K}_{\boldsymbol{V}}({\boldsymbol{\phi}}^{\boldsymbol{V}}, {\boldsymbol{\phi}}^{\boldsymbol{V}})^{-1}\boldsymbol{v}, 
	\end{cases}
\end{align*}
Hence, we can formulate \eqref{OptGPProb_Cpttl} as a finite-dimensional optimization problem
\begin{align}
\label{OptGPProb_Cpttl_fin}
\begin{cases}
\min\limits_{\boldsymbol{z}, \boldsymbol{\rho}, \boldsymbol{v}, \overline{H}} \boldsymbol{z}^TK_u(\boldsymbol{\phi}^u, \boldsymbol{\phi}^u)^{-1}\boldsymbol{z} + \boldsymbol{\rho}^TK_m({\boldsymbol{\phi}}^m, {\boldsymbol{\phi}}^m)^{-1}\boldsymbol{\rho} + \boldsymbol{v}^T{K}_{\boldsymbol{V}}({\boldsymbol{\phi}}^{\boldsymbol{V}}, {\boldsymbol{\phi}}^{\boldsymbol{V}})^{-1}\boldsymbol{v} + |\overline{H}|^2 \\
\quad\quad\quad\quad + \frac{1}{\gamma^2} |\boldsymbol{\rho}^{(2)} - \boldsymbol{m}^o|^2 + |\Sigma^{-1}(\boldsymbol{v}^{(2)} - \boldsymbol{V}^o)|^2+\beta \bigg|\frac{1}{M_\Omega}\sum\limits_{i=1}^{M_\Omega}z^{(1), \Omega}_i\bigg|^2+\beta\bigg|\frac{1}{M_\Omega}\sum\limits_{i=1}^{M_\Omega}\rho^{(1), \Omega}_i-1\bigg|^2\\
\text{s.t.}\quad G(\boldsymbol{z}^{(2)}, \boldsymbol{\rho}^{(3)}, \overline{H}, \boldsymbol{v}^{(1)}) = 0.
\end{cases}
\end{align}
To deal with the nonlinear constraints in \eqref{OptGPProb_Cpttl_fin}, we introduce a prescribed penalization parameter $\alpha>0$ and consider the following relaxation 
\begin{align}
\label{OptGPProb_Cpttl_f_relax}
\begin{cases}
\min\limits_{\boldsymbol{z}, \boldsymbol{\rho}, \boldsymbol{v}, \overline{H}} \boldsymbol{z}^TK_u(\boldsymbol{\phi}^u, \boldsymbol{\phi}^u)^{-1}\boldsymbol{z} + \boldsymbol{\rho}^TK_m({\boldsymbol{\phi}}^m, {\boldsymbol{\phi}}^m)^{-1}\boldsymbol{\rho} + \boldsymbol{v}^T{K}_{\boldsymbol{V}}({\boldsymbol{\phi}}^{\boldsymbol{V}}, {\boldsymbol{\phi}}^{\boldsymbol{V}})^{-1}\boldsymbol{v} + |\overline{H}|^2 \\
\quad\quad\quad\quad + \frac{1}{\gamma^2} |\boldsymbol{\rho}^{(2)} - \boldsymbol{m}^o|^2 + |\Sigma^{-1}(\boldsymbol{v}^{(2)} - \boldsymbol{V}^o)|^2+\beta \bigg|\frac{1}{M_\Omega}\sum\limits_{i=1}^{M_\Omega}z^{(1), \Omega}_i\bigg|^2+\beta\bigg|\frac{1}{M_\Omega}\sum\limits_{i=1}^{M_\Omega}\rho^{(1), \Omega}_i-1\bigg|^2\\
\quad\quad\quad\quad + \alpha|G(\boldsymbol{z}^{(2)}, \boldsymbol{\rho}^{(3)}, \overline{H}, \boldsymbol{v}^{(1)})|^2.
\end{cases}
\end{align} 
The problem \eqref{OptGPProb_Cpttl_f_relax} is the foundation of the GP method for solving \eqref{OptGPProb}. 
We use the Gauss--Newton method to solve \eqref{OptGPProb_Cpttl_f_relax} (see Section 3 of  \cite{chen2021solving}).

\begin{remark}
	\label{rmk1dnugget}
	Generally, for a kernel $K$ and a vector $\boldsymbol{\phi}$ of linear operators in the dual space of RKHS associated with $K$,  the Gram matrix $K(\boldsymbol{\phi}, \boldsymbol{\phi})$ is ill-conditioned. To compute $K(\boldsymbol{\phi}, \boldsymbol{\phi})^{-1}$, we perform the  Cholesky decomposition on $K(\boldsymbol{\phi}, \boldsymbol{\phi})+\eta R$, where $\eta>0$ is a  chosen regularization  constant, and $R$ is a block diagonal nugget constructed using the approach introduced in \cite{chen2021solving}. In the numerical experiments, we precompute the Cholesky decomposition of  $K(\boldsymbol{\phi}, \boldsymbol{\phi})+\eta R$ and store Cholesky factors for further uses. 
\end{remark} 

\subsection{A Concrete Example}
\label{sub:concex}
To illustrate the general framework described in Subsection \ref{sub:OptProb}, we apply our method to solve a one-dimensional stationary MFG, which is analogous to the examples in Section \ref{secNumercialExpe}. 
Let $\mathbb{T}$ be the one-dimensional torus and be characterized by $[0, 1)$.  
Let $V^*:\mathbb{T}\mapsto \mathbb{R}$ be a smooth solution. We consider the following stationary MFG
\begin{align}
	\label{1dsmfg}
	\begin{cases}
		\frac{(u_x^*)^2}{2}+V^*(x) = (m^*)^2 + \overline{H}^*, & \text{on}\ \mathbb{T},\\
		-(m^*(u_x^*))_x = 0, & \text{on}\ \mathbb{T},\\
		\int_{\mathbb{T}}m^* \dif x = 1, \int_{\mathbb{T}}u^* \dif x = 0, &
	\end{cases}
\end{align}
where $u^*$ is the value function, $m^*$ presents the probability density of the population, and $\overline{H}^*$ is a real number. 
We note that given $V^*$, \eqref{1dsmfg} admits a smooth solution $(u^*, m^*, \overline{H}^*)$ such that $u^*=0$, $m^*=\sqrt{V^* - \overline{H}^*}$, and $\overline{H}^*$ is a constant satisfying $\int_{\mathbb{T}}\sqrt{V^* - \overline{H}^*}\dif x = 1$. In this subsection, our goal is to recover $u^*, m^*, V^*$, and $\overline{H}^*$ using the framework presented in the previous section under noisy observations of  $m^*$ and $V^*$, . More precisely, let $\{x_i^o\}_{i=1}^{N_m}$ be a collection of observation points for $m^*$ and let $\{x_i^{o,1}\}_{i=1}^{N_v^1}$ be another set of observation points for $V^*$. Let $\boldsymbol{m}^o$ be a noisy observation of $(m^*(x_i^o))_{i=1}^{N_m}$ such that ${m}^o_i - m^*(x_i^o)\sim \mathcal{N}(0, \gamma^2)$. Meanwhile, let $\boldsymbol{V}^o$ be a noisy observation of $V^*$ at  $\{x_i^{o,1}\}_{i=1}^{N_v^1}$  such that $V^*(\boldsymbol{x}^{o,1}) - \boldsymbol{V}^o \sim \mathcal{N}(0, \Sigma)$, where $\Sigma$ is the covariance matrix for noises. Let $\{x_i\}_{i=1}^M$ be a collection of collocation points on $\mathbb{T}$. Let $\mathcal{U}$, $\mathcal{M}$, and $\mathcal{V}$ be the RKHSs associated with positive definite kernels $K_u$, $K_m$, and $K_{{V}}$, respectively. Then, we recover $(u^*, m^*, V^*, \overline{H}^*)$  by a minimizer of the following optimal recovery problem
\begin{align}
	\label{1ssmfgoptpk}
	\begin{cases}
		\min\limits_{(u, m , V, \overline{H})\in \mathcal{U}\times\mathcal{M}\times\mathcal{V}\times\mathbb{R}} \|u\|_{\mathcal{U}}^2+\|m\|_{\mathcal{M}}^2 + \|V\|_{\mathcal{V}}^2+|\overline{H}|^2 + \frac{1}{\gamma^2}|[\boldsymbol{\delta}^o, m] - \boldsymbol{m}^o|^2 + |\Sigma^{-1}( \boldsymbol{V}(\boldsymbol{x}^{o,1}) - \boldsymbol{V}^o)|^2\\
  \quad\quad\quad\quad\quad\quad+ \beta\bigg|\frac{1}{M}\sum\limits_{i=1}^{M}m(x_i)-1\bigg|^2+\beta\bigg|\frac{1}{M}\sum\limits_{i=1}^{M}u(x_i)\bigg|^2\\
		\text{s.t.}\  \frac{u_x^2(x_i)}{2} + V(x_i)= m^2(x_i) + \overline{H}, \quad\quad\quad\forall i = 1,\dots,M,\\
		\quad \ \  m_{x}(x_i)u_x(x_i) + m(x_i)u_{xx}(x_i)  = 0, \quad\forall i = 1,\dots,M,
	\end{cases}
\end{align}
where $\beta>0$ is a penalization parameter. 
Let $K_u$, $K_m$, and $K_{{V}}$ be kernels associated with RKHSs $\mathcal{U}$, $\mathcal{M}$, and ${\mathcal{V}}$ such that $u^*\in \mathcal{U}$, $m^*\in \mathcal{M}$, and ${V}^*\in {\mathcal{V}}$. Then, the constraints of \eqref{1ssmfgoptpk} are well-defined. 
Let $(u^\dagger, m^\dagger, \overline{H}^\dagger)$ be a minimizer to \eqref{1ssmfgoptpk}, whose existence is given by Theorem \ref{exi11} later. Then, \eqref{1ssmfgoptpk} can be viewed as MAP estimators of  GPs  conditioned on the MFG at the sample points $\{x_i\}_{i=1}^M$ while matching the observations. The conditioned GPs are not Gaussian since the constraints in \eqref{1ssmfgoptpk} are nonlinear.

Next, we rewrite \eqref{1ssmfgoptpk} as a two-level optimization problem
\begin{align}
	\label{1ssmfgopt2lvpk}
	\begin{cases}
		\min\limits_{\boldsymbol{z},\boldsymbol{\rho},\boldsymbol{v}, \overline{H}}\begin{cases}
			\min\limits_{(u, m, V)\in \mathcal{U}\times\mathcal{M}\times\mathcal{V}} \|u\|_{\mathcal{U}}^2+\|m\|_{\mathcal{M}}^2 + \|{V}\|_{\mathcal{V}}^2\\
			\text{s.t.}\  u(x_i) = z_i^{(1)}, u_x(x_i) = z_i^{(2)}, u_{xx}(x_i) = z_i^{(3)}, \quad\quad \forall i=1,\dots,M,\\
			\quad\ \  m(x_i) = \rho_i^{(1)},  m_x(x_i) = \rho_i^{(2)}, \quad\quad\quad\quad\quad\quad\quad\quad \forall i=1,\dots,M,\\
            \quad\ \  m(x_i^o) = \rho_i^{(3)}, \quad\quad\quad\quad\quad\quad\quad\quad\quad\quad\quad\quad\quad\quad\ \forall i = 1, \dots, N_m,\\
            \quad\ \ V(x_i) = v_i^{(1)}, \quad\quad\quad\quad\quad\quad\quad\quad\quad\quad\quad\quad\quad\quad\ \  \forall i = 1, \dots, M,\\
            \quad\ \  V(x_i^{o,1}) = v_i^{(2)},  \quad\quad\quad\quad\quad\quad\quad\quad\quad\quad\quad\quad\quad\quad  \forall i = 1, \dots, N_v^1,
		\end{cases}\\
		\quad\quad+|\overline{H}|^2 + \frac{1}{\gamma^2}|\boldsymbol{\rho}^{(3)} - \boldsymbol{m}^o|^2 + |\Sigma^{-1}(\boldsymbol{v} - \boldsymbol{V}^o)|^2 + \beta\bigg|\frac{1}{M}\sum\limits_{i=1}^{M}\rho_i^{(1)}- 1\bigg|^2+\beta\bigg|\frac{1}{M}\sum\limits_{i=1}^{M}z_i^{(1)}\bigg|^2\\
		\text{s.t.}\ \frac{(z^{(2)}_i)^2}{2} + v_i^{(1)}=(\rho_i^{(1)})^2 + \overline{H},\quad \forall i=1,\dots,M,\\
		\quad\ \   \rho_i^{(2)}z_i^{(2)} + \rho_i^{(1)}z_i^{(3)}  = 0, \quad\quad\quad\,\, \forall i=1,\dots,M,
	\end{cases}
\end{align}
where $\boldsymbol{z}=(z^{(1)}_1,\dots, z^{(1)}_M, z^{(2)}_1,\dots, z^{(2)}_M, z^{(3)}_1,\dots, z^{(3)}_M)$, $\boldsymbol{\rho}=(\rho^{(1)}_1,\dots, \rho^{(1)}_M, \rho^{(2)}_1,\dots, \rho^{(2)}_M, \rho_{1}^{(3)}, \dots, \rho_{N_m}^{(3)})$, and $\boldsymbol{v}=(v^{(1)}_1,\dots, v^{(1)}_M, v_{1}^{(2)}, \dots, v_{N_v^1}^{(2)})$.  
Let $\delta_x$ be the Dirac delta function concentrated at $x$. We define $\boldsymbol{\delta}=(\delta_{x_i})_{i=1}^M$, $\boldsymbol{\delta}^o=(\delta_{x_i^o})_{i=1}^{N_m}$, and $\boldsymbol{\delta}^{o,1}=(\delta_{x_i^{o,1}})_{i=1}^{N_v^1}$. Let $\boldsymbol{\phi}^u=(\boldsymbol{\delta}, \boldsymbol{\delta}\circ \partial_x, \boldsymbol{\delta}\circ \partial_{xx})$, $\boldsymbol{\phi}^m = (\boldsymbol{\delta}, \boldsymbol{\delta}\circ \partial_x, \boldsymbol{\delta}^o)$, and $\boldsymbol{\phi}^V=(\boldsymbol{\delta}, \boldsymbol{\delta}^{o,1})$. Let $(u^\dagger, m^\dagger, V^\dagger)$ be the minimizer of the first level optimization problem in \eqref{1ssmfgopt2lvpk}. Then,  by the representer formulas in \eqref{goprs} (see \cite[Sec. 17.8]{owhadi2019operator}),  we get
\begin{align*}
	\begin{cases}
		u^\dagger(x)= K_u(x,\boldsymbol{\phi}^u) K_u(\boldsymbol{\phi}^u, \boldsymbol{\phi}^u)^{-1}\boldsymbol{z},\\
  m^\dagger(x)= K_m(x,\boldsymbol{\phi}^m) K_m(\boldsymbol{\phi}^m, \boldsymbol{\phi}^m)^{-1}\boldsymbol{\rho},\\
  V^\dagger(x)= K_V(x,\boldsymbol{\phi}^V) K_V(\boldsymbol{\phi}^V, \boldsymbol{\phi}^V)^{-1}\boldsymbol{v}.
	\end{cases}
\end{align*}
Thus, we have
\begin{align*}
	\begin{cases}
		\|u^\dagger\|_{\mathcal{U}}^2=\boldsymbol{z}^TK_u(\boldsymbol{\phi}^u, \boldsymbol{\phi}^u)^{-1}\boldsymbol{z},\\
  \|m^\dagger\|_{\mathcal{M}}^2=\boldsymbol{\rho}^TK_m(\boldsymbol{\phi}^m, \boldsymbol{\phi}^m)^{-1}\boldsymbol{\rho},\\
  \|V^\dagger\|_{\mathcal{V}}^2=\boldsymbol{v}^TK_V(\boldsymbol{\phi}^V, \boldsymbol{\phi}^V)^{-1}\boldsymbol{v}. 
	\end{cases}
\end{align*}
Hence, we can formulate \eqref{1ssmfgopt2lvpk} as a finite-dimensional optimization problem
\begin{align}
	\label{1ssmfgoptfnpk}
	\begin{cases}
		\min\limits_{\boldsymbol{z},\boldsymbol{\rho},\boldsymbol{v}, \overline{H}} \boldsymbol{z}^TK_u(\boldsymbol{\phi}^u, \boldsymbol{\phi}^u)^{-1}\boldsymbol{z} + \boldsymbol{\rho}^TK_m(\boldsymbol{\phi}^m, \boldsymbol{\phi}^m)^{-1}\boldsymbol{\rho} + \boldsymbol{v}^TK_V(\boldsymbol{\phi}^V, \boldsymbol{\phi}^V)^{-1}\boldsymbol{v} \\
		\quad\quad+|\overline{H}|^2 + \frac{1}{\gamma^2}|\boldsymbol{\rho}^{(3)} - \boldsymbol{m}^o|^2 + |\Sigma^{-1}(\boldsymbol{v} - \boldsymbol{V}^o)|^2 + \beta\bigg|\frac{1}{M}\sum\limits_{i=1}^{M}\rho_i^{(1)}- 1\bigg|^2+\beta\bigg|\frac{1}{M}\sum\limits_{i=1}^{M}z_i^{(1)}\bigg|^2\\
		\text{s.t.}\ \frac{(z^{(2)}_i)^2}{2} + v_i^{(1)}=(\rho_i^{(1)})^2 + \overline{H},\forall i=1,\dots,M,\\
		\quad\ \   \rho_i^{(2)}z_i^{(2)} + \rho_i^{(1)}z_i^{(3)}  = 0, \forall i=1,\dots,M,
	\end{cases}
\end{align}
To deal with the nonlinear constraints in \eqref{1ssmfgoptfnpk}, we introduce a prescribed penalization parameter $\alpha>0$ and consider the following relaxation 
\begin{align}
	\label{relax1D}
	\begin{cases}
		\min\limits_{\boldsymbol{z},\boldsymbol{\rho},\boldsymbol{v}, \overline{H}} \boldsymbol{z}^TK_u(\boldsymbol{\phi}^u, \boldsymbol{\phi}^u)^{-1}\boldsymbol{z} + \boldsymbol{\rho}^TK_m(\boldsymbol{\phi}^m, \boldsymbol{\phi}^m)^{-1}\boldsymbol{\rho} + \boldsymbol{v}^TK_V(\boldsymbol{\phi}^V, \boldsymbol{\phi}^V)^{-1}\boldsymbol{v} \\
		\quad\quad+|\overline{H}|^2 + \frac{1}{\gamma^2}|\boldsymbol{\rho}^{(3)} - \boldsymbol{m}^o|^2 + |\Sigma^{-1}(\boldsymbol{v} - \boldsymbol{V}^o)|^2 + \beta\bigg|\frac{1}{M}\sum\limits_{i=1}^{M}\rho_i^{(1)}- 1\bigg|^2+\beta\bigg|\frac{1}{M}\sum\limits_{i=1}^{M}z_i^{(1)}\bigg|^2\\
		\quad\quad + \alpha\sum_{i=1}^M|\frac{(z^{(2)}_i)^2}{2} + v_i^{(1)}-(\rho_i^{(1)})^2 - \overline{H}|^2 + \alpha \sum_{i=1}^M|\rho_i^{(2)}z_i^{(2)} + \rho_i^{(1)}z_i^{(3)}|^2.	\end{cases}
\end{align}
In numerical implementations, we  use the Gauss--Newton method to solve \eqref{relax1D} (see also Section 3 of  \cite{chen2021solving}).

\subsection{Existence of Minimizers}
\label{sub:eximi}

By adapting the proof of  Theorem 3.3 in \cite{mou2022numerical}, the next theorem shows that  there exists a solution to \eqref{OptGPProb}. Hence, we establish the solvability of the optimization problem  in  \eqref{OptGPProb}.
\begin{theorem}
	\label{exi11}
Let $\mathcal{U}$, $\mathcal{M}$, and $\boldsymbol{\mathcal{V}}$ be RKHSs associated with kernels $K_u$, $K_m$, and $K_{\boldsymbol{V}}$. Assume that $(u^*, m^*, \overline{H}^*, \boldsymbol{V}^*)$ satisfies the MFG in  \eqref{gsmfgs} such that $u^*\in \mathcal{U}$, $m^*\in \mathcal{M}$, $\overline{H}^*\in \mathbb{R}$, and $\boldsymbol{V}^*\in \boldsymbol{\mathcal{V}}$. 
 Define the vectors of linear operators $\boldsymbol{\phi}^u$, $\boldsymbol{\phi}^m$, and $\boldsymbol{\phi}^{\boldsymbol{V}}$ as outlined in Subsection \ref{sub:OptProb}. 
	Under Assumption \ref{hyp_pb}, the minimization problem in \eqref{OptGPProb} admits a minimizer $(u^\dagger, m^\dagger, \overline{H}^\dagger, \boldsymbol{V}^\dagger)$ such that 
	\begin{align}
	\label{exp1df}
	\begin{cases}
u^\dagger(x)=K_u(x,\boldsymbol{\phi}^u)K_u(\boldsymbol{\phi}^u, \boldsymbol{\phi}^u)^{-1}\boldsymbol{z}^\dagger,\\
		m^\dagger(x)=K_m(x,{\boldsymbol{\phi}}^m) K({\boldsymbol{\phi}}^m, {\boldsymbol{\phi}}^m)^{-1}\boldsymbol{\rho}^\dagger,\\
		\boldsymbol{V}^\dagger(x)= {K}_{\boldsymbol{V}}(x,{\boldsymbol{\phi}}^{\boldsymbol{V}}) {K}_{\boldsymbol{V}}({\boldsymbol{\phi}}^{\boldsymbol{V}},  {\boldsymbol{\phi}}^{\boldsymbol{V}})^{-1}\boldsymbol{v}^\dagger. 
	\end{cases}
\end{align}
	where $(\boldsymbol{z}^\dagger, \boldsymbol{\rho}^\dagger, \overline{H}^\dagger, \boldsymbol{V}^\dagger)$ is a minimizer of  \eqref{OptGPProb_Cpttl_fin}.   
\end{theorem}
\begin{proof}
	By the above arguments, \eqref{OptGPProb} is equivalent to \eqref{OptGPProb_Cpttl_fin}. Hence, the key is to prove the existence of a minimizer to \eqref{OptGPProb_Cpttl_fin}. 
	The argument here is similar to the proof of Theorem 3.3 in \cite{mou2022numerical}. Let $(u^*, m^*, \overline{H}^*, \boldsymbol{V}^*)$ be the solution to \eqref{gsmfgs}. We define the tuple $(\boldsymbol{z}_*, \boldsymbol{\rho}_*, \boldsymbol{v}_*)$ such that $\boldsymbol{z}_* = [\boldsymbol{\phi}^u, u]$, $\boldsymbol{\rho}_* = [\boldsymbol{\phi}^m, m]$, and $\boldsymbol{v}_* = [\boldsymbol{\phi}^{\boldsymbol{V}}, \boldsymbol{V}]$. Then, $(\boldsymbol{z}_*, \boldsymbol{\rho}_*, \overline{H}^*, \boldsymbol{v}_*)$ satisfies the constraints in \eqref{OptGPProb_Cpttl_fin}. For a little abuse of notations, we denote by $G(\boldsymbol{z}, \boldsymbol{\rho},\ \overline{H}, \boldsymbol{v})=0$ the constraints in \eqref{OptGPProb_Cpttl_fin} and define
	\begin{align*}
		\begin{split}
			\mathcal{C}_1=
			&\bigg\{(\boldsymbol{z}, \boldsymbol{\rho}, \overline{H}, \boldsymbol{v})| \boldsymbol{z}^TK_u(\boldsymbol{\phi}^u, \boldsymbol{\phi}^u)^{-1}\boldsymbol{z} + \boldsymbol{\rho}^TK_m({\boldsymbol{\phi}}^m, {\boldsymbol{\phi}}^m)^{-1}\boldsymbol{\rho} + \boldsymbol{v}^T{K}_{\boldsymbol{V}}({\boldsymbol{\phi}}^{\boldsymbol{V}}, {\boldsymbol{\phi}}^{\boldsymbol{V}})^{-1}\boldsymbol{v} + |\overline{H}|^2 \\
&\quad\quad\quad\quad + \frac{1}{\gamma^2} |\boldsymbol{\rho}^{(2)} - \boldsymbol{m}^o|^2 + |\Sigma^{-1}(\boldsymbol{v}^{(2)} - \boldsymbol{V}^o)|^2+\beta \bigg|\frac{1}{M_\Omega}\sum\limits_{i=1}^{M_\Omega}z^{(1), \Omega}_i\bigg|^2+\beta\bigg|\frac{1}{M_\Omega}\sum\limits_{i=1}^{M_\Omega}\rho^{(1), \Omega}_i-1\bigg|^2\\
			&\quad\quad\quad\quad\leq \boldsymbol{z}^T_*K_u(\boldsymbol{\phi}^u, \boldsymbol{\phi}^u)^{-1}\boldsymbol{z}_* + \boldsymbol{\rho}^T_*K_m({\boldsymbol{\phi}}^m, {\boldsymbol{\phi}}^m)^{-1}\boldsymbol{\rho}_* + \boldsymbol{v}^T_*{K}_{\boldsymbol{V}}({\boldsymbol{\phi}}^{\boldsymbol{V}}, {\boldsymbol{\phi}}^{\boldsymbol{V}})^{-1}\boldsymbol{v}_* + |\overline{H}^*|^2 \\
&\quad\quad\quad\quad + \frac{1}{\gamma^2} |\boldsymbol{\rho}^{(2)}_* - \boldsymbol{m}^o|^2 + |\Sigma^{-1}(\boldsymbol{v}^{(2)}_* - \boldsymbol{V}^o)|^2+\beta \bigg|\frac{1}{M_\Omega}\sum\limits_{i=1}^{M_\Omega}z^{(1), \Omega}_{*, i}\bigg|^2+\beta\bigg|\frac{1}{M_\Omega}\sum\limits_{i=1}^{M_\Omega}\rho^{(1), \Omega}_{*, i}-1\bigg|^2\bigg\}.
		\end{split}
	\end{align*} Then, \eqref{OptGPProb_Cpttl_fin} is equivalent to 
	\begin{align}
		\label{mfg1dform2}
		\begin{cases}
			\min\limits_{\boldsymbol{z},\boldsymbol{\rho}, \overline{H}, \boldsymbol{v}} \boldsymbol{z}^TK_u(\boldsymbol{\phi}^u, \boldsymbol{\phi}^u)^{-1}\boldsymbol{z} + \boldsymbol{\rho}^TK_m({\boldsymbol{\phi}}^m, {\boldsymbol{\phi}}^m)^{-1}\boldsymbol{\rho} + \boldsymbol{v}^T{K}_{\boldsymbol{V}}({\boldsymbol{\phi}}^{\boldsymbol{V}}, {\boldsymbol{\phi}}^{\boldsymbol{V}})^{-1}\boldsymbol{v} + |\overline{H}|^2 \\
\quad\quad\quad\quad + \frac{1}{\gamma^2} |\boldsymbol{\rho}^{(2)} - \boldsymbol{m}^o|^2 + |\Sigma^{-1}(\boldsymbol{v}^{(2)} - \boldsymbol{V}^o)|^2+\beta \bigg|\frac{1}{M_\Omega}\sum\limits_{i=1}^{M_\Omega}z^{(1), \Omega}_i\bigg|^2+\beta\bigg|\frac{1}{M_\Omega}\sum\limits_{i=1}^{M_\Omega}\rho^{(1), \Omega}_i-1\bigg|^2\\
			\text{s.t.}\  (\boldsymbol{z}, \boldsymbol{\rho}, \overline{H}, \boldsymbol{v})\in \mathcal{C}:=\bigg\{(\boldsymbol{z}, \boldsymbol{\rho}, \overline{H}, \boldsymbol{v})|G(\boldsymbol{z}, \boldsymbol{\rho}, \overline{H}, \boldsymbol{v})=0\bigg\}\cap \mathcal{C}_1.
		\end{cases}
	\end{align}
	By the continuity of $G$ in $(\boldsymbol{z}, \boldsymbol{\rho}, \overline{H}, \boldsymbol{v})$ and the fact that $(\boldsymbol{z}_*, \boldsymbol{\rho}_*, \overline{H}^*, \boldsymbol{v}_*)\in \mathcal{C}$, $\mathcal{C}$ is compact and non-empty. Hence, the objective function \eqref{mfg1dform2} achieves a minimum on $\mathcal{C}$. Thus, \eqref{OptGPProb_Cpttl_fin} admits a minimizer. We conclude \eqref{exp1df} by \eqref{1drepre}. 
\end{proof}
\subsection{A Probabilistic Perspective}
\label{sub:probpersp}
Here, we give a probabilistic interpretation of the GP method  for the MFG inverse problems. First, we put independent Gaussian priors on the solution  $(u^*, m^*, \overline{H}^*, \boldsymbol{V}^*)$ to \eqref{gsmfgs} such that $u^*\sim\mathcal{N}(0, {K}_u)$, $m^*\sim\mathcal{N}(0, {K}_m)$, $\overline{H}^*\sim\mathcal{N}(0, 1
)$, and $V_i^*\sim\mathcal{N}(0, {K}_{V_i})$,  where $K_u$, $K_m$, $K_{V_i}$ are kernels associated with the RKHSs $\mathcal{U}$, $\mathcal{M}$, and $\mathcal{V}_i$, $i=1,\dots, N_v$. Let $\boldsymbol{\phi}^u$, $\boldsymbol{\phi}^m$, and  $\boldsymbol{\phi}^{\boldsymbol{V}}$ be the vectors of linear operators defined in Subsection  \ref{sub:OptProb}. Let $\boldsymbol{z}$, $\boldsymbol{\rho}$, and $\boldsymbol{v}$ be observations defined as  
\begin{align*}
\boldsymbol{z} = [\boldsymbol{\phi}^u, u^*], \boldsymbol{\rho} = [\boldsymbol{\phi}^m, m^*], \text{ and } \boldsymbol{v} = [\boldsymbol{\phi}^{\boldsymbol{V}}, \boldsymbol{V}^*], 
\end{align*}
Thus, $\boldsymbol{z}\sim \mathcal{N}(0, K_u(\boldsymbol{\phi}^u, \boldsymbol{\phi}^u))$, $\boldsymbol{\rho}\sim \mathcal{N}(0, K_m(\boldsymbol{\phi}^m, \boldsymbol{\phi}^m))$ , $\boldsymbol{v}\sim \mathcal{N}(0, K_{\boldsymbol{V}}(\boldsymbol{\phi}^{\boldsymbol{V}}, \boldsymbol{\phi}^{\boldsymbol{V}}))$. Moreover, the probability densities of $\boldsymbol{z}$, $\boldsymbol{\rho}$, and $\boldsymbol{V}$ satisfy
\begin{align*}
p(\boldsymbol{z})=\frac{1}{C_{\boldsymbol{z}}\sqrt{\operatorname{det}(K(\boldsymbol{\phi}^u, \boldsymbol{\phi}^u))}}e^{-\frac{1}{2}\boldsymbol{z}^TK(\boldsymbol{\phi}^u, \boldsymbol{\phi}^u)^{-1}\boldsymbol{z}},\\
p(\boldsymbol{\rho})=\frac{1}{C_{\boldsymbol{\rho}}\sqrt{\operatorname{det}(K(\boldsymbol{\phi}^m, \boldsymbol{\phi}^m))}}e^{-\frac{1}{2}\boldsymbol{\rho}^TK(\boldsymbol{\phi}^m, \boldsymbol{\phi}^{m})^{-1}\boldsymbol{\rho}},\\
p(\boldsymbol{v})=\frac{1}{C_{\boldsymbol{v}}\sqrt{\operatorname{det}(K(\boldsymbol{\phi}^{\boldsymbol{V}}, \boldsymbol{\phi}^{\boldsymbol{V}}))}}e^{-\frac{1}{2}\boldsymbol{v}^TK(\boldsymbol{\phi}^{\boldsymbol{V}}, \boldsymbol{\phi}^{\boldsymbol{V}})^{-1}\boldsymbol{v}},
 \end{align*}
where $C_{\boldsymbol{z}}$, $C_{\boldsymbol{\rho}}$, and $C_{\boldsymbol{v}}$ are constants guaranteeing the unit mass of the probability density. By assumptions, $\boldsymbol{m}^o$ is a noisy observation of $\boldsymbol{\rho}^{(2)}$, i.e. $\boldsymbol{\epsilon} = \boldsymbol{\rho}^{(2)} - \boldsymbol{m}^0$, $\boldsymbol{\epsilon}\sim\mathcal{N}(0, \gamma^2 I)$, where $\boldsymbol{\epsilon}$ is independent with other random variables. Thus, 
$\gamma^{-1}(\boldsymbol{\rho}^{(2)} - \boldsymbol{m}^o)\sim \mathcal{N}(0, I)$ and $\gamma^{-1}(\boldsymbol{\rho}^{(2)} - \boldsymbol{m}^o)$ is independent with other random variables. Similarly, $\Sigma^{-1}(\boldsymbol{v}^{(2)} - \boldsymbol{V}^o)\sim \mathcal{N}(0, I)$ and is another independent random vector. Then,  solving   Problem \ref{st_prob} by \eqref{OptGPProb_Cpttl_fin} is equivalent to finding $\boldsymbol{z}$, $\boldsymbol{\rho}$, $\overline{\boldsymbol{H}}$, and $\boldsymbol{v}$ with the maximum likelihood while satisfying the PDE system and with penalizations on noisy observations and the means of $\boldsymbol{z}$ and $\boldsymbol{\rho}$, i.e., \eqref{OptGPProb_Cpttl_fin} is equivalent to 
\begin{align*}
\begin{cases}
\min\limits_{\boldsymbol{z}, \boldsymbol{\rho}, \boldsymbol{v}, \overline{H}} 
-\log p(\boldsymbol{z}) - \log p(\boldsymbol{\rho}) - \log p(\overline{H}) - \log p(\boldsymbol{v}) - \log p(\gamma^{-1}(\boldsymbol{\rho}^{(2)}-\boldsymbol{m}^o))\\
\quad\quad\quad\quad - \log p(\Sigma^{-1}(\boldsymbol{v}^{(2)} - \boldsymbol{V}^o))+\beta \bigg|\frac{1}{M_\Omega}\sum\limits_{i=1}^{M_\Omega}z^{(1), \Omega}_i\bigg|^2+\beta\bigg|\frac{1}{M_\Omega}\sum\limits_{i=1}^{M_\Omega}\rho^{(1), \Omega}_i-1\bigg|^2\\
\text{s.t.}\quad G(\boldsymbol{z}^{(2)}, \boldsymbol{\rho}^{(3)}, \overline{H}, \boldsymbol{v}^{(1)}) = 0.
\end{cases}
\end{align*}

\section{An Inverse Problem For Time-dependent MFGs}
\label{secTimDep}
This section considers Problem \ref{main_prob} in settings of time-dependent MFGs and introduces a GP framework for solving the inverse problem. Differing from stationary MFGs, time-dependent MFGs typically don't incorporate mass constraints, a feature naturally ensured by  PDEs. In Subsection \ref{sub:tdprb}, we reformulate Problem \ref{main_prob} for time-dependent MFGs. Subsection \ref{subsec:tdgp} details the GP framework built upon a forward-backward  time-discretization scheme. In contrast to the approach outlined in \cite{chen2021solving}, our methodology is distinctively tailored for time-dependent MFGs involving forward-backward PDEs. In this context, neither a purely forward nor a purely backward Euler scheme is applicable. Instead, we integrate a mixed forward-backward time discretization technique, which is proposed in  \cite{achdou2010mean}, with GPs for the inverse problem.

Another feasible approach for solving the inverse problems of time-dependent MFGs is to treat the temporal variable as an additional spatial variable, applying the methods described in the preceding section to address corresponding challenges. Here, we propose a time-discretized approach, where the time-dependent PDE system is discretized in time and GPs are employed to approximate the unknowns in the entire time-independent PDE system. This discretization offers two main benefits: firstly, it separates the complexities of time and space, simplifying kernel selection based on spatial variability; secondly, it allows for distinct kernel choices at different time slices, catering to the varying spatial regularities of functions over time.

\subsection{An Inverse Problem for Time-dependent MFGs}
\label{sub:tdprb}
Let $\Omega \subset \mathbb{R}^d$. In this section, we discuss the inverse problem for time-dependent MFGs of the following form
\begin{align}
	\label{gtdmfgs}
	\begin{cases}
		-\partial_t u^* +\mathcal{P}_1(u^*, m^*,  \boldsymbol{V}^*)(t, x)= 0, &\forall (t,x) \in (0, T)\times\operatorname{int} \Omega,\\	
        \partial_t m^* +\mathcal{P}_2(u^*, m^*,  \boldsymbol{V}^*)(t, x)= 0, &\forall (t,x) \in  (0, T)\times\operatorname{int} \Omega,\\
		\mathcal{B}(u^*, m^*, \boldsymbol{V}^*)(t, x) = 0, &\forall  (t, x) \in (0, T) \times \partial \Omega,\\
		u^*(T, x) = \phi^*(x), m^*(0, x) = \mu^*(x), &\forall x \in \overline{\Omega}. 
	\end{cases}
\end{align}  
where $\boldsymbol{V}^*$
comprises a set of functions that depict environment configurations, $\phi^*$  represents the terminal cost of agents, and $\mu^*$ is the initial distribution of agents. Given $\boldsymbol{V}^*$, $\mu^*$, and $\phi^*$, $(u^*, m^*)$ solves the time-dependent MFG such that $u^*$ stands for the value function and $m^*$ represents the density of the population of agents. In this context, $\mathcal{P}_1$ and $\mathcal{P}_2$ are two nonlinear operators, while $\mathcal{B}$ stands for a boundary operator. We proceed under the assumption that equation \eqref{gtdmfgs} admits a unique classical solution $(u^*, m^*)$ when the underlying environment configuration $\boldsymbol{V}^*$, the initial distribution $\mu^*$, and the terminal cost $\phi^*$ are given. We detail Problem \ref{main_prob} in the setting of time-dependent MFGs of the form \eqref{gtdmfgs} in the following problem.
\begin{problem}
\label{td_prob}
Consider a set of functions on $\Omega$, denoted by $\boldsymbol{V}^*=\{V^*_i\}_{i=1}^{N_v}$, with $N_v\in \mathbb{N}$. Let ${\phi}^*$ be a function defined over $\overline{\Omega}$, and let $\mu^*$ denote a distribution on $\Omega$. When provided with $\boldsymbol{V}^*$, $\phi^*$, and ${\mu}^*$, \eqref{gtdmfgs} yields a unique classical solution $(u^*, m^*)$. Our objective is to estimate the quantities $u^*$, $m^*$, $\boldsymbol{V}^*$, $\mu^*$, and $\phi^*$ based on  partial noisy observations of $m^*$ and of a subset $\boldsymbol{V}_o^*$ from $\boldsymbol{V}^*$. Specifically, we have

\begin{enumerate}
    \item \textbf{Partial noisy observations of $m^*$}. We possess a set of linear operators $\{\phi^o_l\}_{l=1}^{N_m}$, $N_m\in \mathbb{N}$ at which we have imprecise observations of $m^*$. These observations are denoted by $\boldsymbol{m}^o$. Formally, $\boldsymbol{m}^o = ([\phi_1^o, m^*], \dots, [\phi_{N_m}^o, m^*]) + \boldsymbol{\epsilon}$ where $\boldsymbol{\epsilon}\sim \mathcal{N}(0, \gamma^2 I)$.
    
    \item \textbf{Partial noisy observations of $\boldsymbol{V}^*$}. Denote a subset of $\boldsymbol{V}^*$ by $\boldsymbol{V}_o^*=\{V_j^*\}_{j\in \mathcal{O}_v}$, where $\mathcal{O}_v\subset \{1,\dots, N_v\}$. For this subset, we have observations at specific sets of linear operators, $\{\{\psi_e^{o,j}\}_{e=1}^{N^{j}_{v}}\}_{j\in \mathcal{O}_v}$. The observations are given by
    \begin{align*}
    \boldsymbol{V^o}=(\boldsymbol{V}^o_j)_{j\in \mathcal{O}_v}, \text{ where }\boldsymbol{V}^o_j = ([\psi_1^{o,j}, V_j^*], \dots, [\psi_{N_v}^{o,j}, V_j^*] + \boldsymbol{\epsilon}_j, \boldsymbol{\epsilon}_j \sim \mathcal{N}(0, \gamma_j^2I).
    \end{align*}
\end{enumerate}

\textbf{Inverse Problem Statement}

Suppose that agents play the time-dependent MFG specified in \eqref{gtdmfgs}, we infer $u^*$, $m^*$, $\boldsymbol{V}^*$, $\boldsymbol{\phi}^*$, and $\mu^*$ based on the data $\boldsymbol{m}^o$ and $\boldsymbol{V}^o$.
\end{problem}

\subsection{A Time-discretization Framework}
\label{subsec:tdgp}
In this section, we present a GP framework to solve Problem \ref{td_prob}. This is accomplished by utilizing a time-discretized variant of \eqref{gtdmfgs}, which is given in \eqref{tdgtdmfgs}. The key challenge in \eqref{gtdmfgs} lies in its incorporation of coupled forward-backward PDEs. To effectively handle this, we employ a mixed forward-backward time-discretization approach. This methodology is inspired by and builds upon established strategies found in the MFG literature \cite{achdou2010mean, achdou2012iterative}. The key difference between our discretization given in  \eqref{tdgtdmfgs} and the schemes in \cite{achdou2010mean, achdou2012iterative}
is that we only discretize \eqref{gtdmfgs} in time and approximate solutions in the resulting time-independent PDEs by GPs but they discretize the MFG both in time and space  by using the adjoint structure inherent in MFGs. Therefore, the existence and uniqueness of the solution to the discretized PDE system in \eqref{tdgtdmfgs} can be proved using the same fixed-point argument as in \cite{achdou2010mean}.

Let $\{t_k\}_{k=0}^{N_T}$ be an equiv-distributed grid points on the time interval $[0, T]$ with $t_0 = 0$ and $t_{N_T}=T$. We approximate $u$, $m$, and $\boldsymbol{V}$ at the time grid point by sets of functions $\{u_k\}_{k=0}^{N_T}$, $\{m_k\}_{k=0}^{N_T}$, and $\{\boldsymbol{V}_k\}_{k=0}^{N_T}$ such that $u_k\approx u(t_k)$, $m_k\approx m(t_k)$, and $\boldsymbol{V}_k\approx \boldsymbol{V}(t_k)$ for $k=0, \dots, N_T$. Then, we  discretize \eqref{gtdmfgs} in time and get the following system
\begin{align}
\label{tdgtdmfgs}
\begin{cases}
-\frac{u_{k+1}(x)-u_k(x)}{\Delta t} +\mathcal{P}_1(u_k, m_{k+1},  \boldsymbol{V}_{k})(x)= 0, &\forall x \in \operatorname{int} \Omega, k\in \{0, \dots, N_T - 1\},\\	
        \frac{m_{k+1}(x) - m_k(x)}{\Delta t} +\mathcal{P}_2(u_k, m_{k+1},  \boldsymbol{V}_{k})(x)= 0, &\forall x \in  \operatorname{int} \Omega, k\in \{0, \dots, N_T - 1\},\\
		\mathcal{B}(u_{k}, m_{k+1}, \boldsymbol{V}_k)(x) = 0, &\forall  x \in \partial \Omega, k\in \{0,\dots,  N_T - 1\}, \\
		u_{N_T}(x) = \phi^*(x), m_0(x) = \mu^*(x), &\forall x \in \overline{\Omega}. 
\end{cases}
\end{align}
To simplify the presentation, we assume that the time grid points are rich enough such that  noisy observations of the unknown variables occur only at discrete time slots. If observations are made between these intervals, the issue can be addressed by either employing a more refined time discretization or interpolating between the values of adjacent time slices. Next, we impose a GP prior on each unknown and approximate unknowns by GPs conditioned on PDEs while penalizing the observations. Let $\{x_i\}_{i=1}^M$ be a set of samples such that $x_1,\dots, x_{M_\Omega}\in \operatorname{int}\Omega$ and $x_{M_\Omega + 1},\dots,x_M\in \partial\Omega$ for $1\leq M_\Omega\leq M$.  Let $(\{u_k^*\}_{k=0}^{N_T}, \{m_k^*\}_{k=0}^{N_T}, \{\boldsymbol{V}_k^*\}_{k=0}^{N_T})$ be the solution of \eqref{tdgtdmfgs}. Denote by $\{\mathcal{U}_k\}_{k=0}^{N_T}$ and  $\{\mathcal{M}_k\}_{k=0}^{N_T}$ the collections of RKHSs such that $u_k^*\in \mathcal{U}_k$ and  $m_k^*\in \mathcal{M}_k$ for each $0\leq k\leq N_T$. Meanwhile, let $\boldsymbol{\mathcal{V}}:=\{\{\mathcal{V}_{i,k}\}_{k=1}^{N_T}\}_{i=1}^{N_v}$ be a collection of RKHSs where $V_{i,k}^*\in \mathcal{V}_{i,k}$ and let $\boldsymbol{\mathcal{V}}_k =\otimes_{{i=1}}^{N_v}\mathcal{V}_{i,k}$. Here, we assume that the observation operators for $m^*$ are fixed for all the time slots and collect them into a vector $\boldsymbol{\phi}^o$. Let $\boldsymbol{m}^o_k$ be a noisy observation of $[\boldsymbol{\phi}^o, m_k^*]$ such that $([\boldsymbol{\phi}^o, m_k^*]-\boldsymbol{m}^o_k) \sim \mathcal{N}(0, \gamma_k^2I)$.  Let $\Psi_k$ be a collection of linear functionals in $\boldsymbol{\mathcal{V}}^*_k$, the dual  of $\boldsymbol{\mathcal{V}}_k$, and let $\boldsymbol{V}_k^o$ be a noisy observation of $[\Psi_k, \boldsymbol{V}^*_k]$ such that $([\Psi_k, \boldsymbol{V}^*_k]-\boldsymbol{V}^o_k)\sim \mathcal{N}(0, \Sigma_k)$, where $\Sigma_k$ is the covariance matrix for the noise. Then, to solve Problem \ref{td_prob}, we approximate $(u_k, m_k, \boldsymbol{V}_k)_{k=0}^{N_T}$ by a minimizer of the following optimal recovery problem
\begin{align}
\label{OptTDGPProb}
\begin{cases}
\min\limits_{\overset{(u_k, m_k, \boldsymbol{V}_k)\in \mathcal{U}_k\times\mathcal{M}_k\times \boldsymbol{\mathcal{V}}_k}{k=0, \dots, N_T}}\sum_{k=0}^{N_T}\|u_k\|_{\mathcal{U}_k}^2 + \sum_{k=0}^{N_T}\|m_k\|_{\mathcal{M}_k}^2 + \sum_{k=0}^{N_T}\|\boldsymbol{V}_k\|_{\boldsymbol{\mathcal{V}}_k}^2 \\
\quad\quad\quad\quad\quad\quad\quad\quad+ \sum_{k=0}^{N_T}\frac{1}{\gamma^2_k} |[\boldsymbol{\phi}^o, m_k] - \boldsymbol{m}^o_k|^2
+ \sum_{k=0}^{N_T}|\Sigma^{-1}_k([\Psi_k, \boldsymbol{V}_k] - \boldsymbol{V}^o_k)|^2\\
\operatorname{s.t.}\quad -\frac{u_{k+1}(x_i)-u_k(x_i)}{\Delta t} +\mathcal{P}_1(u_k, m_{k+1},  \boldsymbol{V}_{k})(x_i)= 0,  \forall i \in \{1,\dots, M_\Omega\}, k\in \{0, \dots, N_T - 1\},\\	
        \quad\quad\quad \frac{m_{k+1}(x_i) - m_k(x_i)}{\Delta t} +\mathcal{P}_2(u_k, m_{k+1},  \boldsymbol{V}_{k})(x_i)= 0,   \forall i \in \{1,\dots, M_\Omega\}, k\in \{0, \dots, N_T - 1\},\\
		\quad\quad\quad\mathcal{B}(u_{k}, m_{k+1}, \boldsymbol{V}_k)(x_i) = 0, \,\,\,\quad\quad\quad\quad\quad\quad\quad\quad \forall i \in \{M_{\Omega}+1, \dots, M\}, k\in \{0, \dots, N_T - 1\},\\
		\quad\quad\quad u_{N_T}(x_i) = \phi^*(x_i), m_0(x_i) = \mu^*(x_i), \,\quad\quad\quad\quad \forall i \in \{1,\dots, M\}.
\end{cases}    
\end{align}
Similar arguments to those outlined in the proof of Theorem \ref{exi11} can be applied to establish the existence of minimizers in the context of \eqref{OptTDGPProb}.

Next, we reformulate \eqref{OptTDGPProb} into a finite dimensional minimization problem. For this transformation, we adopt specific assumptions regarding the structures of $\mathcal{P}_1$, $\mathcal{P}_2$, and $\mathcal{B}$ as follows.
\begin{hyp}
\label{hyp_tdpb}
Assume that there exists collections of bounded linear operators $\boldsymbol{L}^{u, \Omega}$, $\boldsymbol{L}^{u, \partial\Omega}$, $\boldsymbol{L}^{m, \Omega}$, $\boldsymbol{L}^{m, \partial \Omega}$, $\boldsymbol{L}^{\boldsymbol{{V}}, \Omega}$, and $\boldsymbol{L}^{\boldsymbol{{V}}, \partial\Omega}$,  and continuous nonlinear maps $P_1$, $P_2$,  and $B$ such that
\begin{align*}
\begin{cases}
\mathcal{P}_1(u, m, \boldsymbol{V})(x) = P_1(\boldsymbol{L}^{u, \Omega}(u)(x), \boldsymbol{L}^{m, \Omega}(m)(x), \boldsymbol{L}^{\boldsymbol{V}, \Omega}(\boldsymbol{V})(x)), \forall x\in \operatorname{int}\Omega, \\
\mathcal{P}_2(u, m, \boldsymbol{V})(x) = P_2(\boldsymbol{L}^{u, \Omega}(u)(x), \boldsymbol{L}^{m, \Omega}(m)(x), \boldsymbol{L}^{\boldsymbol{V}, \Omega}(\boldsymbol{V})(x)), \forall x\in \operatorname{int}\Omega, \\
\mathcal{B}(u, m, \boldsymbol{V})(x) = B(\boldsymbol{L}^{u, \partial \Omega}(u)(x), \boldsymbol{L}^{m, \partial \Omega}(m)(x),  \boldsymbol{L}^{\boldsymbol{V}, \partial \Omega}(\boldsymbol{V})(x)), \forall x\in \partial \Omega. 
\end{cases}
\end{align*}
In the following texts, we denote $\boldsymbol{L}^{u} := (\boldsymbol{L}^{u, \Omega}, \boldsymbol{L}^{u, \partial\Omega})$, $\boldsymbol{L}^{m} := (\boldsymbol{L}^{m, \Omega}, \boldsymbol{L}^{m, \partial\Omega})$, and $\boldsymbol{L}^{\boldsymbol{{V}}} := (\boldsymbol{L}^{\boldsymbol{{V}}, \Omega}, \boldsymbol{L}^{\boldsymbol{{V}}, \partial\Omega})$.
\end{hyp}
Under Assumption \ref{hyp_tdpb}, we rewrite \eqref{OptTDGPProb} as
\begin{align}
\label{OptGPtdProb_Cpt}
\begin{cases}
\min\limits_{\overset{(u_k, m_k, \boldsymbol{V}_k)\in \mathcal{U}_k\times\mathcal{M}_k\times \boldsymbol{\mathcal{V}}_k}{k=0, \dots, N_T}}\sum_{k=0}^{N_T}\|u_k\|_{\mathcal{U}_k}^2 + \sum_{k=0}^{N_T}\|m_k\|_{\mathcal{M}_k}^2 + \sum_{k=0}^{N_T}\|\boldsymbol{V}_k\|_{\boldsymbol{\mathcal{V}}_k}^2\\
\quad\quad\quad\quad\quad\quad\quad\quad+ \sum_{k=0}^{N_T}\frac{1}{\gamma^2_k} |[\boldsymbol{\phi}^o, m_k] - \boldsymbol{m}^o_k|^2+ \sum_{k=0}^{N_T}|\Sigma^{-1}_k([\Psi_k, \boldsymbol{V}_k] - \boldsymbol{V}^o_k)|^2\\
\operatorname{s.t.}\quad -\frac{[\delta_{x_i}, u_{k+1} - u_k]}{\Delta t} + P_1([\delta_{x_i}\circ \boldsymbol{L}^{\Omega, u}, u_k], [\delta_{x_i}\circ \boldsymbol{L}^{\Omega, m}, m_{k+1}],  [\delta_{x_i}\circ \boldsymbol{L}^{\Omega, \boldsymbol{V}}, \boldsymbol{V}_k]) = 0, \quad \\
\quad\quad\quad\quad\quad\quad\quad\quad\quad\quad\quad\quad\quad\quad\quad\quad\quad\quad\quad\quad\quad\quad\forall i \in \{1,\dots, M_\Omega\}, k\in \{0, \dots, N_T-1\},\\	
        \quad\quad\quad \frac{[\delta_{x_i}, m_{k+1} - m_k]}{\Delta t} + P_2([\delta_{x_i}\circ \boldsymbol{L}^{\Omega, u}, u_k], [\delta_{x_i}\circ \boldsymbol{L}^{\Omega, m}, m_{k+1}],  [\delta_{x_i}\circ \boldsymbol{L}^{\Omega, \boldsymbol{V}}, \boldsymbol{V}_k])=0,\\
\quad\quad\quad\quad\quad\quad\quad\quad\quad\quad\quad\quad\quad\quad\quad\quad\quad\quad\quad\quad\quad\quad\forall i \in \{1,\dots, M_\Omega\}, k\in \{0, \dots, N_T-1\},\\
		\quad\quad\quad B([\delta_{x_i}\circ \boldsymbol{L}^{\partial\Omega, u}, u], [\delta_{x_i}\circ \boldsymbol{L}^{\partial\Omega, m}, m], [\delta_{x_i}\circ \boldsymbol{L}^{\partial\Omega, \boldsymbol{V}}, \boldsymbol{V}])=0, \\
\quad\quad\quad\quad\quad\quad\quad\quad\quad\quad\quad\quad\quad\quad\quad\quad\quad\quad\quad\quad\quad\quad\forall i \in \{M_\Omega+1,\dots, M\}, k\in \{0, \dots, N_T-1\},\\
		\quad\quad\quad [\delta_{x_i}, u_{N_T}] = \phi^*(x_i), [\delta_{x_i}, m_0] = \mu^*(x_i), \,\,\quad\quad \forall i \in \{1,\dots, M\},
\end{cases}    
\end{align}
Denote $\boldsymbol{\delta}:=(\boldsymbol{\delta}^{\Omega}, \boldsymbol{\delta}^{\partial\Omega})$, where $\boldsymbol{\delta}^\Omega=(\delta_{x_1},\dots, \delta_{x_{M_\Omega}})$ and $\boldsymbol{\delta}^{\partial\Omega}=(\delta_{x_{M_\Omega}+1},\dots, \delta_{x_{M}})$. Let $\widetilde{\boldsymbol{\phi}}^u = \boldsymbol{\delta}\circ \boldsymbol{L}^u$, $\widetilde{\boldsymbol{\phi}}^m = \boldsymbol{\delta}\circ \boldsymbol{L}^m$, and $\widetilde{\boldsymbol{\phi}}^{\boldsymbol{V}} = \boldsymbol{\delta}\circ \boldsymbol{L}^{\boldsymbol{V}}$.
To solve \eqref{OptGPtdProb_Cpt}, we introduce latent variables and rewrite \eqref{OptGPtdProb_Cpt} as 
\begin{align}
\label{OptGPProb_Cpttltd}
\begin{cases}
\min\limits_{\boldsymbol{z}, \boldsymbol{\rho}, \boldsymbol{v}}
\begin{cases}
\min\limits_{\overset{(u_k, m_k, \boldsymbol{V}_k)\in \mathcal{U}_k\times\mathcal{M}_k\times \boldsymbol{\mathcal{V}}_k}{k=0,\dots, N_T}} \sum_{k=0}^{N_T}\|u_k\|_{\mathcal{U}_k}^2 + \sum_{k=0}^{N_T}\|m_k\|_{\mathcal{M}_k}^2 + \sum_{k=0}^{N_T} \|\boldsymbol{V}_k\|_{\boldsymbol{\mathcal{V}}_k}^2\\
\text{s.t.} \quad [\boldsymbol{\delta}^{\Omega}, u_k] = \boldsymbol{z}^{(1), \Omega}_k, [\boldsymbol{\delta}^{\partial\Omega}, u_k] = \boldsymbol{z}^{(1), \partial\Omega}_k,  [\widetilde{\boldsymbol{\phi}}^u, u_k] = \boldsymbol{z}^{(2)}_k, \\ \quad
\quad[\boldsymbol{\delta}^{\Omega}, m_k] = \boldsymbol{\rho}^{(1), \Omega}_k, [\boldsymbol{\delta}^{\partial\Omega}, m_k] = \boldsymbol{\rho}^{(1), \partial\Omega}_k, [\boldsymbol{\phi}^o, m_k] =  \boldsymbol{\rho}^{(2)}_k, 
[\widetilde{\boldsymbol{\phi}}^m, m_k] = \boldsymbol{\rho}^{(3)}_k,\\ \quad\quad[\widetilde{\boldsymbol{\phi}}^{\boldsymbol{V}}, \boldsymbol{V}_k] = \boldsymbol{v}^{(1)}_k,  [\Psi, \boldsymbol{V}_k] = \boldsymbol{v}^{(2)}_k, 
\end{cases}\\
\quad\quad\quad\quad + \sum_{k=0}^{N_T}\frac{1}{\gamma^2_k} |\boldsymbol{\rho}^{(2)}_k - \boldsymbol{m}^o_k|^2 + \sum_{k=0}^{N_T}|\Sigma^{-1}_k(\boldsymbol{v}^{(2)}_k - \boldsymbol{V}^o_k)|^2\\
\operatorname{s.t.}\quad -\frac{z_{k+1, i}^{(1), \Omega} - z_{k, i}^{(1), \Omega}}{\Delta t} + P_1(\boldsymbol{z}_k^{(2), \Omega}, \boldsymbol{\rho}_{k+1}^{(3), \Omega},  \boldsymbol{v}_{k}^{(1), \Omega}) = 0, \quad \forall i \in \{1,\dots, M_\Omega\}, k\in \{0, \dots, N_T-1\},\\	
        \quad\quad\quad \frac{\rho_{k+1, i}^{(1), \Omega} - \rho_{k, i}^{(1), \Omega}}{\Delta t} + P_2(\boldsymbol{z}_k^{(2), \Omega}, \boldsymbol{\rho}_{k+1}^{(3), \Omega},  \boldsymbol{v}_{k}^{(1), \Omega})=0,  \quad\,\, \forall i \in \{1,\dots, M_\Omega\}, k\in \{0, \dots, N_T-1\},\\
		\quad\quad\quad B(\boldsymbol{z}_k^{(2), \partial \Omega}, \boldsymbol{\rho}_{k+1}^{(3), \partial \Omega}, \boldsymbol{v}_k^{(1), \partial \Omega})=0, \quad\quad\quad\quad\quad\quad\,\,\, \forall i \in \{M_{\Omega}+1, \dots, M\}, k\in \{0, \dots, N_T-1\},\\
		\quad\quad\quad z_{N_T, i}^{(1)} = \phi^*(x_i), \rho_{0, i}^{(1)} = \mu^*(x_i), \quad\quad\quad\quad\quad\quad\quad\,\, \forall i \in \{1,\dots, M\},
\end{cases}
\end{align}
where $\boldsymbol{z} = (\boldsymbol{z}_k)_{k=0}^{N_T}, \boldsymbol{z}_k=(\boldsymbol{z}^{(1), \Omega}_k, \boldsymbol{z}^{(1), \partial\Omega}_k, \boldsymbol{z}^{(2)}_k)$, $\boldsymbol{\rho}=(\boldsymbol{\rho}_k)_{k=0}^{N_T}, \boldsymbol{\rho}_k=(\boldsymbol{\rho}^{(1), \Omega}_k, \boldsymbol{\rho}^{(1), \partial\Omega}_k, \boldsymbol{\rho}^{(2)}_k, \boldsymbol{\rho}^{(3)}_k)$, and $\boldsymbol{v}=(\boldsymbol{v}_k)_{k=0}^{N_T}, \boldsymbol{v}_k=(\boldsymbol{v}^{(1)}_k, \boldsymbol{v}^{(2)}_k)$.  Denote $\boldsymbol{\phi}^u := (\boldsymbol{\delta}, \widetilde{\boldsymbol{\phi}}^u)$,  $\boldsymbol{\phi}^m :=(\boldsymbol{\delta}, \boldsymbol{\phi}^o, \widetilde{\boldsymbol{\phi}}^m)$, and $\boldsymbol{\phi}^{\boldsymbol{V}} :=(\widetilde{\boldsymbol{\phi}}^{\boldsymbol{V}}, \Psi)$. 

Let $\{K_{u,k}\}_{k=0}^{N_T}$, $\{K_{m,k}\}_{k=0}^{N_T}$, $\{K_{\boldsymbol{V},k}\}_{k=0}^{N_T}$ be collections of kernels corresponding to the sets of RKHSs $\{\mathcal{U}_k\}_{k=0}^{N_T}$, $\{\mathcal{M}_k\}_{k=0}^{N_T}$, and $\{\boldsymbol{\mathcal{V}}_k\}_{k=0}^{N_T}$ respectively. Let $(u_k^\dagger, m_k^\dagger, \boldsymbol{V}_k^\dagger)_{k=0}^{N_T}$ be a minimizer to the first level minimization problem in \eqref{OptGPProb_Cpttltd}. Then, 
by the representer formulas in \eqref{goprs} and \eqref{reprefmlve} (also see \cite[Sec. 17.8]{owhadi2019operator}), we obtain
\begin{align*}
	\begin{cases}
		u_k^\dagger(x)=K_{u,k}(x,\boldsymbol{\phi}^u)K_{u,k}(\boldsymbol{\phi}^u, \boldsymbol{\phi}^u)^{-1}\boldsymbol{z}_k,\\
		m_k^\dagger(x)=K_{m,k}(x,{\boldsymbol{\phi}}^m) K_{m,k}({\boldsymbol{\phi}}^m, {\boldsymbol{\phi}}^m)^{-1}\boldsymbol{\rho}_k,\\
		\boldsymbol{V}_k^\dagger(x)= {K}_{\boldsymbol{V},k}(x,{\boldsymbol{\phi}}^{\boldsymbol{V}}) {K}_{\boldsymbol{V},k}({\boldsymbol{\phi}}^{\boldsymbol{V}},  {\boldsymbol{\phi}}^{\boldsymbol{V}})^{-1}\boldsymbol{v}_k. 
	\end{cases}
\end{align*}
Thus, we have
\begin{align*}
	\begin{cases}		\|u_k^\dagger\|_{\mathcal{U}_k}^2=\boldsymbol{z}_k^TK_{u, k}(\boldsymbol{\phi}^u, \boldsymbol{\phi}^u)^{-1}\boldsymbol{z}_k,\\
\|m_k^\dagger\|_{\mathcal{M}_k}^2=\boldsymbol{\rho}_k^TK_{m,k}({\boldsymbol{\phi}}^m, {\boldsymbol{\phi}}^m)^{-1}\boldsymbol{\rho}_k, \\
\|\boldsymbol{V}_k^\dagger\|_{\boldsymbol{\mathcal{V}}_k}^2=\boldsymbol{v}_k^T{K}_{\boldsymbol{V},k}({\boldsymbol{\phi}}^{\boldsymbol{V}}, {\boldsymbol{\phi}}^{\boldsymbol{V}})^{-1}\boldsymbol{v}_k.
	\end{cases}
\end{align*}
Hence, we can formulate \eqref{OptGPProb_Cpttltd} as a finite-dimensional optimization problem
\begin{align}
\label{OptGPProb_tdCpttl_fin}
\begin{cases}
\min\limits_{\boldsymbol{z}, \boldsymbol{\rho}, \boldsymbol{v}} \sum_{k=0}\boldsymbol{z}_k^TK_{u,k}(\boldsymbol{\phi}^u, \boldsymbol{\phi}^u)^{-1}\boldsymbol{z}_k + \sum_{k=0}\boldsymbol{\rho}_k^TK_{m,k}(\boldsymbol{\phi}^m, \boldsymbol{\phi}^m)^{-1}\boldsymbol{\rho}_k + \sum_{k=0}\boldsymbol{v}_k^T{K}_{\boldsymbol{V},k}(\boldsymbol{\phi}^{\boldsymbol{V}}, \boldsymbol{\phi}^{\boldsymbol{V}})^{-1}\boldsymbol{v}_k\\
\quad\quad\quad\quad + \sum_{k=0}^{N_T}\frac{1}{\gamma^2_k} |\boldsymbol{\rho}^{(2)}_k - \boldsymbol{m}^o_k|^2 + \sum_{k=0}^{N_T}|\Sigma^{-1}_k(\boldsymbol{v}^{(2)}_k - \boldsymbol{V}^o_k)|^2\\
\operatorname{s.t.}\quad -\frac{z_{k+1, i}^{(1), \Omega} - z_{k, i}^{(1), \Omega}}{\Delta t} + P_1(\boldsymbol{z}_k^{(2), \Omega}, \boldsymbol{\rho}_{k+1}^{(3), \Omega},  \boldsymbol{v}_{k}^{(1), \Omega}) = 0, \quad \forall i \in \{1,\dots, M_\Omega\}, k\in \{0, \dots, N_T-1\},\\	
        \quad\quad\quad \frac{\rho_{k+1, i}^{(1), \Omega} - \rho_{k, i}^{(1), \Omega}}{\Delta t} + P_2(\boldsymbol{z}_k^{(2), \Omega}, \boldsymbol{\rho}_{k+1}^{(3), \Omega},  \boldsymbol{v}_{k}^{(1), \Omega})=0,  \quad\,\, \forall i \in \{1,\dots, M_\Omega\}, k\in \{0, \dots, N_T-1\},\\
		\quad\quad\quad B(\boldsymbol{z}_k^{(2), \partial \Omega}, \boldsymbol{\rho}_{k+1}^{(3), \partial \Omega}, \boldsymbol{v}_k^{(1), \partial \Omega})=0, \quad\quad\quad\quad\quad\quad\,\,\, \forall i \in \{M_{\Omega}+1, \dots, M\}, k\in \{0, \dots, N_T-1\},\\
		\quad\quad\quad z_{N_T, i}^{(1)} = \phi^*(x_i), \rho_{0, i}^{(1)} = \mu^*(x_i), \quad\quad\quad\quad\quad\quad\quad\,\, \forall i \in \{1,\dots, M\}.
\end{cases}
\end{align}
In the numerical experiments presented in Section \ref{secNumercialExpe}, we use the Gauss--Newton method to solve \eqref{OptGPProb_tdCpttl_fin}.

\section{Numerical Experiments}
\label{secNumercialExpe}
In this section, we illustrate the effectiveness of our approach using various examples. All experiments involve reconstructing complete population profiles, agent strategies, and entire environmental configurations from partial, noisy data about populations and environments. Specifically, Subsection \ref{sub:num:smfg0} highlights the determination of the viscosity constant, a critical parameter measuring uncertainty in agents' dynamics. Subsections \ref{sub:num:smfguk} and \ref{sub:num:smfgnlc} are dedicated to identifying the coupling functions in MFGs, which encode agents' interactions. Subsection \ref{sub:num:mspfcd} addresses a misspecification problem, demonstrating our method's robustness by recovering configurations of a MFG with local coupling using a model designed for non-local coupling. Lastly, in Subsection \ref{sub:num:tdmfg}, we apply our methodology to a time-dependent MFG scenario, showcasing its adaptability to dynamic settings.

In all of our experiments, we employ Python in conjunction with the JAX package, which facilitates automatic differentiation. The current setup relies solely on CPU processing. However, notable enhancements in performance can be realized by incorporating accelerated hardware like Graphics Processing Units (GPUs), which are designed to significantly expedite computational tasks.

\subsection{A Stationary MFG}
\label{sub:num:smfg0}
In this example, we study inverse problems related to the following stationary MFG
\begin{align}
\label{exp:smfg}
\begin{cases}
-\nu \Delta u + H(x, \nabla u) = m^2 - V(x) + \overline H, & \forall x \in \mathbb{T}^2,\\
-\nu \Delta m - \div(D_p H(x, \nabla u) m) = 0, & \forall x\in \mathbb{T}^2, \\
\int_{\mathbb{T}^2} u \, dx = 0, \int_{\mathbb{T}^2} m \, dx = 1, &
\end{cases}    
\end{align}
where \( H(x, p) = \frac{1}{2} |p|^2 \). We identify \(\mathbb{T}^2\) by \([-0.5, 0.5) \times [-0.5, 0.5)\), and the true potential \( V(x) = \sin(2\pi y) + \sin(2\pi x) + \cos(4\pi x) \).

Our first goal is to recover  \( (u, m, V, \overline{H}) \) given noisy observations of the population density \( m \) under a fixed viscosity \( \nu \). To verify the accuracy of the recovered quantities, we first solve \eqref{exp:smfg} using the GP method in \cite{mou2022numerical} given $\nu$ and $V$. These numerical solutions, alongside the explicit formulation of $V$, are viewed as reference solutions. This designation underscores their role as benchmarks against which we evaluate the accuracy of the corresponding recovered solutions. For simplicity, we denote the reference solutions by $(u^*, m^*, V^*, \overline{H}^*)$. To solve the inverse problem, we randomly select \( M = 400 \) collocation points from \(\mathbb{T}^2\) and choose the first \( I = 40\) points as observation points. We treat the addition of  the values of $m^*$ at the first $I$ observation points with  independent standard Gaussian noises following $\mathcal{N}(0, \gamma^2I)$, where $\gamma = 10^{-3}$, as the observed data for $m$. We employed periodic kernels for recovering \( m \), \( u \), and \( V \). For each kernel, we set the lengthscale parameter to  \( 1.41 \), and incorporated adaptive diagonal regularization terms  (``nuggets") with parameters \( \eta = 10^{-8} \) to the covariance matrices. Figure \ref{fig:contour1} presents the results of our experiment in recovering \( m \), \( u \), and \( V \) with a viscosity \( \nu = 0.1 \). We notice that the accuracy in recovering the variables $u$ and $m$ is superior to that of $V$ and $\overline{H}$. The reason behind this lies in the nature of the variables and their interdependencies. Given known values of density $m$ and viscosity $\nu$, the solution for $u$ can be precisely determined through the FP equation and integration constraints (see discussions in the introduction section). In contrast, the recovery of the pair $(V, \overline{H})$ faces a uniqueness challenge. This challenge arises because if a pair $(V, \overline{H})$ satisfies the equation in \eqref{exp:smfg}, then so does the pair $(V + C, \overline{H} + C)$ for any constant $C$. This implies that without direct observations of $V$, it's impossible to uniquely determine the values of both $V$ and $\overline{H}$. Furthermore, we observe that the pointwise error in estimating $V - \overline{H}$ is less than that in estimating $V$ alone. The $L_2$ error for the recovered $V$ is $0.311662$, in contrast, the $L_2$ error for the recovered $V-\overline{H}$ is significantly lower, recorded at $0.0035159$. Hence, it is promising to accurately recover $V$ up to a constant when $\nu$ and only noisy data on $m$ are given.

Figure \ref{fig:log-log plot1} illustrates the log-log plot of error versus the number of observed data when \( \nu = 0.1 \). We see that the errors decrease as we increase the number of observations. The results are obtained by averaging over 50 independent realizations.

In scenarios where viscosity \( \nu \) is unknown, we demonstrate the feasibility of simultaneously recovering both \( \nu \) and \( V \) from observations of the population \( m \) and \( V \). Table \ref{tab:recovered viscosity} shows the recovered \( \nu \) values under various true viscosities. We choose \( M = 400 \) sample points, out of which \( I = 20 \) data points of \( m \) and \( V \) are noisily observed with additive noises following $\mathcal{N}(0, \gamma^2I)$, where $\gamma=10^{-3}$.  We see that our GP framework can accurately recover the uncertainty in the MFG system.

\begin{figure}
  \centering
  \begin{subfigure}{0.24\textwidth} 
    \includegraphics[width=\linewidth]{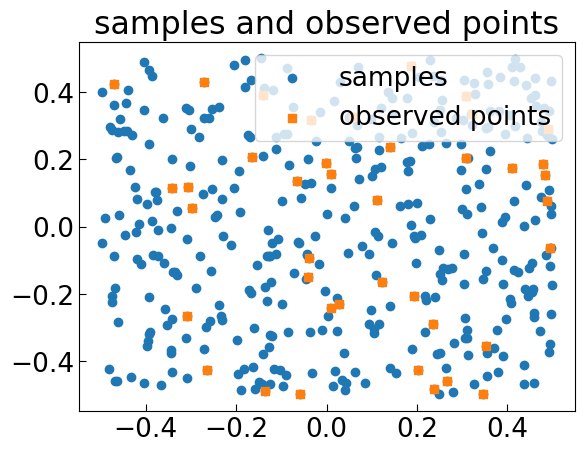}
    \caption{Samples \& observations}
  \end{subfigure}
  \hfil
  \begin{subfigure}{0.24\textwidth}
    \includegraphics[width=\linewidth]{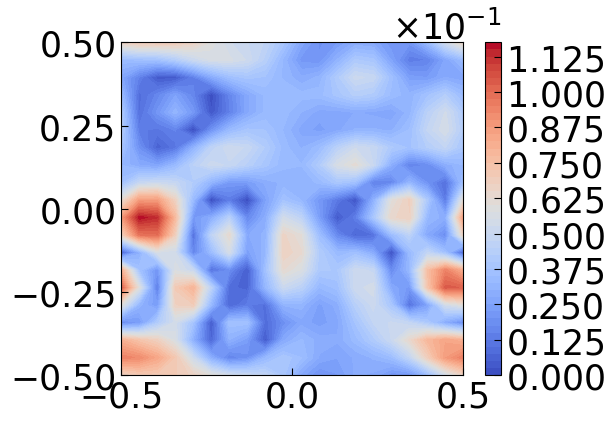}
    \caption{error contour of $V-\bar H$}
  \end{subfigure}
   \vspace{1em} 
 
  \begin{subfigure}{0.24\textwidth}
    \includegraphics[width=\linewidth]{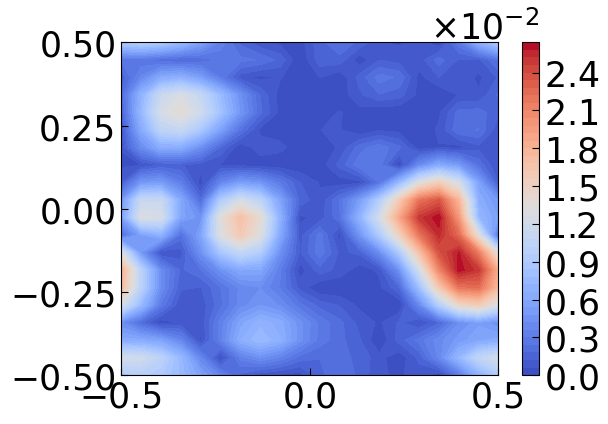}
    \caption{error contour of $m$}
  \end{subfigure}
  \hfill
  \begin{subfigure}{0.24\textwidth}
    \includegraphics[width=\linewidth]{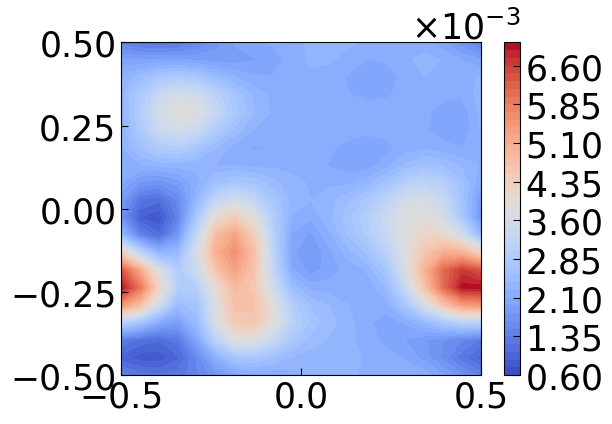}
    \caption{error contour of $u$}
  \end{subfigure}
  \hfill
  \begin{subfigure}{0.24\textwidth}
    \includegraphics[width=\linewidth]{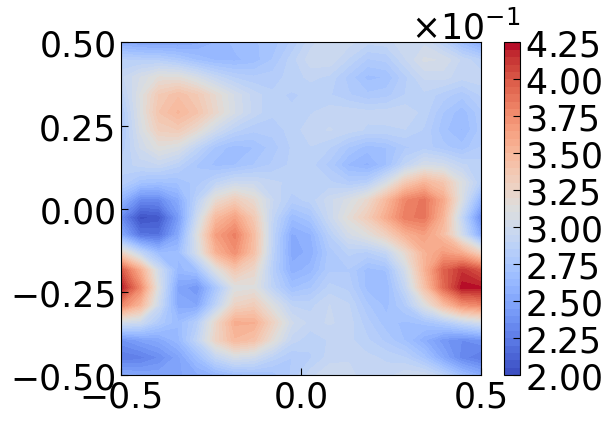}
    \caption{error contour of $V$}
  \end{subfigure}

  \vspace{1em} 
  
  \begin{subfigure}{0.24\textwidth} 
    \includegraphics[width=\linewidth]{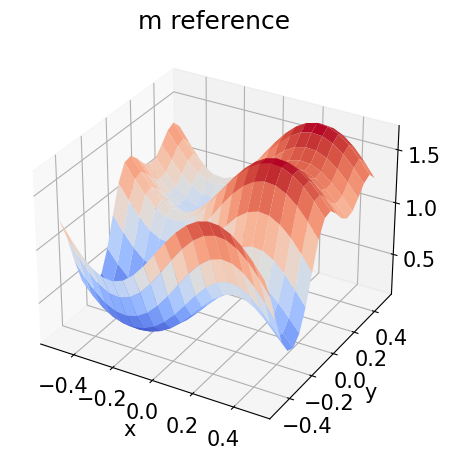}
    \caption{baseline for $m$}
  \end{subfigure}
  \hfill
  \begin{subfigure}{0.24\textwidth}
    \includegraphics[width=\linewidth]{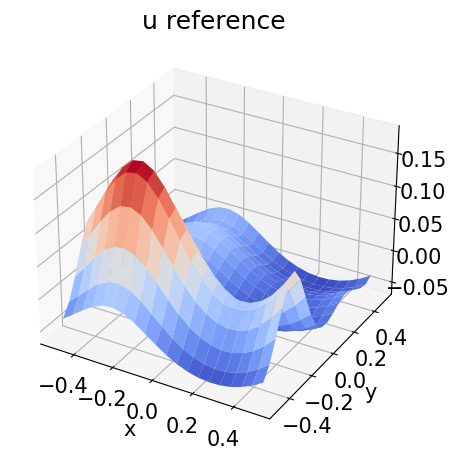}
    \caption{baseline for $u$}
  \end{subfigure}
  \hfill
  \begin{subfigure}{0.24\textwidth}
    \includegraphics[width=\linewidth]{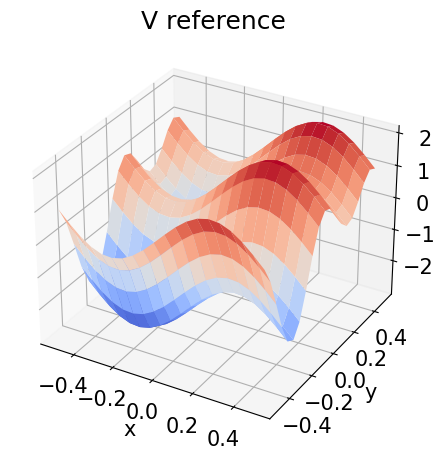}
    \caption{baseline for $V$}
  \end{subfigure}
  
  \vspace{1em} 
  
  \begin{subfigure}{0.24\textwidth} 
    \includegraphics[width=\linewidth]{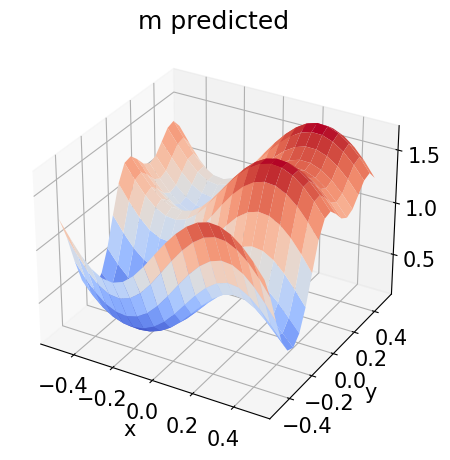}
    \caption{recovered $m$}
  \end{subfigure}
  \hfill
  \begin{subfigure}{0.24\textwidth}
    \includegraphics[width=\linewidth]{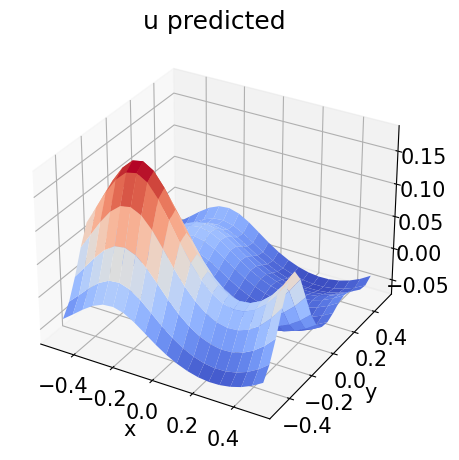}
    \caption{recovered $u$}
  \end{subfigure}
  \hfill
  \begin{subfigure}{0.24\textwidth}
    \includegraphics[width=\linewidth]{example1/V_reference1.png}
    \caption{recovered $V$}
  \end{subfigure}
  
  \caption{Numerical results for the MFG in \eqref{exp:smfg}. $\nu=0.1$, $I = 40$ observed data on $m$, $M=400$ collocation points, ground truth $\overline H=-1.066853$, recovered $\overline H^{\dagger}=-0.752241$. The $L_2$ norm of error of $V$ is $0.311662$, and the $L_2$ norm of error of $V-\overline H$ is $0.035159$.}
  \label{fig:contour1}
\end{figure}

\begin{figure}
  \centering
  \begin{subfigure}{0.3\linewidth}
\includegraphics[width=\linewidth]{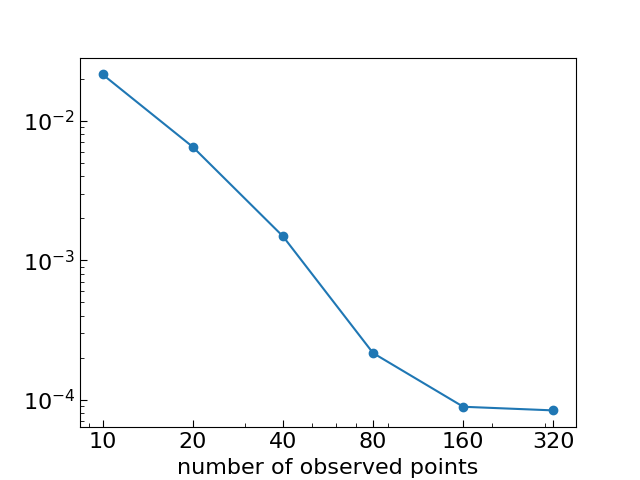}
    \caption{log-log plot of $L_2$ error of $m$}
  \end{subfigure}
  \hfill
  \begin{subfigure}{0.3\linewidth}
    \includegraphics[width=\linewidth]{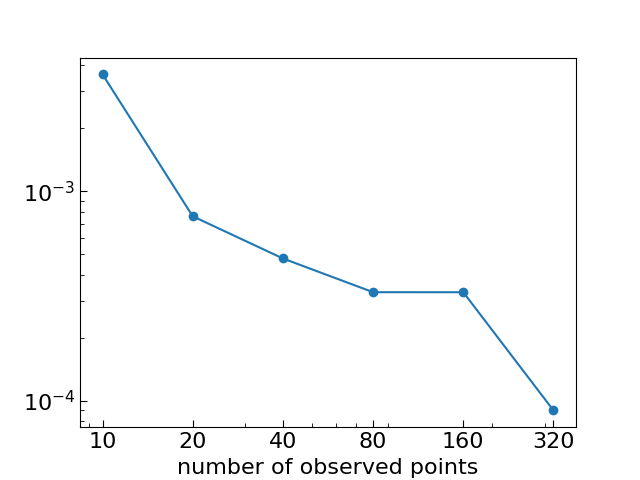}
    \caption{log-log plot of $L_2$ error of $u$}
  \end{subfigure}
  \hfill
  \begin{subfigure}{0.3\linewidth}
    \includegraphics[width=\linewidth]{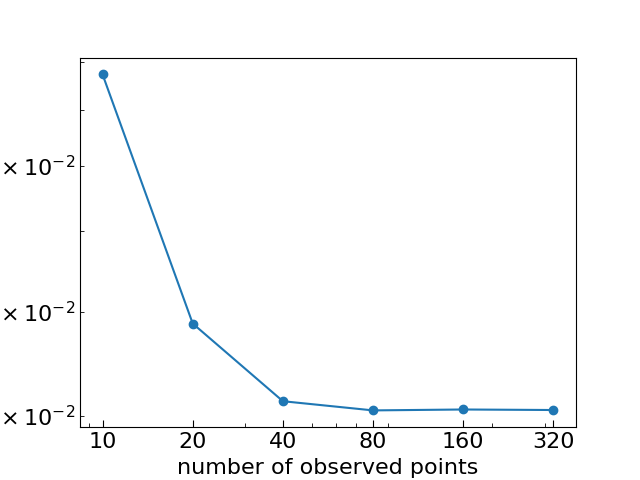}
    \caption{log-log plot of $L_2$ error of $V$}
  \end{subfigure}
  \caption{Numerical results for the MFG in \eqref{exp:smfg}. Log-log plots of errors v.s. number of observed data $I$, when $\nu=0.1$. }
  \label{fig:log-log plot1}
\end{figure}

\begin{table}[h!]
  \centering
  \caption{Recovered viscosity constants for the MFG in \eqref{exp:smfg} with $20$ points of observations (both $m$ and $V$ are observed), and $M=400$ collocation points. $\nu$ stands for the true viscosity and $\nu^\dagger$ represents the recovered viscosity. }
  \label{tab:recovered viscosity}
  \begin{tabular}{|c|c|c|c|c|c|c|c|c|}
    \hline
    $\nu$ & 1 & 0.5 & 0.1 & 0.05 & 0.01\\
    \hline
    $\nu^{\dagger}$& 1.000804 & 0.497811 & 0.099312 & 0.050905 & 0.015023 \\
    \hline
  \end{tabular}
\end{table}

\subsection{A Stationary MFG with an Unknown Local Coupling}
\label{sub:num:smfguk}

We investigate a stationary MFG with an unknown local coupling, formulated as:
\begin{align}
\label{exp:smfguk}
\begin{cases}
-\Delta u + H(x, \nabla u) = \Gamma(m) - V(x) + \overline H, & \forall x\in \mathbb{T}^2, \\
-\Delta m - \div (D_p H(x, \nabla u) m) = 0, & \forall x\in \mathbb{T}^2, \\
\int_{\mathbb{T}^2} u \, dx = 0, \int_{\mathbb{T}^2} m \, dx = 1, &
\end{cases}    
\end{align}
where \( H(x, p) = \frac{1}{2} |p|^2 \), \( \Gamma(m) = m^{\alpha} \), and the exponent \( \alpha \) is a positive constant. We identify \( \mathbb{T}^2\) with \( [-0.5, 0.5) \times [-0.5, 0.5) \), and choose the true potential \( V(x) = \cos(2\pi x) + \sin(2\pi y) + \sin(4\pi x) \). Our objective is to recover both \( V \) and \( \alpha \) based on noisy observations of the population density \( m \) and  \( V \).

As before, we use the GP method in \cite{mou2022numerical} to compute the numerical solutions for \eqref{exp:smfguk} for given $\alpha$ and $V$. These numerical solutions, in conjunction with  the explicit formulation of $V$ and $\alpha$ serve as reference solutions for accuracy comparisons. We randomly select \( M = 400 \) collocation points from \( \mathbb{T}^2 \), observing the values of \( m \) at the first \( I = 40 \) points and the values of \( V \) at the first \( I_V = 20 \) points as data. Independent Gaussian noises following \( \mathcal{N}(0, \gamma^2 I) \) with a standard deviation of \( \gamma = 10^{-3} \) are added to these observations. Periodic kernels are deployed to recover \( m \), \( u \), and \( V \), with lengthscale parameters being 0.6, 0.6, and 1, respectively. Additionally, adaptive diagonal regularization terms ('nuggets') with parameters \( \eta = 10^{-8} \) are incorporated into the covariance matrices for regularization purposes.

Figure \ref{fig:contour2} presents our experimental results in recovering \( m \), \( u \), and \( V \) with  \( \alpha = 2 \). The results show that we recover $u, m, V, \overline{H}$, and the coupling $\Gamma$ within reasonable accuracy with limited observational data. Table \ref{tab:recovered power} displays the recovered values of \( \alpha \) and \( \overline H \) under various ground truths, with \( I_V = 20 \) observed data points for \( V \), \( I = 40 \) observed data points for \( m\), and \( M = 400 \) samples.

\begin{figure}
  \centering
  \begin{subfigure}{0.24\textwidth} 
    \includegraphics[width=\linewidth]{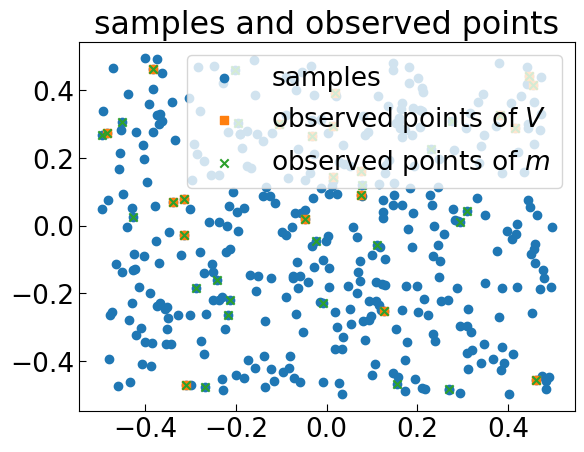}
    \caption{Samples \& observations}
  \end{subfigure}
  \hfill
  \begin{subfigure}{0.24\textwidth}
    \includegraphics[width=\linewidth]{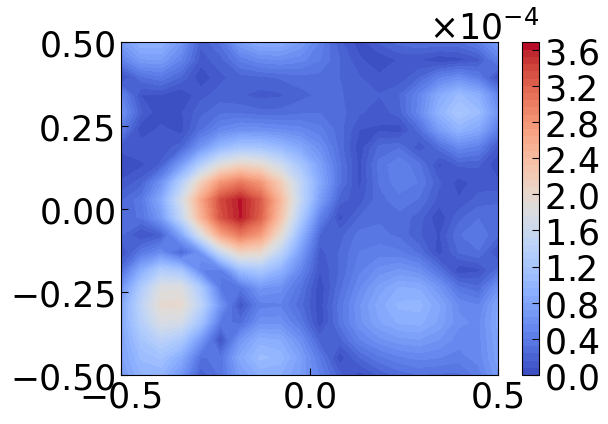}
    \caption{error contour of $m$}
  \end{subfigure}
  \hfill
  \begin{subfigure}{0.24\textwidth}
    \includegraphics[width=\linewidth]{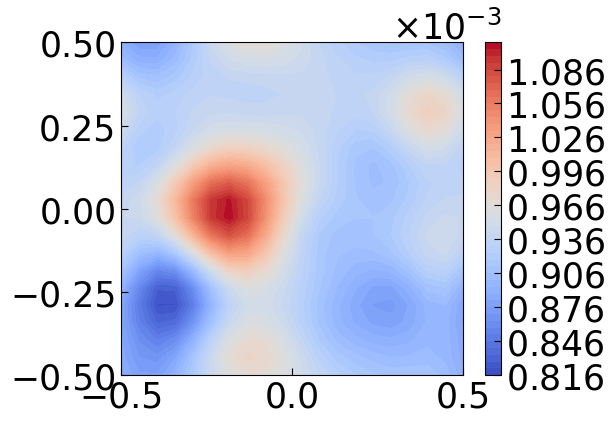}
    \caption{error contour of $u$}
  \end{subfigure}
  \hfill
  \begin{subfigure}{0.24\textwidth}
    \includegraphics[width=\linewidth]{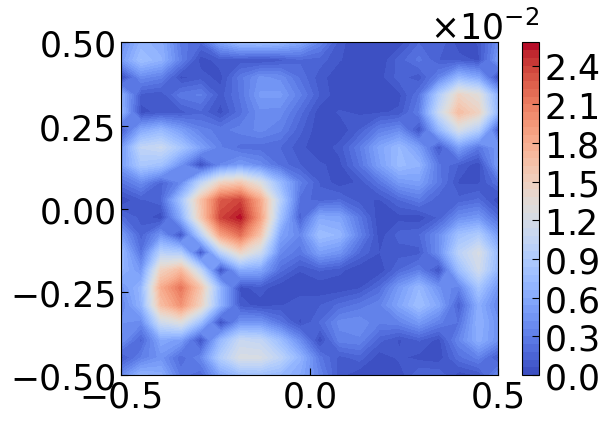}
    \caption{error contour of $V$}
  \end{subfigure}

  \vspace{1em} 
 \begin{subfigure}{0.24\textwidth}
    \includegraphics[width=\linewidth]{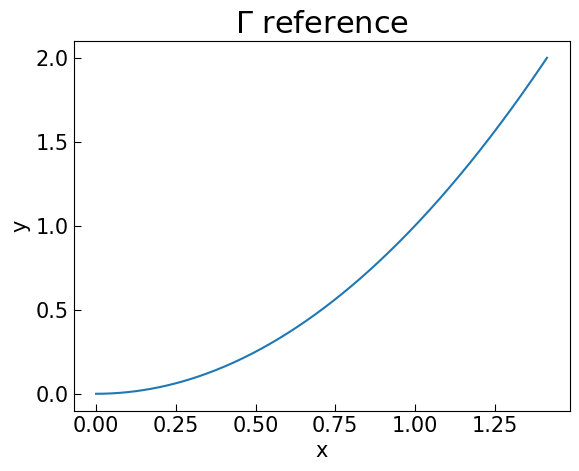}
    \caption{baseline for $\Gamma$}
  \end{subfigure}
  \hfill
  \begin{subfigure}{0.24\textwidth} 
    \includegraphics[width=\linewidth]{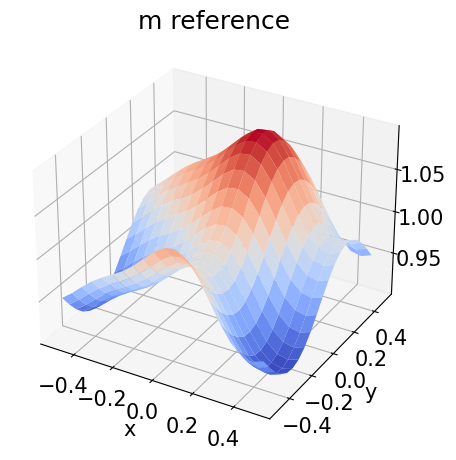}
    \caption{baseline for $m$}
  \end{subfigure}
  \hfill
  \begin{subfigure}{0.24\textwidth}
    \includegraphics[width=\linewidth]{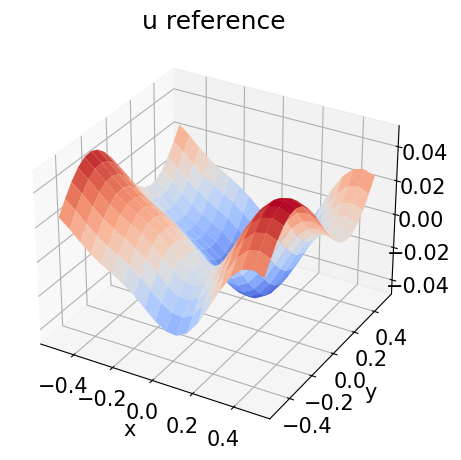}
    \caption{baseline for $u$}
  \end{subfigure}
  \hfill
  \begin{subfigure}{0.24\textwidth}
    \includegraphics[width=\linewidth]{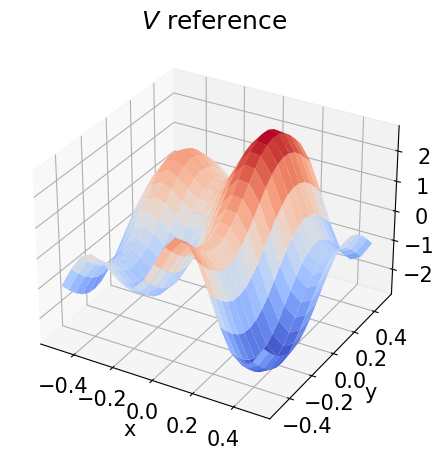}
    \caption{baseline for $V$}
  \end{subfigure}

  \vspace{1em} 
  \begin{subfigure}{0.24\textwidth}
    \includegraphics[width=\linewidth]{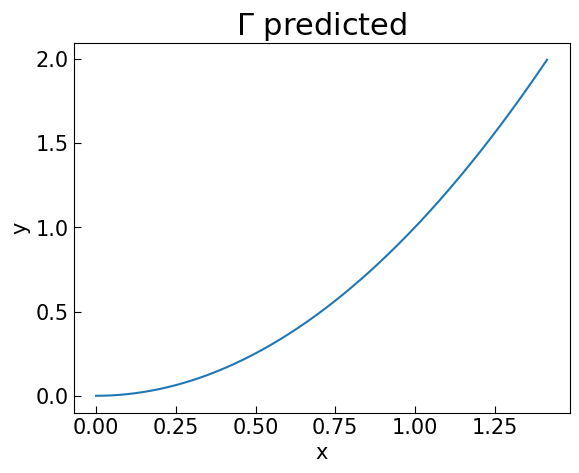}
    \caption{recovered $\Gamma$}
  \end{subfigure}
  \hfill
  \begin{subfigure}{0.24\textwidth} 
    \includegraphics[width=\linewidth]{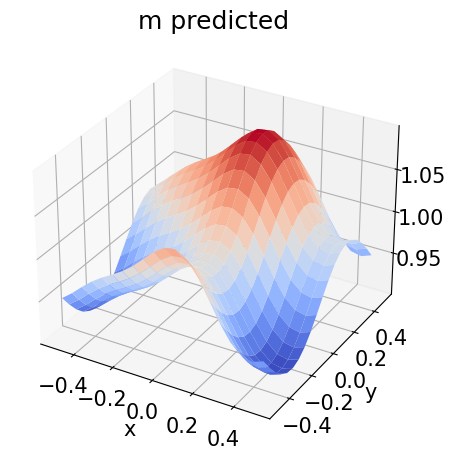}
    \caption{recovered $m$}
  \end{subfigure}
  \hfill
  \begin{subfigure}{0.24\textwidth}
    \includegraphics[width=\linewidth]{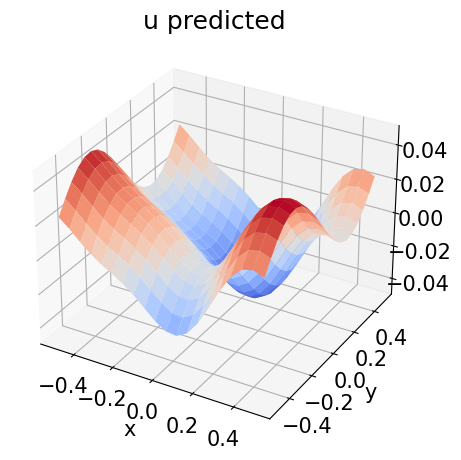}
    \caption{recovered $u$}
  \end{subfigure}
  \hfill
  \begin{subfigure}{0.24\textwidth}
    \includegraphics[width=\linewidth]{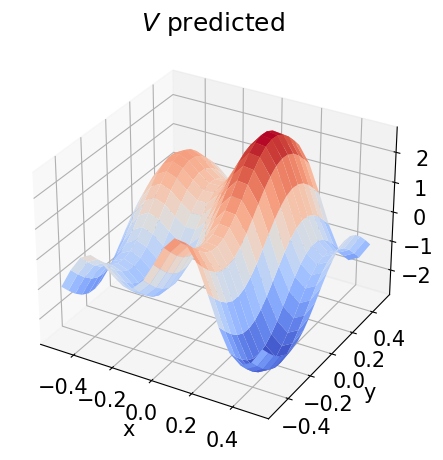}
    \caption{recovered $V$}
  \end{subfigure}

  \caption{Numerical results for the MFG in \eqref{exp:smfguk} with $\Gamma(m)=m^\alpha$.  $I=40$ observed data on $m$ and $I_V=20$ observed data on $V$, $M=400$ collocation points, ground truth $\alpha=2$, $\overline H= -0.975970$, recovered $\alpha^{\dagger}=1.990808$, $\overline H^{\dagger}=-0.966401$.}
  \label{fig:contour2}
\end{figure}

\begin{table}
  \centering
  \caption{Numerical results for the MFG in \eqref{exp:smfguk} with $\Gamma(m)=m^\alpha$. Recovered power $\alpha$ and $\overline H$ with $I_V=20$ data points of $V$, $I=40$ observations on $m$, and $M=400$ sample points. $\alpha$ and $\overline{H}$ stand for the references and $\alpha^\dagger$ and $\overline{H}^\dagger$ represent the recovered variables.  }
  \label{tab:recovered power}
  \begin{tabular}{|c|c|c|c|c|c|c|c|c|}
    \hline
    $\alpha$ & 0.5 & 1 & 2 & 3 & 5\\
    \hline
    $\alpha^{\dagger}$& 0.504769 & 1.182533 & 1.990808 & 2.89700 & 4.776276 \\
    \hline
     \rule{0pt}{2.3ex}   $\overline{H}^{}$ & -0.970467 & -0.971760& -0.975970 & -0.985608 &-0.998791 \\
    \hline
    $\overline H^{\dagger}$ &-0.970979 & -0.991179& -0.966402 &-0.985758 &-0.981125\\
    \hline
  \end{tabular}
\end{table}

\subsection{A Stationary MFG with an Unknown Non-local Coupling}
\label{sub:num:smfgnlc}
In this example, we study the stationary MFG \eqref{exp:smfguk} with a non-local coupling $\Gamma$ defined as 
\[
\Gamma(m) = \int_{\mathbb{T}^2} e^{-\frac{(x-y)^2}{2\sigma^2}}m(y) \, \dif y,
\]
where the lengthscale \(\sigma\) is unknown. We identify \(\mathbb{T}^2\) with \( [0, 1) \times [0, 1)\), and the true potential \(V(x) = \sin(2\pi y) + \sin(2\pi x) + \cos(4\pi x)\). Our objective is to recover \(m, u, V, \overline{H}\) and \(\sigma\) based on noisy observations of the population density \(m\) and coefficient \(V\).

We use the GP method in \cite{mou2022numerical} to compute the reference solutions given $\sigma$ and $V$. Then, we take $M=950$ sample points in domain, of which $900$ points are Gauss-Legendre quadrature points and $50$ points are randomly selected observation points for \(m\) and \(V\), i.e., \(I = I_V = 50\). Independent Gaussian noises following \(\mathcal{N}(0, \gamma^2 I)\) with a standard deviation of \(\gamma = 10^{-3}\) are added to these observations. The Gauss-Legendre quadrature rule is utilized to compute the convolution in the non-local coupling term. Periodic kernels are chosen for \(m\), \(u\), and \(V\) with lengthscale parameters $0.6$, $0.6$, and $1$, respectively. Additionally, adaptive diagonal regularization terms (``nuggets") with parameters \(\eta = 10^{-8}\) are incorporated into the covariance matrices.

Figure \ref{fig:contour3} presents our experimental results in recovering \(m\), \(u\), and \(V\) with \(\sigma = 0.5\). Table \ref{tab:recovered lengthscale} displays the recovered values of \(\sigma\) and \(\overline{H}\) under various ground truths, with \(I_V = I = 50\) observed \(m, V\) data points and a total of \(M = 950\) sample points. These results demonstrate the accuracy of our GP framework.

\begin{figure}
  \centering
  \begin{subfigure}{0.24\textwidth} 
    \includegraphics[width=\linewidth]{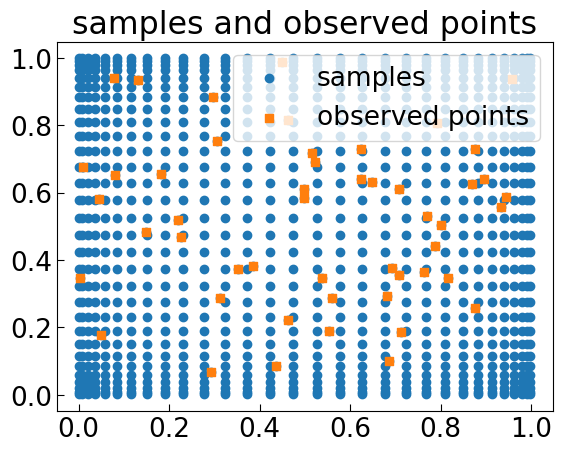}
    \caption{Samples \& observations}
  \end{subfigure}
  \hfill
  \begin{subfigure}{0.24\textwidth}
    \includegraphics[width=\linewidth]{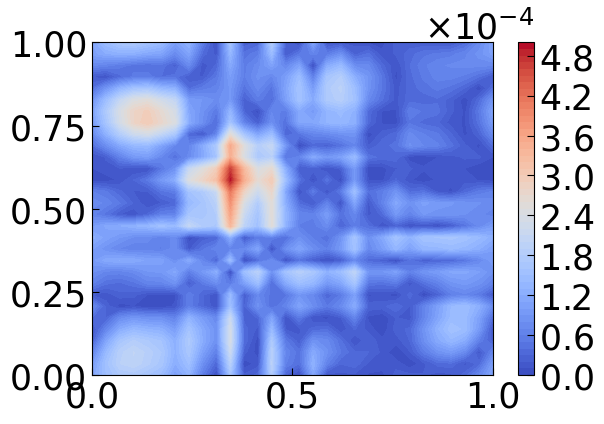}
    \caption{error contour of $m$}
  \end{subfigure}
  \hfill
  \begin{subfigure}{0.24\textwidth}
    \includegraphics[width=\linewidth]{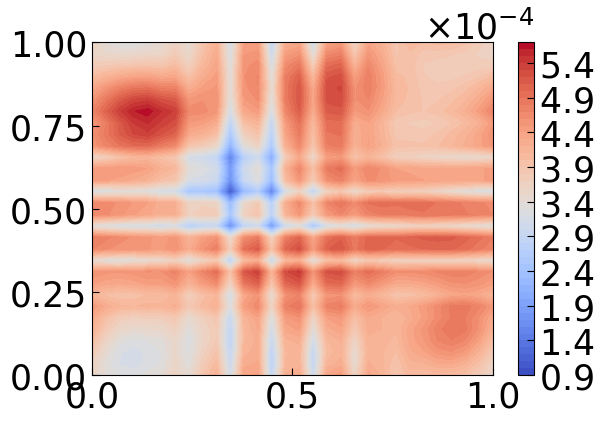}
    \caption{error contour of $u$}
  \end{subfigure}
  \hfill
  \begin{subfigure}{0.24\textwidth}
    \includegraphics[width=\linewidth]{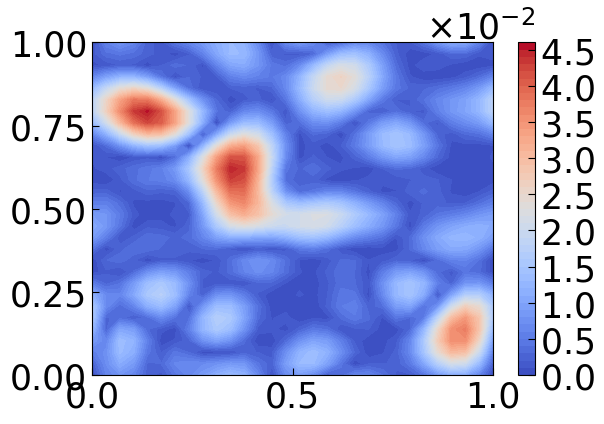}
    \caption{error contour of $V$}
  \end{subfigure}

  \vspace{1em} 
 \begin{subfigure}{0.24\textwidth}
    \includegraphics[width=\linewidth]{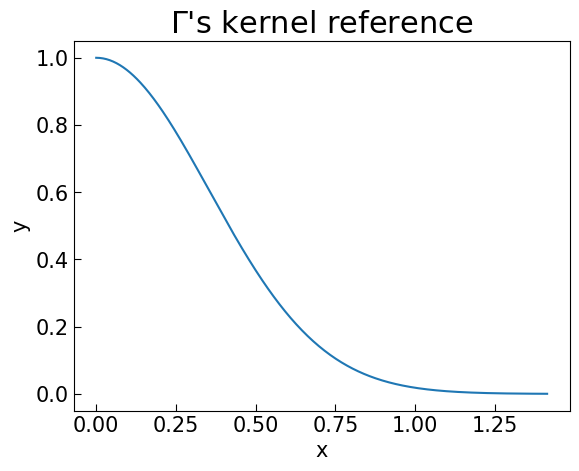}
    \caption{baseline for $\Gamma$}
  \end{subfigure}
  \hfill
  \begin{subfigure}{0.24\textwidth} 
    \includegraphics[width=\linewidth]{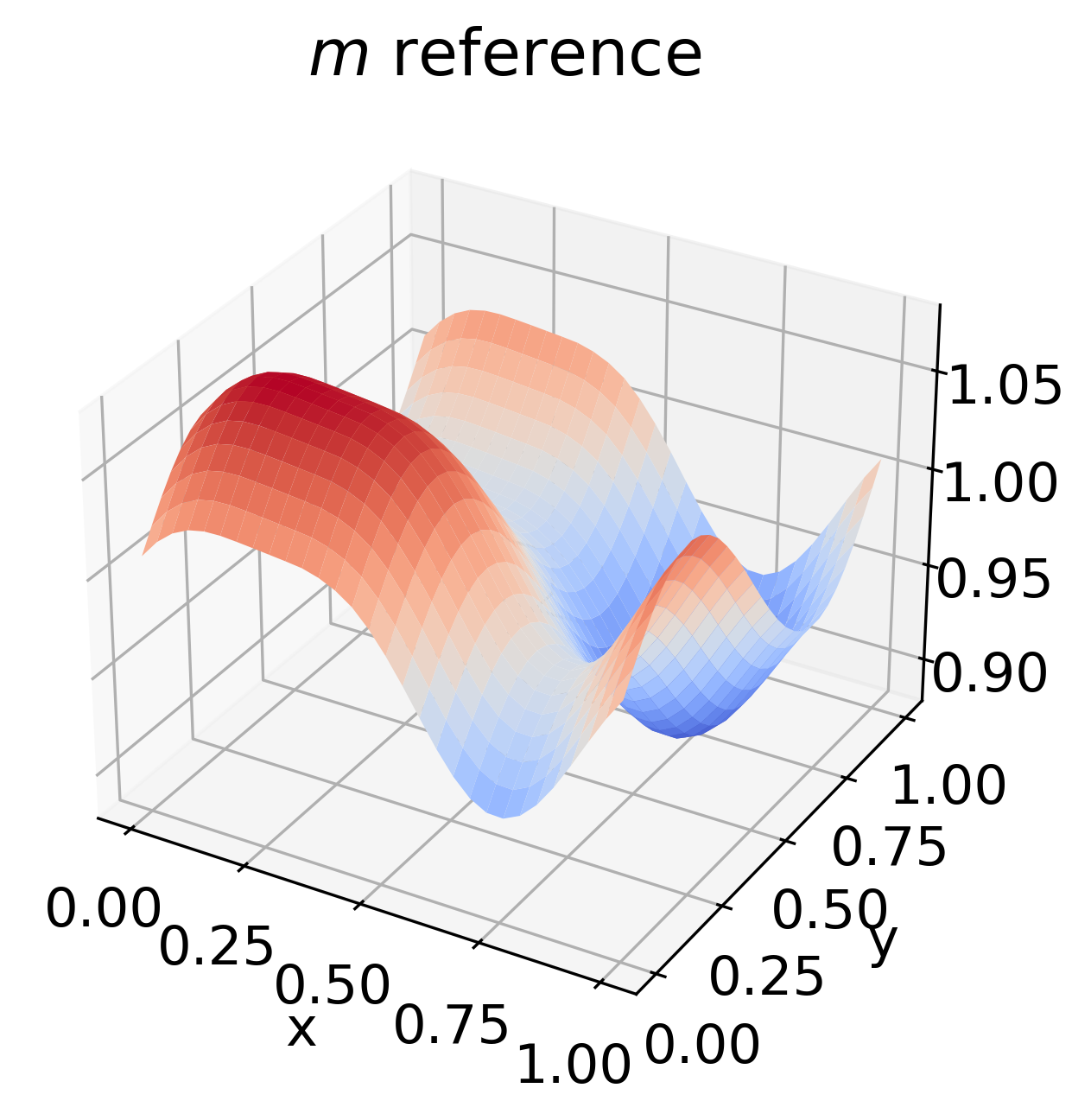}
    \caption{baseline for $m$}
  \end{subfigure}
  \hfill
  \begin{subfigure}{0.24\textwidth}
    \includegraphics[width=\linewidth]{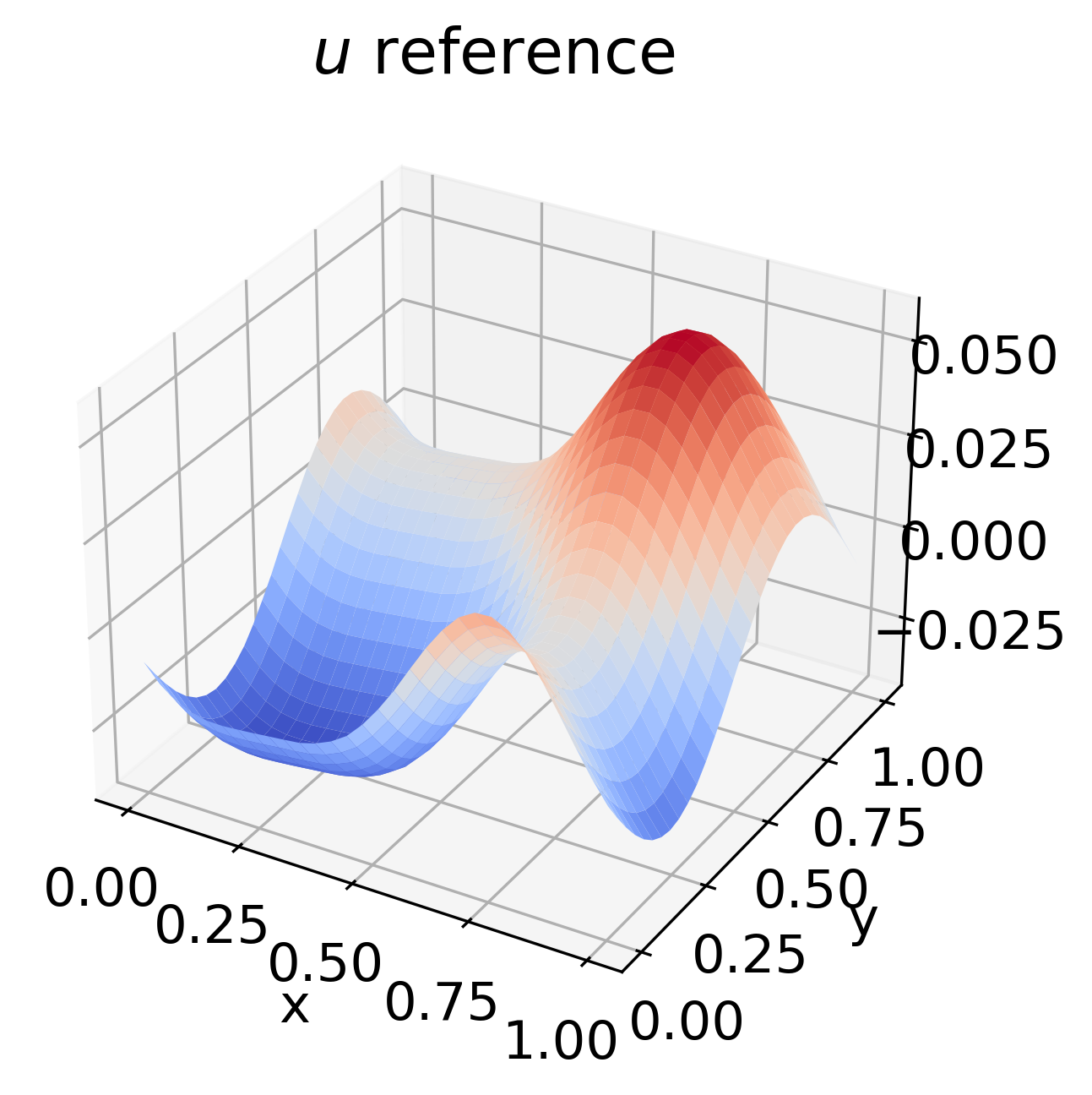}
    \caption{baseline for $u$}
  \end{subfigure}
  \hfill
  \begin{subfigure}{0.24\textwidth}
    \includegraphics[width=\linewidth]{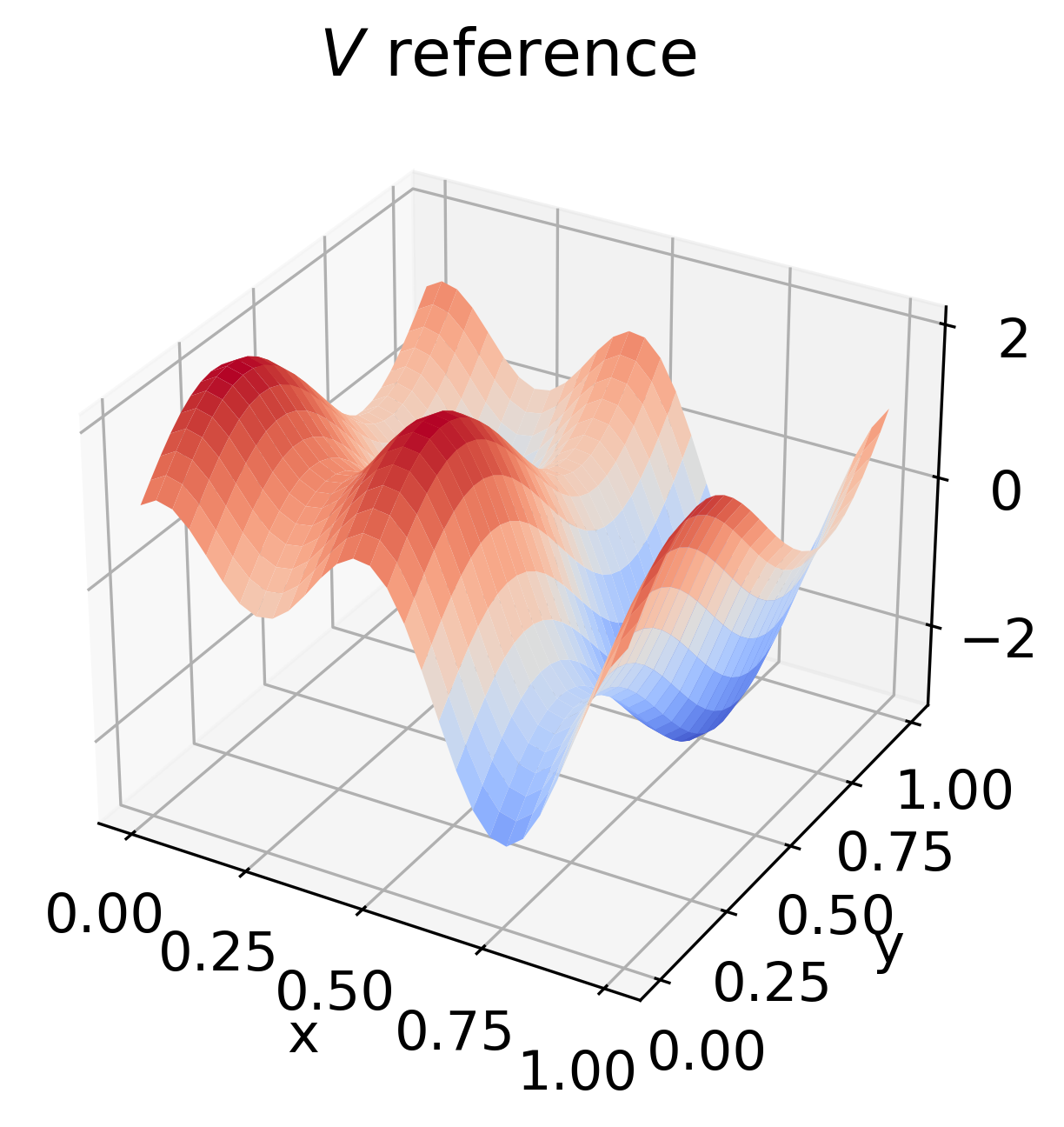}
    \caption{baseline for $V$}
  \end{subfigure}

  \vspace{1em} 
  \begin{subfigure}{0.24\textwidth}
    \includegraphics[width=\linewidth]{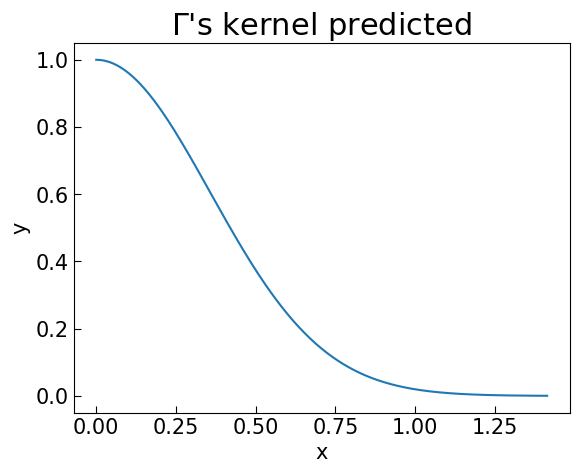}
    \caption{recovered $\Gamma$}
  \end{subfigure}
  \hfill
  \begin{subfigure}{0.24\textwidth} 
    \includegraphics[width=\linewidth]{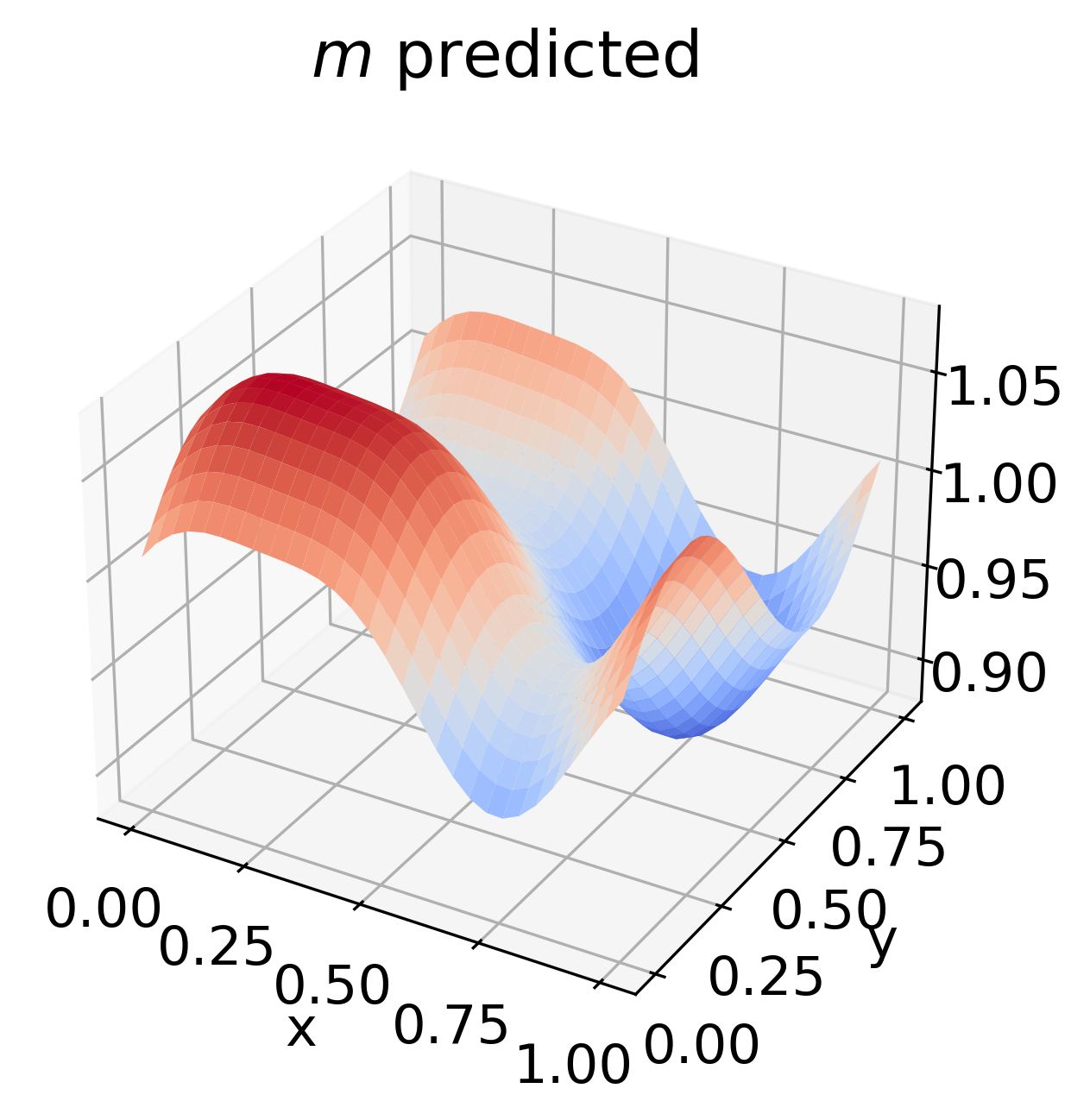}
    \caption{recovered $m$}
  \end{subfigure}
  \hfill
  \begin{subfigure}{0.24\textwidth}
    \includegraphics[width=\linewidth]{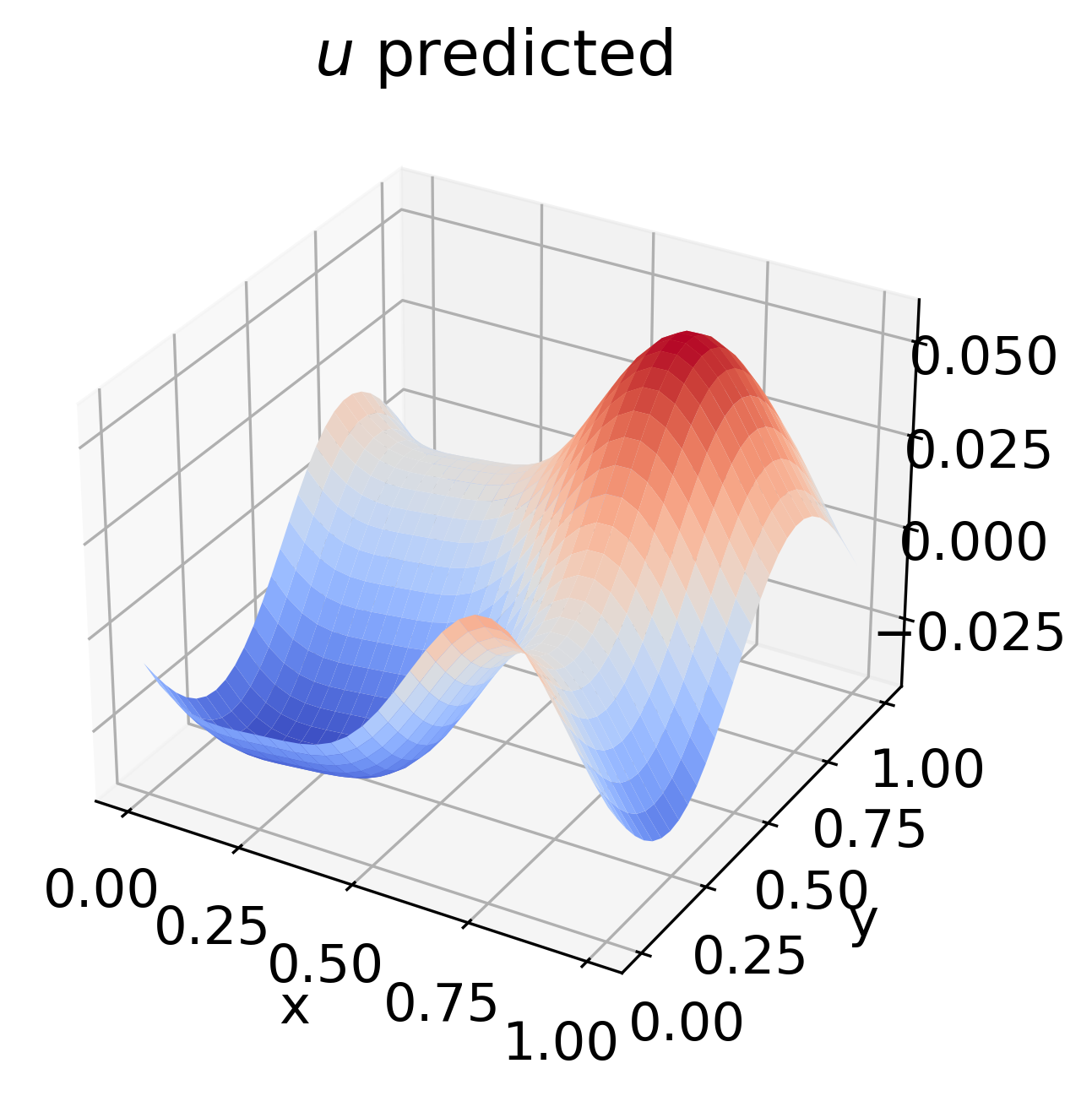}
    \caption{recovered $u$}
  \end{subfigure}
  \hfill
  \begin{subfigure}{0.24\textwidth}
    \includegraphics[width=\linewidth]{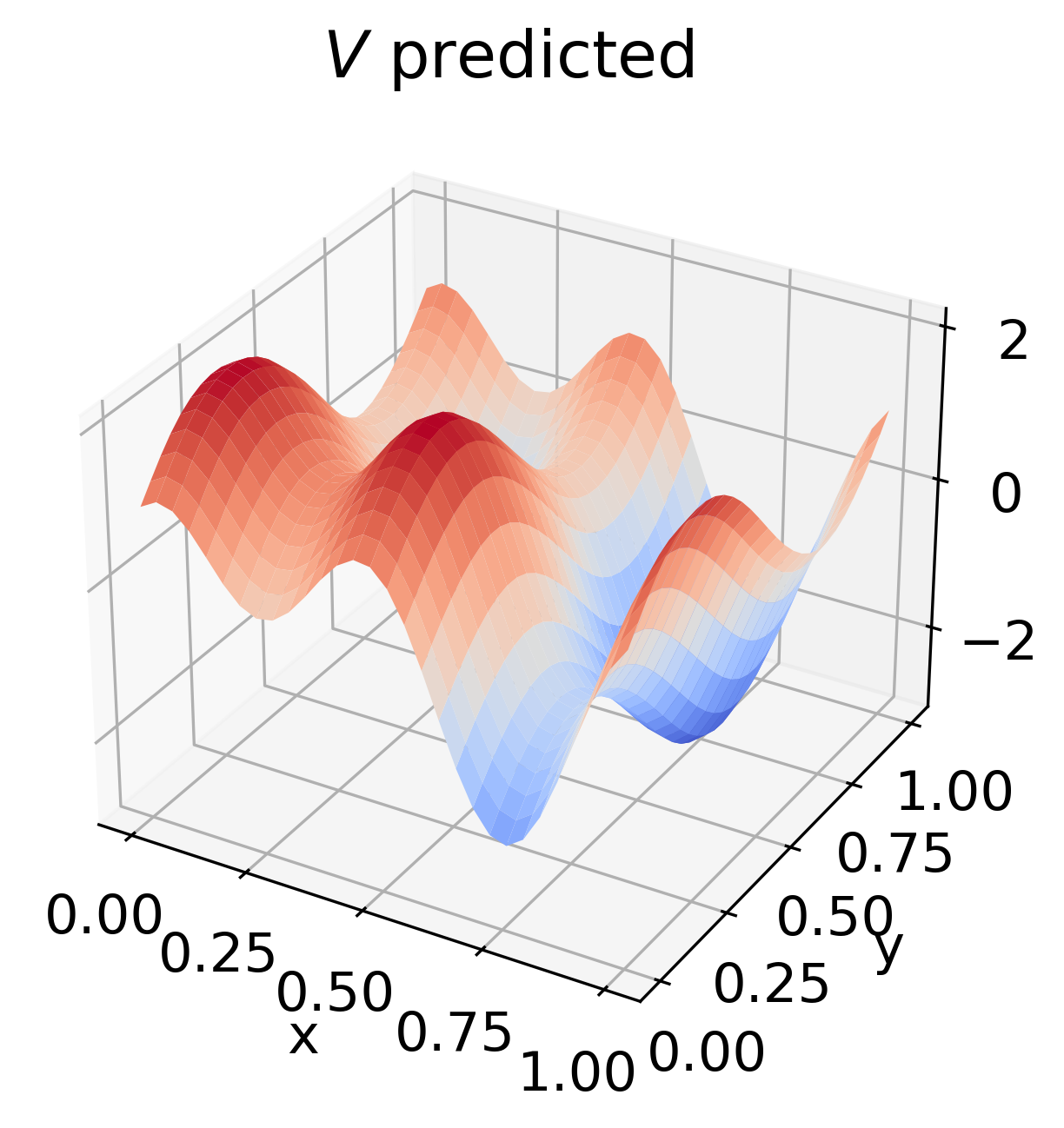}
    \caption{recovered $V$}
  \end{subfigure}

  \caption{Numerical results for the MFG in \eqref{exp:smfguk} with $\Gamma(m)=\int_{\mathbb{T}^2}e^{\frac{(x-y)^2}{2\sigma^2}}m(y)\dif y$. $50$ observation points for $m$ and $V$, $M=950$ sample points, ground truth $\sigma=0.5$, $\overline H=-0.545407$, recovered $\sigma^{\dagger}=0.504769$, $\overline H^{\dagger} = -0.626463$.}
  \label{fig:contour3}
\end{figure}

\begin{table}
  \centering
  \caption{Numerical results for the MFG in \eqref{exp:smfguk} with $\Gamma(m)=\int_{\mathbb{T}^2}e^{\frac{(x-y)^2}{2\sigma^2}}m(y)\dif y$. Recovered values of $\sigma$ and $\overline{H}$ with $50$ observation points for $m$ and $V$ and $M=950$ sample points. }
  \label{tab:recovered lengthscale}
  \begin{tabular}{|c|c|c|c|c|}
    \hline
    $\sigma$ & 0.4 & 0.5 & 1 & 2\\
    \hline
    $\sigma^{\dagger}$& 0.541382 & 0.504769 & 0.978857 & 1.921419 \\
    \hline
     \rule{0pt}{2.3ex}  $\overline{H}$&-0.431692& -0.545408 & -0.819267 & -0.926975 \\
     \hline
    $\overline{H^{\dagger}}$& -0.584162 &-0.626463 & -0.814232 & -0.921420 \\
    \hline
  \end{tabular}
\end{table}

\subsection{A Misspecified Recovery}
\label{sub:num:mspfcd}
In this example, we study a misspecification problem, where we recover the strategies and environment data generated by a MFG with a local coupling by specifying another MFG model with a non-local coupling. More precisely, we observe noisy data generated by the MFG system in \eqref{exp:smfguk} with  \( V(x) = \cos(2\pi x) + \cos(2\pi y) + \cos(4\pi x) \) and \( \Gamma(m) = m^2 \), but our goal is to recover \( m \), \( u \), \( V \), and \( \overline H \) by assuming that \( \Gamma(m) \) takes the form \( \Gamma(m) = \int_{\mathbb{T}^2} e^{-\frac{(x-y)^2}{2\sigma^2}}m(y) \, dy \), where the lengthscale \( \sigma \) is to be determined.

We use the GP method \cite{mou2022numerical} to solve \eqref{exp:smfguk} with $\Gamma(m)=m^2$ to get the reference solutions. Then, to solve the inverse problems, we take $625$ Gauss-Legendre quadrature points as samples and randomly select \( I = 30 \) observation points for \( m \). Independent Gaussian noises following \(\mathcal{N}(0, \gamma^2 I) \) with a standard deviation of \( \gamma = 10^{-3} \) are added to these observations. The Gauss-Legendre quadrature rule are employed to calculate the convolution in the non-local coupling. In the numerical experiments, periodic kernels are chosen for \( m \), \( u \), and \( V \) with lengthscale parameters set to $0.6$, $0.6$, and $1$, respectively. Additionally, adaptive diagonal regularization terms (``nuggets") with parameters \( \eta = 10^{-8} \) are added into the covariance matrices.

Figure \ref{fig:contour4} presents our experimental results in recovering \( m \), \( u \), and \( V \). It is important to highlight that in this scenario, there are no observations available for  $V$. The accuracy of the results imply the robustness of our GP framework.

\begin{figure}
  \centering
  \begin{subfigure}{0.24\textwidth} 
    \includegraphics[width=\linewidth]{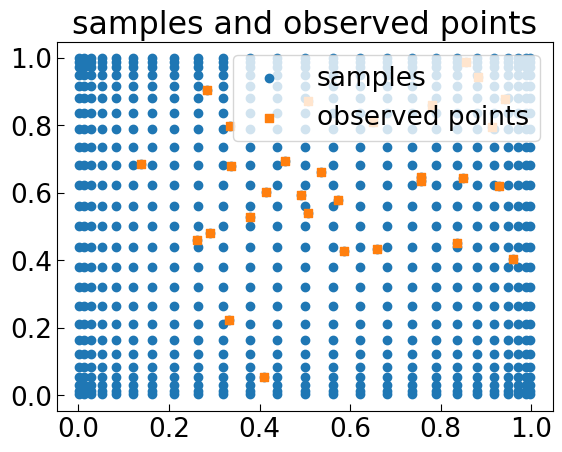}
    \caption{Samples \& observations}
  \end{subfigure}
  \hfill
  \begin{subfigure}{0.24\textwidth}
    \includegraphics[width=\linewidth]{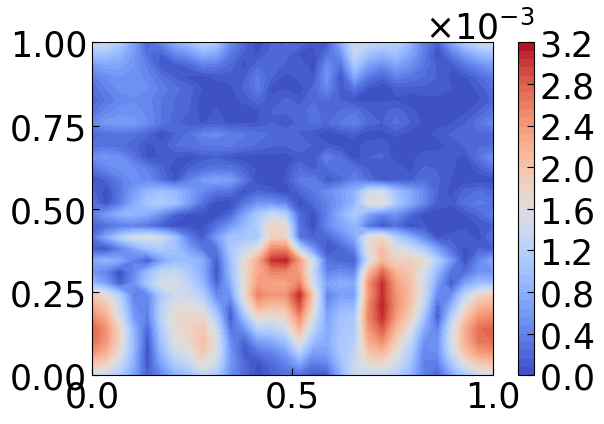}
    \caption{error contour of $m$}
  \end{subfigure}
  \hfill
  \begin{subfigure}{0.24\textwidth}
    \includegraphics[width=\linewidth]{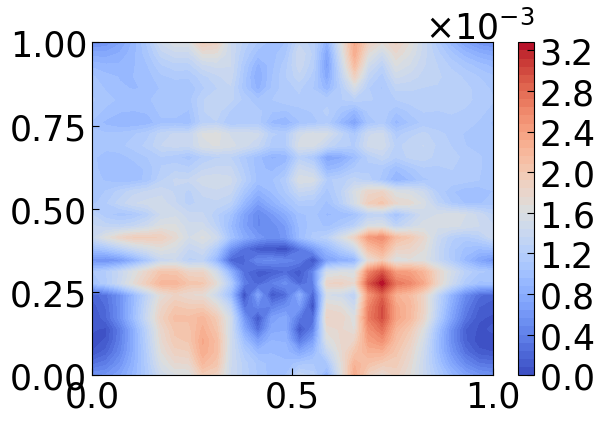}
    \caption{error contour of $u$}
  \end{subfigure}
  \hfill
  \begin{subfigure}{0.24\textwidth}
    \includegraphics[width=\linewidth]{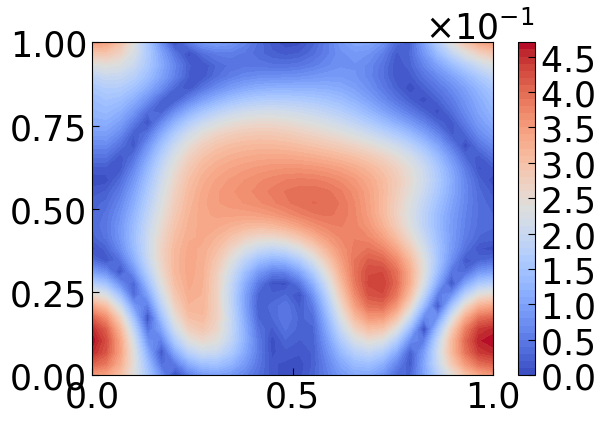}
    \caption{error contour of $V$}
  \end{subfigure}
  
  \vspace{1em} 
  
  \begin{subfigure}{0.24\textwidth} 
    \includegraphics[width=\linewidth]{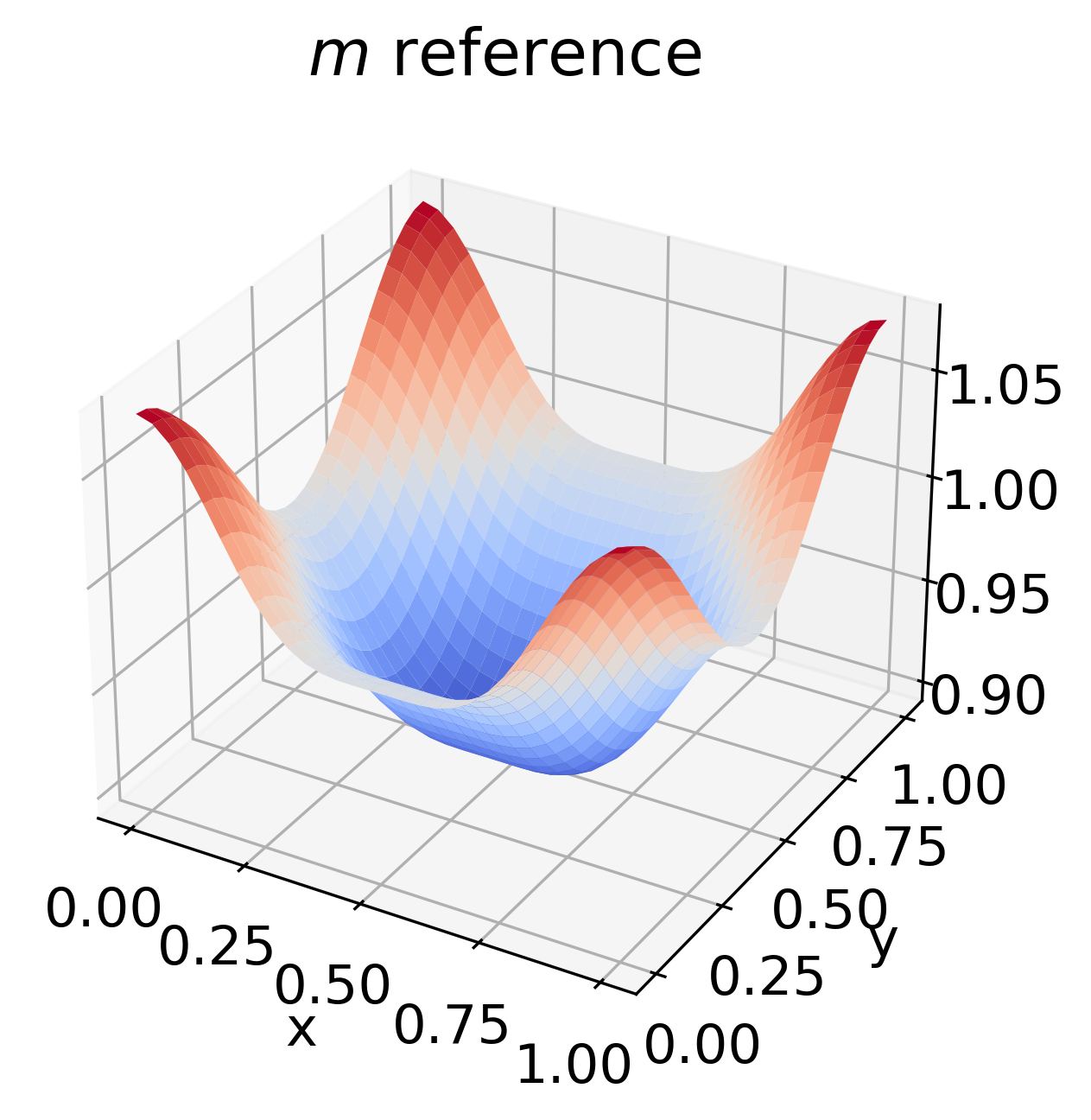}
    \caption{baseline for $m$}
  \end{subfigure}
  \hfill
  \begin{subfigure}{0.24\textwidth}
    \includegraphics[width=\linewidth]{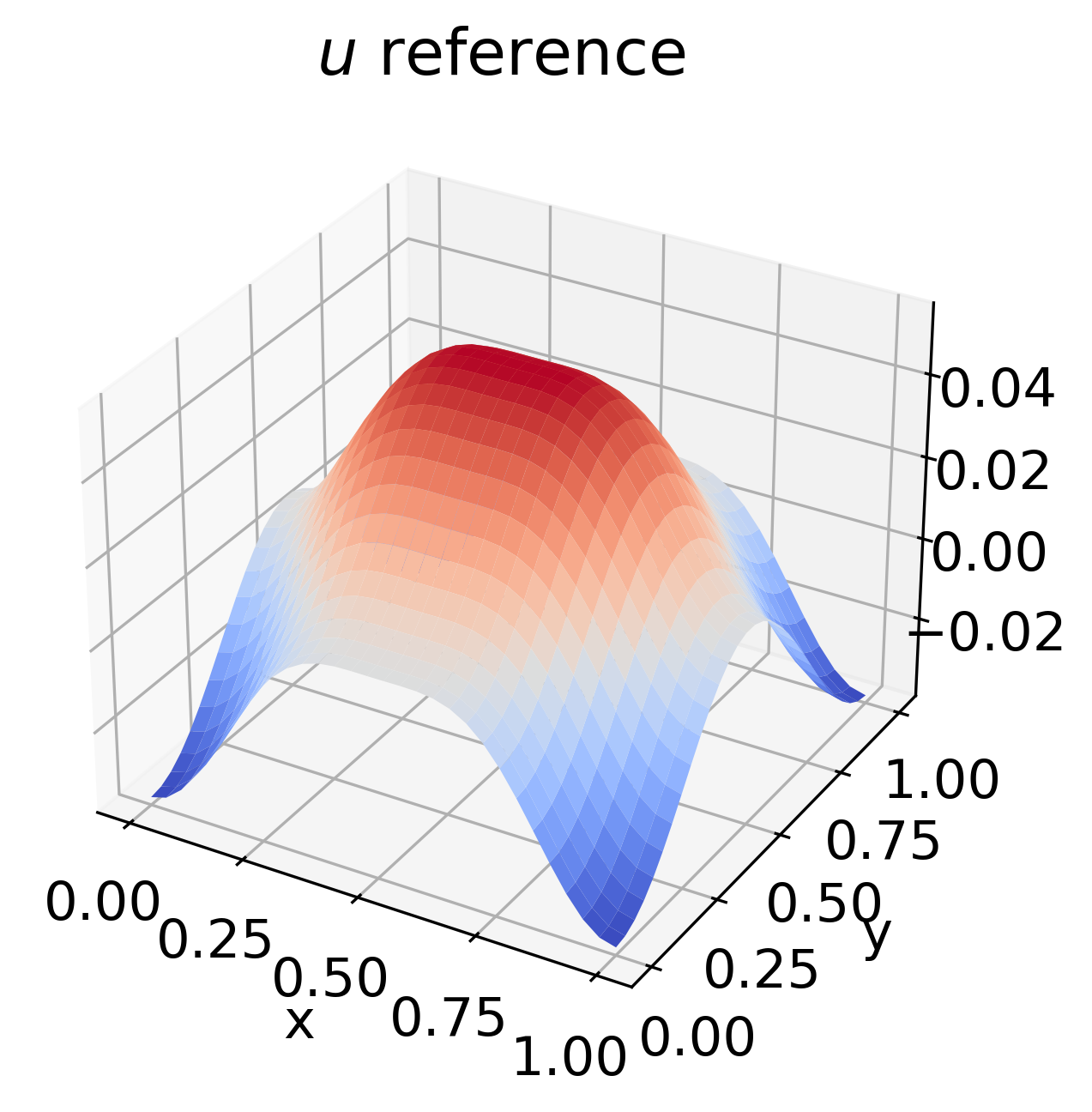}
    \caption{baseline for $u$}
  \end{subfigure}
  \hfill
  \begin{subfigure}{0.24\textwidth}
    \includegraphics[width=\linewidth]{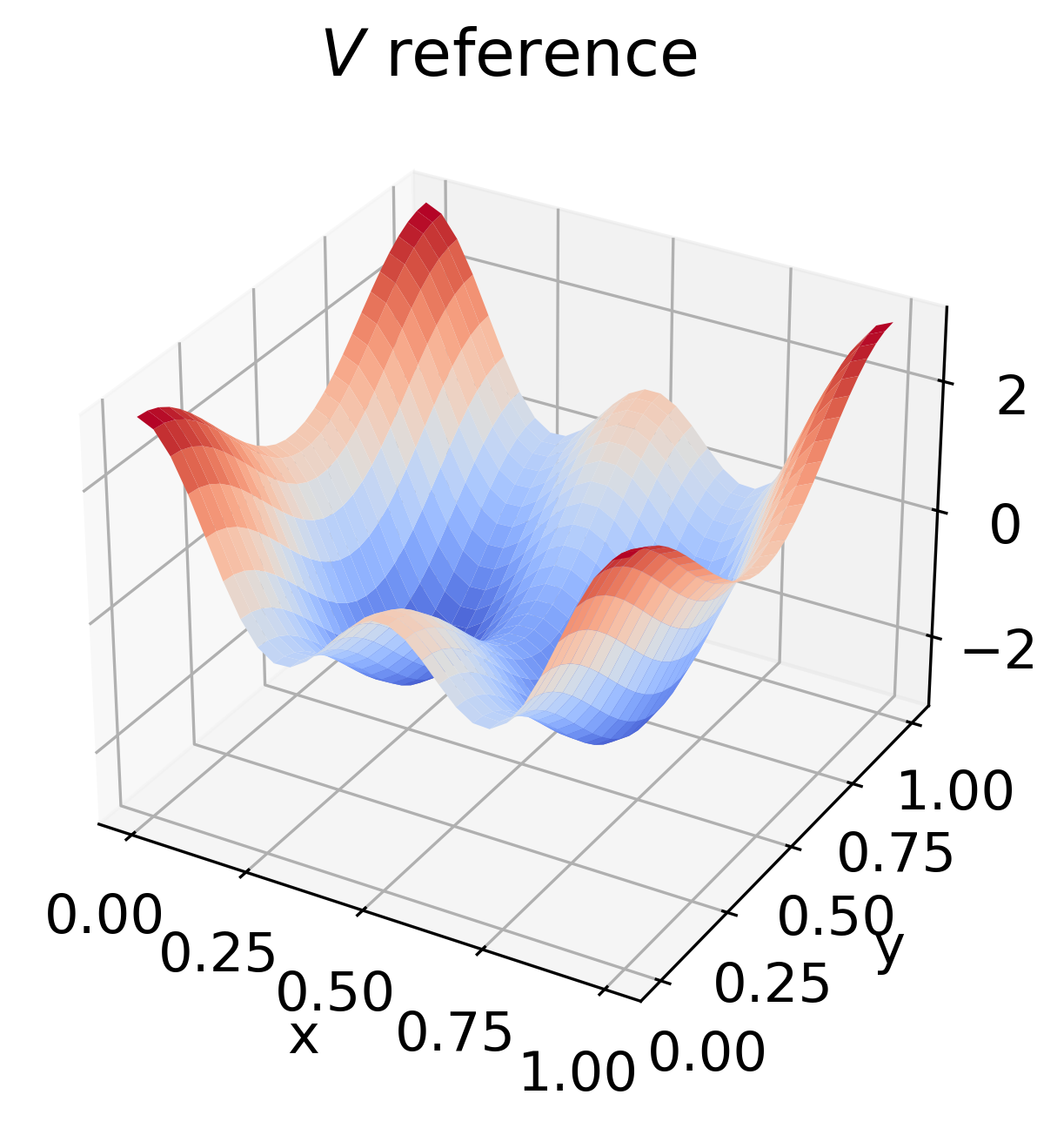}
    \caption{baseline for $V$}
  \end{subfigure}
  
  \vspace{1em} 
  
  \begin{subfigure}{0.24\textwidth} 
    \includegraphics[width=\linewidth]{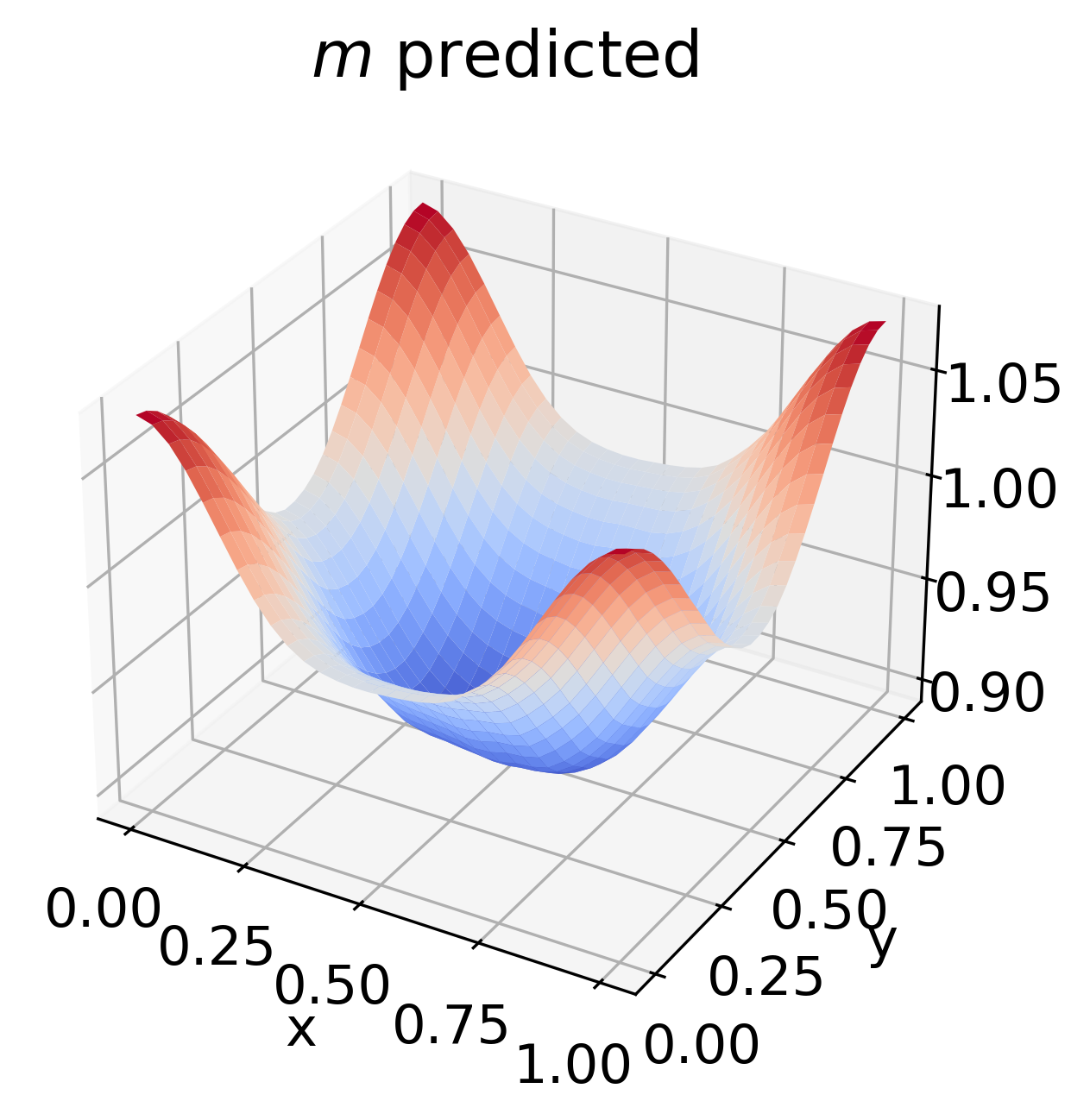}
    \caption{recovered $m$}
  \end{subfigure}
  \hfill
  \begin{subfigure}{0.24\textwidth}
    \includegraphics[width=\linewidth]{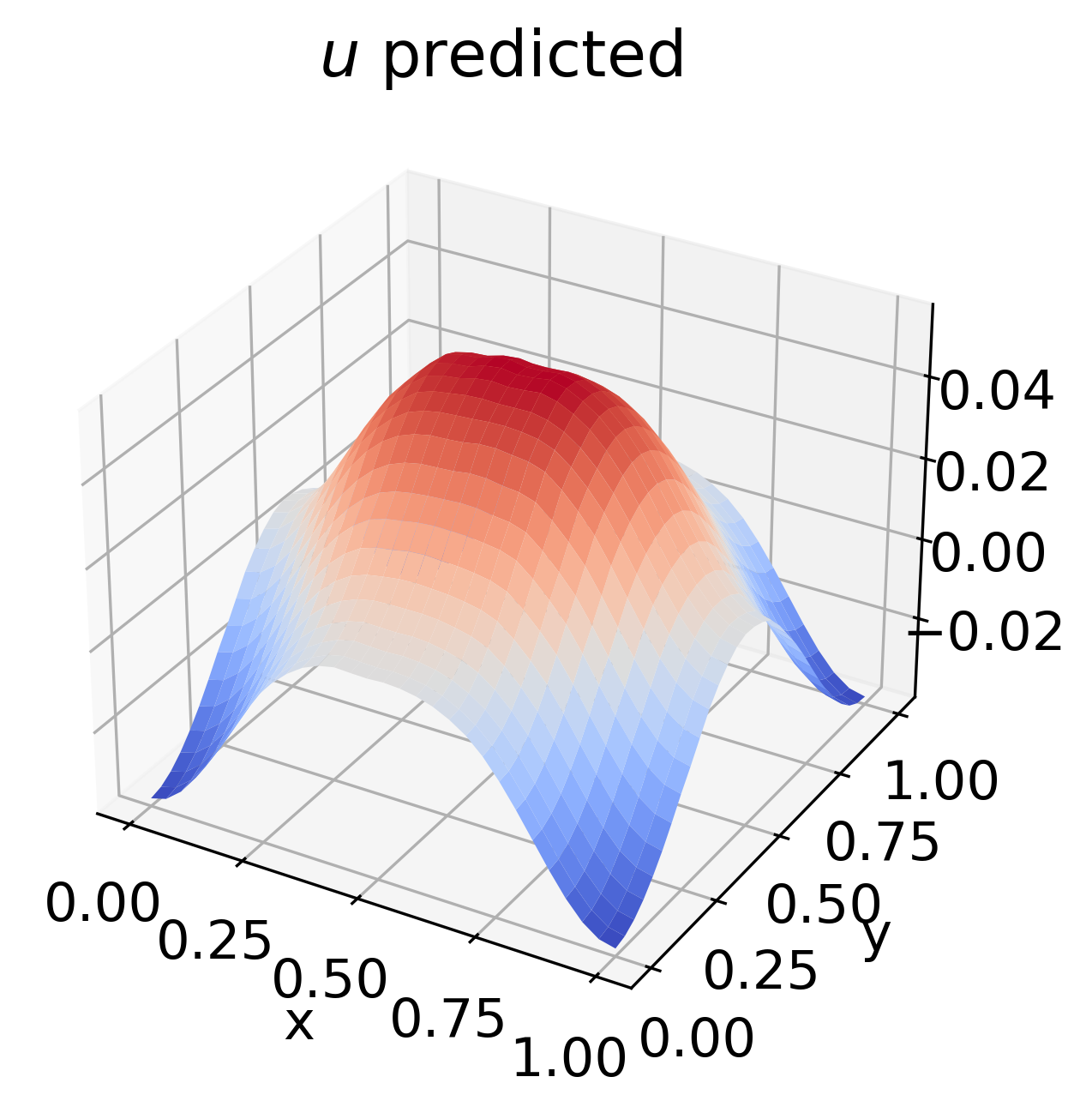}
    \caption{recovered $u$}
  \end{subfigure}
  \hfill
  \begin{subfigure}{0.24\textwidth}
    \includegraphics[width=\linewidth]{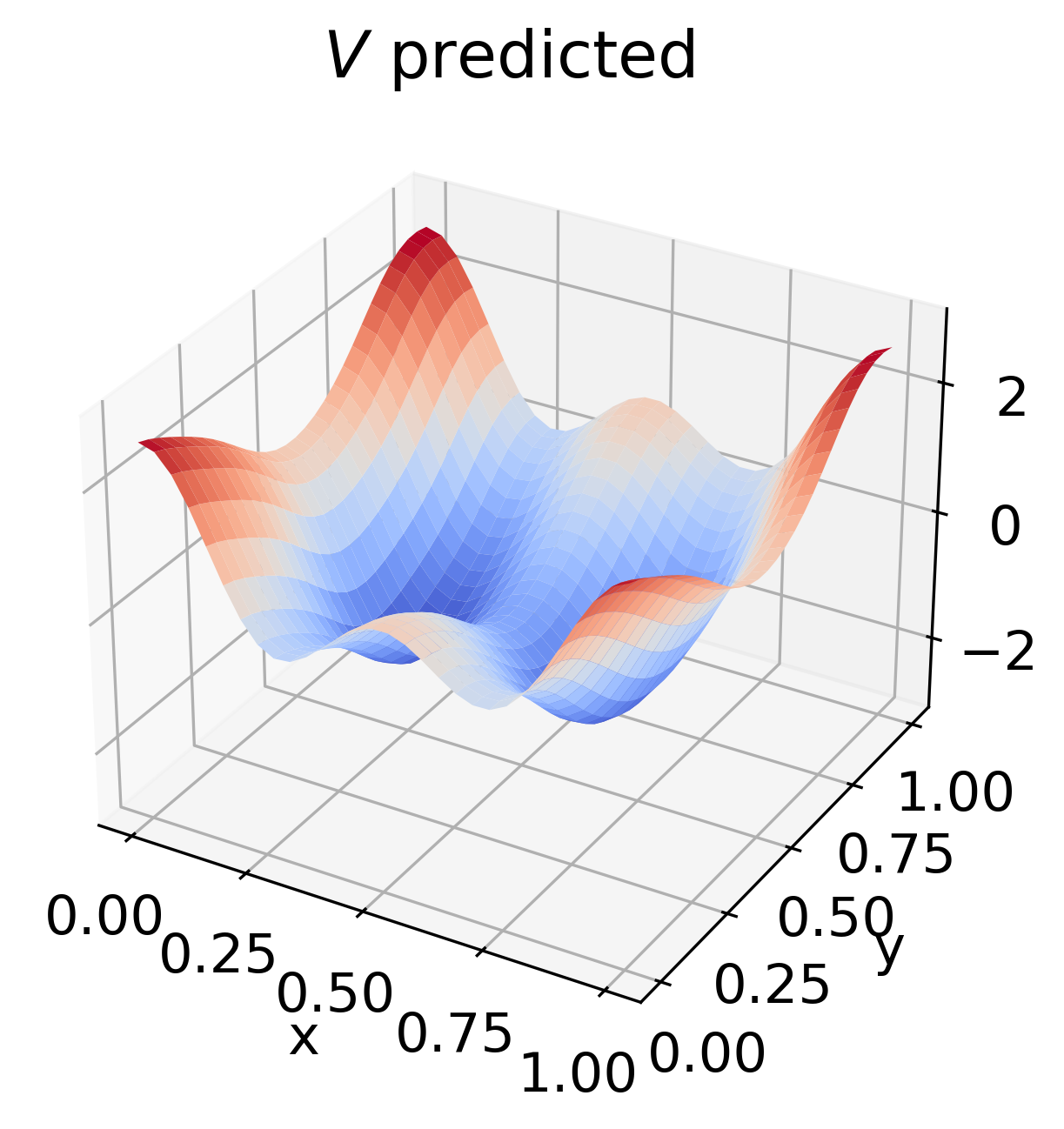}
    \caption{recovered $V$}
  \end{subfigure}
  
  \caption{Numerical results for the misspecification problem. We choose  $M=625$ samples and  $I = 30$ observations on $m$. The ground truth $\overline H=-0.920221$, and the recovered $\overline H^{\dagger}=-0.524876$, recovered $\sigma=0.619232$. }
  \label{fig:contour4}
\end{figure}

\subsection{A Time-dependent MFG}
\label{sub:num:tdmfg}
In this subsection, we investigate a time-dependent MFG with local coupling, represented by the following equations:
\begin{align*}
\begin{cases}
-\partial_t u - \nu \Delta u + H(x, \nabla u) = m^2 - V(x), & \forall (t, x)\in (0, T)\times \mathbb{T}, \\
\partial_t m - \nu \Delta m - \div (D_p H(x, \nabla u) m) = 0, & \forall (t, x)\in (0, T)\times \mathbb{T}, \\
m(0, x) = m_0(x), \quad u(T, x) = u_T(x), & \forall x\in \mathbb{T}, 
\end{cases}    
\end{align*}
where $T=1$, \( H(x, p) = \frac{1}{2} |p|^2 \). Let  \( \mathbb{T}\) be identified by \([-0.5, 0.5) \),  \( m_0(\cdot) = 1 \), and \( u_T(\cdot) = 0 \), and the true potential \( V(x) = 2\sin(2\pi x) + \cos(2\pi x) \). Our goal is to recover \(u, m, V\) in the whole time and space domain under noisy observations of the population density \( m \) and of the potential $V$ for a fixed viscosity \( \nu \).

Following the methodology outlined in Section \ref{subsec:tdgp}, we divide the entire time period into \( N_T = 20 \) intervals. At each interval, we randomly sample \( M = 50 \) collocation points from \( \mathbb{T} \) and observe the value of \( m \) at the first \( I_m = 5 \) points as data, along with observations of \( V \) at \( I_V = 3 \) points. Independent Gaussian noises following \( \mathcal{N}(0, \gamma^2 I) \) with a standard deviation of \( \gamma = 10^{-3} \) are added to these observations. Periodic kernels are used for \( m \), \( u \), and \( V \), with lengthscale parameters $1.5, 1.5$, and $1$, separately. Additionally, we set the nugget  parameters \( \eta = 10^{-8} \) to regularize  the covariance matrices.

Figure \ref{fig:contour5} presents our experimental results in recovering \( m \), \( u \), and \( V \) with a viscosity of \( \nu = 0.1 \). The recovery process demonstrated reasonable accuracy, effectively recover the quantities of interest  with limited observations.

\begin{figure}
  \centering
  \begin{subfigure}{0.27\textwidth}
    \includegraphics[width=\linewidth]{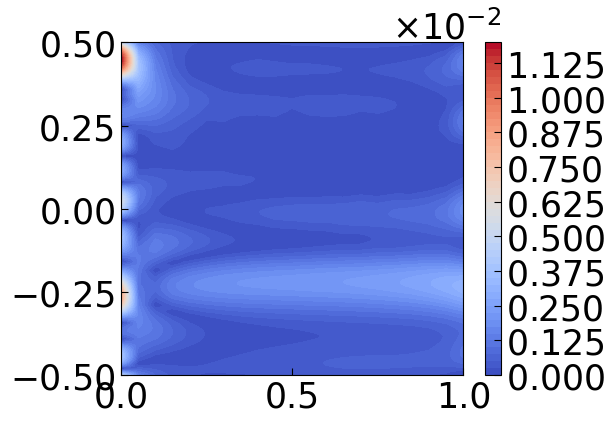}
    \caption{error contour of $m$}
  \end{subfigure}
  \hfill
  \begin{subfigure}{0.27\textwidth}
    \includegraphics[width=\linewidth]{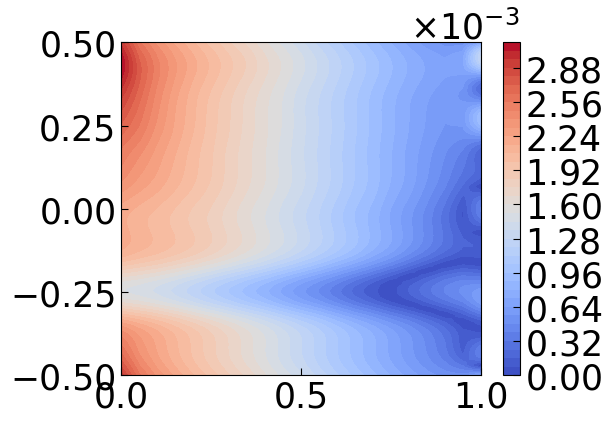}
    \caption{error contour of $u$}
  \end{subfigure}
  \hfill
  \begin{subfigure}{0.27\textwidth}
    \includegraphics[width=\linewidth]{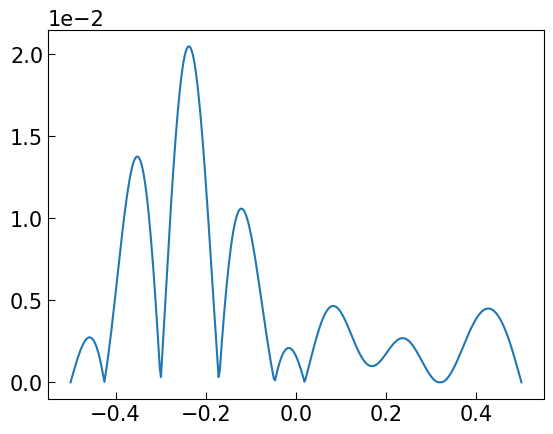}
    \caption{error of $V$}
  \end{subfigure}
  
  \vspace{1em} 
  
  \begin{subfigure}{0.27\textwidth} 
    \includegraphics[width=\linewidth]{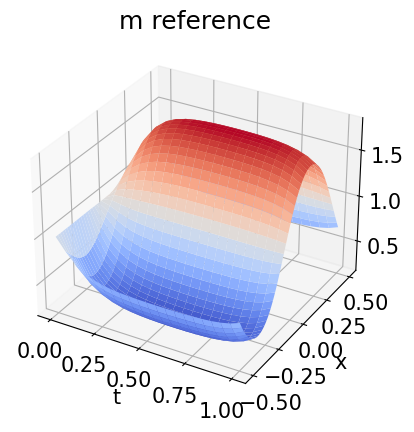}
    \caption{baseline for $m$}
  \end{subfigure}
  \hfill
  \begin{subfigure}{0.27\textwidth}
    \includegraphics[width=\linewidth]{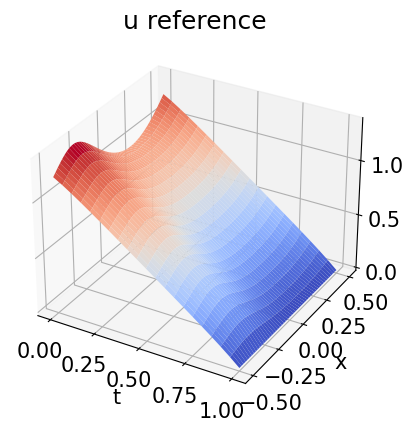}
    \caption{baseline for $u$}
  \end{subfigure}
  \hfill
  \begin{subfigure}{0.27\textwidth}
    \includegraphics[width=\linewidth]{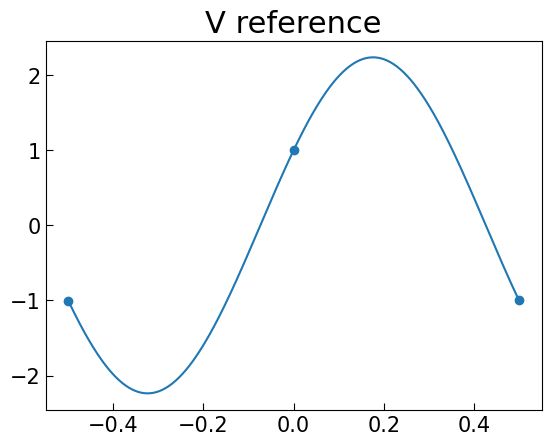}
    \caption{baseline for $V$, the $3$ scatter points are observations on $V$}
  \end{subfigure}
  
  \vspace{1em} 
  
  \begin{subfigure}{0.27\textwidth} 
    \includegraphics[width=\linewidth]{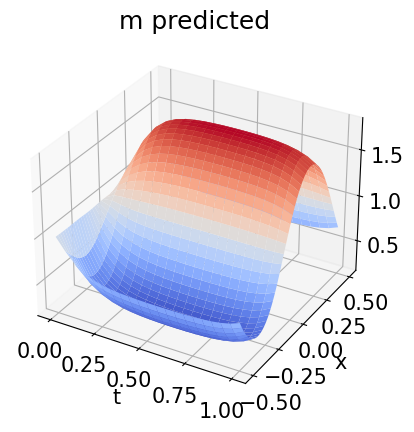}
    \caption{recovered $m$}
  \end{subfigure}
  \hfill
  \begin{subfigure}{0.27\textwidth}
    \includegraphics[width=\linewidth]{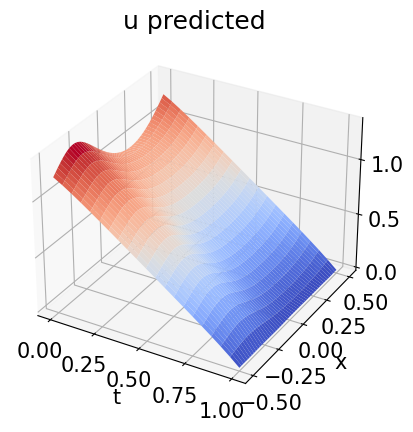}
    \caption{recovered $u$}
  \end{subfigure}
  \hfill
  \begin{subfigure}{0.27\textwidth}
    \includegraphics[width=\linewidth]{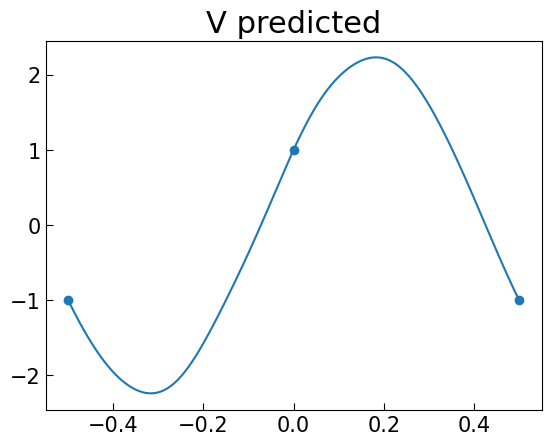}
    \caption{recovered $V$, the $3$ scatter points are observations on $V$}
  \end{subfigure}
  
  \caption{Numerical results for the time-dependent MFG. We set $\nu=0.1$. At each time slot, we select $M=50$ samples on $\mathbb{T}$. Within these samples, we choose $I_m=5$ observation points for $m$ and $I_V=3$ observation points for $V$.}
  \label{fig:contour5}
\end{figure}

\section{Conclusions and Future Works}
In this study, we introduce a GP framework designed to infer strategies and environmental configurations for agents engaged in MFGs. This framework operates under the challenge of working with partial and noisy data regarding the populations of agents and their environmental settings. To tackle this, we employ GPs to approximate the quantities of interest and address the inverse problems through Maximum A Posteriori (MAP) estimations conditioned on PDEs and constraints that are satisfied at specific collocation points and within the observed data.

Our method's effectiveness is showcased through a variety of examples, highlighting its capability to accurately recover diverse quantities in MFG contexts. An intriguing avenue for future work involves integrating approaches that address scalability issues, such as Random Fourier Features, sparse GPs, and mini-batch methods. This integration aims to enhance the scalability of our framework, making it more adept at handling large datasets.

Furthermore, another promising direction is the application of our methodology in fields like economics, biology, or finance. Here, the goal would be to develop MFG models in scenarios where no established models matching the data currently exist, using our approach to interpret and manage the data effectively.

\label{secDisc}

\begin{acknowledgement*}
CM is supported by CityU Start-up Grant 7200684, Hong Kong RGC Grant ECS 21302521, Hong Kong RGC Grant GRF 11311422 and Hong Kong RGC Grant GRF 11303223. XY acknowledges support from the Air Force Office of Scientific Research under MURI award number FA9550-20-1-0358 (Machine Learning and Physics-Based Modeling and Simulation). CZ is supported by Singapore MOE (Ministry of Education) AcRF Grant A-8000453-00-00, IoTex Foundation Industry Grant A-8001180-00-00 and NSFC Grant No. 11871364. 
\end{acknowledgement*}

\bibliographystyle{plain}

\end{document}